\documentclass[a4paper,twocolumn,accepted=2022-09-20]{quantumarticle}
\pdfoutput=1
\usepackage{rotating}
\usepackage[colorlinks,breaklinks,linkcolor=magenta,citecolor=blue,urlcolor=blue]{hyperref}
\usepackage{amsmath,amssymb,amsthm,bm,mathtools,amsfonts,mathrsfs,bbm,dsfont}
\usepackage{physics}
\usepackage{graphicx}
\usepackage{color}
\usepackage{enumitem}

\newcommand{\id}{\ensuremath{\mathds{1}}}
\renewcommand{\vec}[1]{\boldsymbol{#1}}

\newtheorem{thm}{Theorem}
\newtheorem{cor}[thm]{Corollary}
\newtheorem{obs}[thm]{Observation}
\newtheorem{lem}[thm]{Lemma}

\newtheorem{prem}[thm]{Condition}
\newtheorem{defi}[thm]{Definition}

\newtheorem{manuallemmainner}{Lemma}
\newenvironment{manuallemma}[1]{%
  \renewcommand\themanuallemmainner{#1}%
  \manuallemmainner
}{\endmanuallemmainner}

\newtheorem{manualtheoreminner}{Theorem}
\newenvironment{manualtheorem}[1]{%
  \renewcommand\themanualtheoreminner{#1}%
  \manualtheoreminner
}{\endmanualtheoreminner}

\newtheorem{manualobservationinner}{Observation}
\newenvironment{manualobservation}[1]{%
  \renewcommand\themanualobservationinner{#1}%
  \manualobservationinner
}{\endmanualobservationinner}

\begin{document}
\nonfrenchspacing

\title{Transformations of Stabilizer States in Quantum Networks}

\author{Matthias Englbrecht}
\affiliation{Institute for Theoretical Physics, University of Innsbruck, Technikerstraße 21A, 6020 Innsbruck, Austria}

\author{Tristan Kraft}
\affiliation{Institute for Theoretical Physics, University of Innsbruck, Technikerstraße 21A, 6020 Innsbruck, Austria}

\author{Barbara Kraus}
\affiliation{Institute for Theoretical Physics, University of Innsbruck, Technikerstraße 21A, 6020 Innsbruck, Austria}


\begin{abstract}
Stabilizer states and graph states find application in quantum error correction, measurement-based quantum computation and various other concepts in quantum information theory. In this work, we study party-local Clifford (PLC) transformations among stabilizer states. These transformations arise as a physically motivated extension of local operations in quantum networks with access to bipartite entanglement between some of the nodes of the network. First, we show that PLC transformations among graph states are equivalent to a generalization of the well-known local complementation, which describes local Clifford transformations among graph states. Then, we introduce a mathematical framework to study PLC equivalence of stabilizer states, relating it to the classification of tuples of bilinear forms. This framework allows us to study decompositions of stabilizer states into tensor products of indecomposable ones, that is, decompositions into states from the entanglement generating set (EGS). While the EGS is finite up to $3$ parties [Bravyi et al., J. Math. Phys. {\bf 47}, 062106~(2006)], we show that for $4$ and more parties it is an infinite set, even when considering party-local unitary transformations. Moreover, we explicitly compute the EGS for $4$ parties up to $10$ qubits. Finally, we generalize the framework to qudit stabilizer states in prime dimensions not equal to $2$, which allows us to show that the decomposition of qudit stabilizer states into states from the EGS is unique.
\end{abstract}

\maketitle

\section{Introduction}
The pursuit of building large-scale quantum networks to implement a quantum internet~\cite{Ki08,We18} has recently attracted much interest from the experimental and the theoretical perspective~\cite{Ci97,Mu16,Pe13}. In particular, several important building blocks of quantum networks have been demonstrated experimentally~\cite{Du10,Re15}. Such networks open interesting possibilities for advanced quantum information processing tasks, which include distributed quantum computing~\cite{Ci99}, quantum metrology~\cite{Pr18}, and, most prominently, long-distance quantum communication~\cite{Du01}. Beyond future applications, the strong correlations in quantum networks also open new possibilities for fundamental tests of quantum mechanics (cf. Ref.~\cite{Ta21}). As technology progresses it becomes more and more relevant to understand quantum networks from the theoretical point of view. The fundamental building block of these networks are entangled quantum systems. Entanglement itself is quite well understood in the bipartite case, but understanding multipartite entanglement, despite many efforts, still remains to be an extremely challenging task~\cite{Gu09,Ho09,Be06,Sa18}. In particular, our understanding of network entanglement is still quite limited, even in the case of a few qubits only, which is due to its complicated structure. Indeed, recent results demonstrate that its characterization requires new mathematical tools to be explored~\cite{Kr21,Na20,Ha22}.

In this regard, so-called \emph{stabilizer states} constitute an interesting subclass of multipartite entangled pure states. These states are defined as the unique simultaneous eigenstate of a maximal set of commuting Pauli observables, the so-called \emph{stabilizer}~\cite{He04,He06}. The stabilizer formalism was originally introduced by Gottesman in the realm of quantum error correction codes~\cite{Go97}. Nowadays, stabilizer states are known to be important to many aspects of quantum information processing such as, e.g., measurement based quantum computing~\cite{Ra01} or the classical simulation of quantum circuits~\cite{Go98}. \emph{Graph states}~\cite{He04,He06} admit a simple description in terms of mathematical graphs~\footnote{In the context of quantum error correction graph states were first considered in Ref.~\cite{Schl01}.}. It is known that any stabilizer state is equivalent to a graph state under a certain class of local operations, the so-called local \emph{Clifford operations}, which map the set of stabilizer states onto itself. Due to their relevance, as outlined above, gaining a deeper understanding of entanglement in stabilizer states is an important goal.

In this paper we study stabilizer states in quantum networks. Creating bipartite entanglement in quantum networks is considerably less demanding than the distribution of large entangled states. In particular, well connected nodes can share a large amount of bipartite entanglement at basically no cost. This bipartite entanglement can be used, for instance, to implement non-local transformations of stabilizer states via gate teleportation~\cite{Go99}. Here, Clifford operations turn out to be particularly important as they can be implemented deterministically. Thus, well connected nodes having access to the additional bipartite entanglement can be viewed as a single party. We call the operations that can be applied on those parties \emph{party-local Clifford} (PLC) operations. In case each party only holds a single qubit it is known that two stabilizer states are equivalent under local Cliffords if and only if their corresponding graph states are connected by a series of \emph{local complementations}~\cite{Vn04}.

If all qubits of a stabilizer state are grouped into two parties the bipartite stabilizer state can be converted into multiple independent copies of the maximally entangled state $\ket{\phi^+}$, and locally separable states $\ket{0}$ using PLC transformations~\cite{Fa04}. For tripartite stabilizer states it was shown in Ref.~\cite{Br06} that any state can be converted to multiple independent copies of the Greenberger-Horne-Zeilinger (GHZ) state $\ket{GHZ}$ between the three parties, the Bell state $\ket{\phi^+}$ between pairs of parties, and locally separable states $\ket{0}$, see Figure~\ref{fig:fig1}. The fact that these decompositions are unique, even under more general party local unitary (PLU) transformations, motivated the authors of Ref.~\cite{Br06} to call these sets \emph{entanglement generating sets}. These results were later generalized to tripartite qu{\it d}it graph states in Ref.~\cite{Gr11}.

\begin{figure}
	\centering
	\includegraphics[width=0.7\columnwidth]{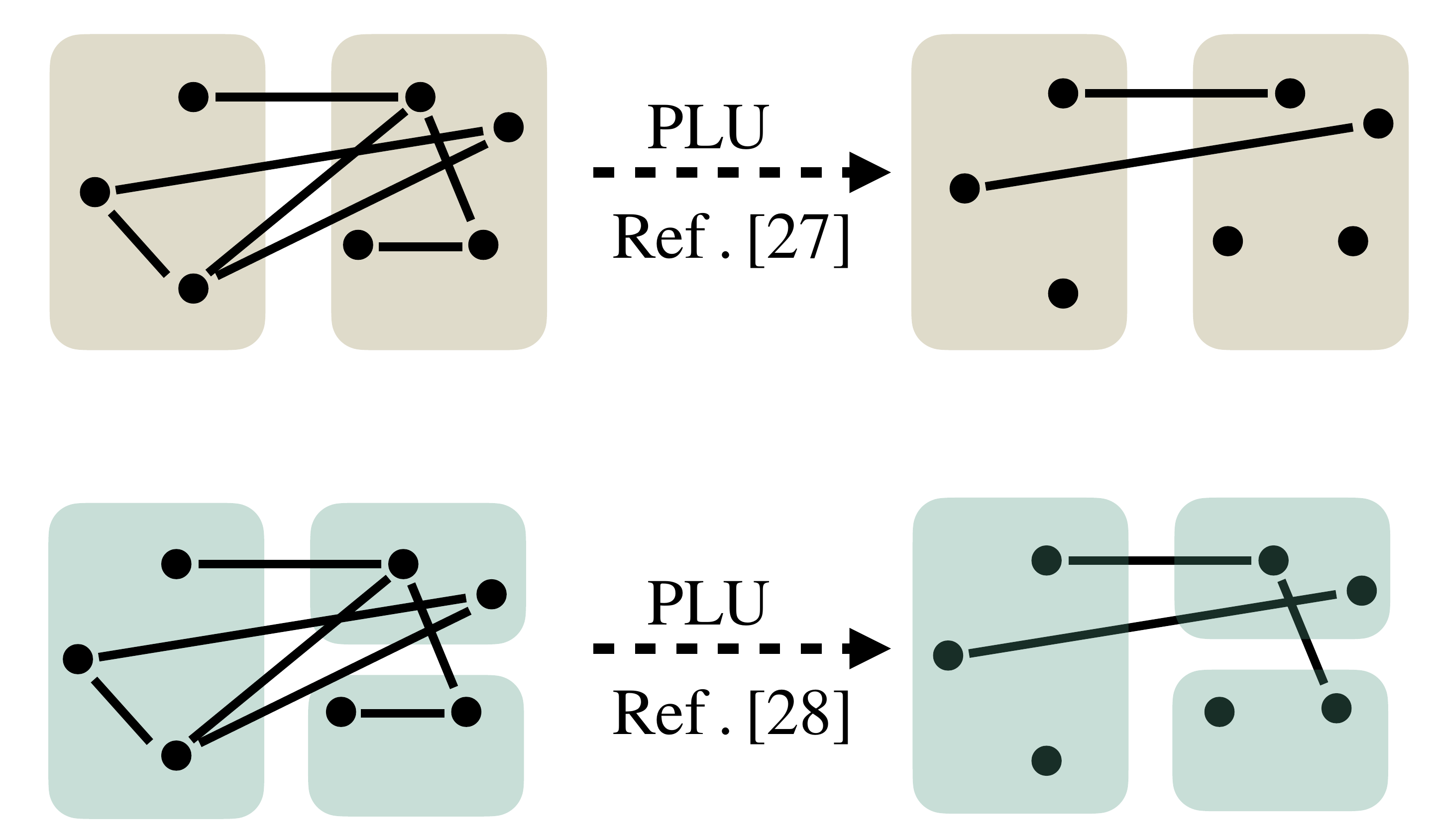}
	\caption{Top: A stabilizer state is distributed over two parties. Under party local unitary transformations the state can be converted to multiple independent copies of the Bell state, and locally separable states. Bottom: The same state, when grouped into three parties, can be decomposed into a GHZ state, a Bell state and locally separable states.\label{fig:fig1}}
\end{figure}

Furthermore, Ref.~\cite{Br06} provides criteria for the PLC equivalence of stabilizer states. In Ref.~\cite{Wi06} homological invariants of stabilizer states under local Clifford operations are studied, and were linked to the GHZ extraction condition of Ref.~\cite{Br06}. The extraction of GHZ states from random stabilizer states was investigated in Refs.~\cite{Le06,Wi08,Ne20}. However, a systematic approach to study the extraction of arbitrary stabilizer states in the multipartite scenario was missing. Putting forward such an approach that allows us to gain a deeper insight into the structure of entanglement in multipartite stabilizer states is precisely the aim of this work.

In this paper we study PLC transformations and state extraction of stabilizer states. First, we introduce notations and definitions relevant in the context of PLC transformations of stabilizer states in Section~\ref{sec_preliminaries}. Then we study how PLC transformations between graph states alter the corresponding graphs in Section~\ref{sec_gen_loc}, generalizing the concept of local complementation. As any stabilizer state is LC equivalent to a graph state this characterizes the action of any PLC operations on an arbitrary stabilizer state. In Section~\ref{sec_invariant_poly} we put forward a complete characterization of PLC equivalence classes based on PLC invariant polynomials which seems computationally not very practical. Then we recall previous works on PLC transformations of stabilizer states (Section \ref{sec_previous_results}). We mainly focus on the results of Ref.~\cite{Br06} which solves the classification of PLC classes for the $3$-partite case and introduces the notion of the entanglement generating set. In Section~\ref{sec_plc_trafo_stab}, we introduce one of our main results, namely, a new mathematical formalism, that we call the \emph{commutation matrix formalism}, that allows us to systematically study PLC transformations and state extraction of stabilizer states shared between an arbitrary number of parties. In this formalism we associate a tuple of alternating bilinear forms to every stabilizer state. We show that the classification of PLC classes is equivalent to the classification of commutation matrices up to simultaneous congruence. This type of problem was studied in linear algebra and we demonstrate how these results can be utilized in the commutation matrix formalism in Sections~\ref{sec_egs} and~\ref{sec_nec_suf_deco}. In particular this leads to necessary and sufficient conditions for a state to be contained in the EGS. In Section~\ref{sec_egs4}, using the commutation matrix formalism, we derive another main result by showing that the EGS$_M$ is an infinite set for $M\geq 4$, even with respect to party-local unitary transformations (PLU). Moreover, the new formalism allows us to compute the EGS$_4$ up to $10$ qubits. Going beyond the EGS$_4$ we provide evidence that the structure of the EGS becomes rapidly more complex. In Section~\ref{sec_qudit_systems} we show that the commutation matrix formalism extends to qudit stabilizer states of prime dimension. As the commutation matrices in this case are defined over the field $\mathbb{Z}_d$ for $d\geq 3$, the results of Ref.~\cite{Se88} imply that every qudit stabilizer states admits a unique decomposition into a tensor product of states in the EGS. Whether or not the uniqueness of the decomposition extends to all qubit stabilizer states is unclear. Finally, in case one considers more general states or operations it happens that the decomposition is not unique anymore, which we discuss in Section~\ref{sec:beyond}.

\section{Party-local operations\label{sec_party-local_operations}}
\subsection{Preliminaries\label{sec_preliminaries}}
Let us introduce some notation and recall some important definitions relevant in the context of PLC transformations of stabilizer states. We denote by $X,Y,Z$ the single-qubit Pauli operators
\begin{equation}
X=\begin{pmatrix}
0&1\\1&0
\end{pmatrix},\ \ Y=\begin{pmatrix}
0&-i\\i&0
\end{pmatrix},\ \ Z=\begin{pmatrix}
1&0\\0&-1
\end{pmatrix},
\end{equation}
and by $\mathcal{P}_n$ the $n$-qubit Pauli group, i.e., the group generated by $n$-fold tensor products of single-qubit Pauli operators. A stabilizer $\mathcal{S}$ is an abelian subgroup of the Pauli group $\mathcal{P}_n$, which does not contain $-\mathds{1}$. A maximal abelian subgroup of $\mathcal{P}_n$ is generated by $n$ independent commuting Pauli operators. The operators of a stabilizer forming a maximal abelian subgroup admit a unique, common $+1$ eigenstate, which is called a stabilizer state. Examples of stabilizer states are the Bell state $\ket{\phi^+}=1/\sqrt{2}(\ket{00}+\ket{11})$ and the $n$-qubit GHZ state $\ket{GHZ}=1/\sqrt{n}(\ket{0\ldots0}+\ket{1\ldots1})$.

It follows from the definition of stabilizer states that the density matrix $\ketbra{\psi}{\psi}$ of a $n$-qubit stabilizer state $\ket{\psi}$ with stabilizer $\mathcal{S}$ can we written as
\begin{equation}\label{eq:stabstate}
    \rho=\ketbra{\psi}{\psi}=\frac{1}{2^{n}}\sum_{s\in\mathcal{S}}s.
\end{equation}
Using this relation together with the fact that Pauli operators are traceless, one finds that the reduced state on a subset of qubits $\alpha$ is given by 
\begin{equation}\label{eq:reducedstate}
    \rho_\alpha=\frac{1}{2^{|\alpha|}}\sum_{s\in \mathcal{S}_\alpha}s,
\end{equation}
where $\mathcal{S}_\alpha$ contains those elements of $\mathcal{S}$ which act trivially on all qubits not in $\alpha$. Here, and in the following, $|\alpha|$ denotes the cardinality of the set $\alpha$. The rank of the reduced state $\varrho_{\alpha}$ depends on how many qubits in $\alpha$ are entangled with qubits outside of $\alpha$, e.g., if all qubits are entangled $\mathcal{S}_{\alpha}$ only contains the identity operator and $\varrho_\alpha$ has full rank.

The unitary normalizer (modulo $U(1)$) of the Pauli group $\mathcal{P}_n$ is called the $n$-qubit Clifford group $\mathcal{C}_n$, i.e., $\mathcal{C}_n=\{U\in U(2^n)|UPU^\dagger\in \mathcal{P}_n\forall P\in\mathcal{P}_n\}\backslash U(1)$. For any $n$, the Clifford group $\mathcal{C}_n$ is generated by single and two qubit gates, e.g.,
\begin{equation}
\begin{split}
    H&=\frac{1}{\sqrt{2}}\begin{pmatrix}
1&1\\1&-1
\end{pmatrix},\ P=\begin{pmatrix}
1&0\\0&i
\end{pmatrix},\\
CZ&=\begin{pmatrix}
1&0\\0&0
\end{pmatrix}\otimes \mathds{1}+\begin{pmatrix}
0&0\\0&1
\end{pmatrix}\otimes Z.
\end{split}
\end{equation}
A subgroup of the Clifford group $\mathcal{C}_n$ is the local Clifford group $\mathcal{C}_n^L$, which consists of those elements of $\mathcal{C}_n$ which can be written as a tensor product of elements of $\mathcal{C}_1$. 

Every stabilizer state is local Clifford (LC) equivalent to a graph state~\cite{Vn04}. A graph state $\ket{G}$ is the unique stabilizer state associated to a simple (no loops, undirected) graph $G=(V,E)$ with vertices $V=\{1,2,\ldots,n\}$ and edges $E\subset \{\{i,j\}|i,j\in V\}$. Its stabilizer $\mathcal{S}_G=\left<g_1,\ldots, g_{n}\right>$ is generated by the \emph{canonical generators}
\begin{equation}\label{eq:generators}
g_i=X_i\bigotimes_{j\in N_i}Z_j,
\end{equation}
where $N_i=\{k\in V|\{k,i\}\in E\}$ is the neighbourhood of vertex $i$ in $G$. A local Clifford transformation which maps a graph state to another graph state corresponds to a sequence of local complementations on the level of the graphs \cite{Vn04}. A local complementation with respect to (w.r.t.) the vertex $j\in E$ removes all edges between neighbours of vertex $j$ and adds all edges between neighbours of vertex $j$, which where previously not present in the graph. That is, it removes the edges in $\{\{k,l\}\in E|k,l\in N_j\}$ and adds all edges in $\{\{k,l\}\not \in E|k,l\in N_j\}$ to the graph $G$. This corresponds to the application of the operator $\exp(-i\pi/4X)$ on qubit $j$, and $\exp(i\pi/4Z)$ on each of its neighbours. Let us mention here that counting the number of local Clifford equivalent graph states is \#P-Complete~\cite{Da20}. Note that the graphical description of local Clifford transformations between graph states has also been studied for qudit systems~\cite{St17,Ri18}. A particularly useful way of representing an $n$-qubit graph state $\ket{G}$ is by its \emph{adjacency matrix} $\Gamma_G\in M_{n\times n}$, defined by
\begin{equation}
(\Gamma_G)_{i,j}=\begin{cases}
1, & \text{if $i$ and $j$ are connected} \\
0, & \, \text{otherwise.}
\end{cases}
\end{equation}
For a graph state $\ket{G}$ any element of its stabilizer $\mathcal{S}$ is a local symmetry. However, the state $\ket{G}$ can have more local symmetries than those in $\mathcal{S}$ (see Ref.~\cite{En20}). For example, consider a graph state $\ket{G}$ where qubit $1$ is only connected to qubit $2$. Then qubit $1$ is called a \emph{leaf}, qubit $2$ is called its \emph{parent} and for any $\alpha\in \mathbb{C}$ the operator
\begin{equation}
    e^{i\alpha X}\otimes e^{-i\alpha Z}\otimes \mathds{1}
\end{equation}
is a local symmetry of $\ket{G}$ (see Ref.~\cite{Zh09}). In Ref.~\cite{En20} some of the authors of the current work studied, among other things, all local symmetries of stabilizer states. While preparing the current work, we became aware of Ref.~\cite{Bo91} (and Ref.~\cite{Bo93}), which studies equivalence of graphs under local complementation. In this context, Ref.~\cite{Bo91} also derives the local Clifford symmetries of graph states.

In this work we are interested in a superset of local Clifford transformations, so called party-local Clifford transformations (PLC). To define these transformations, consider a partition $P(M,n)$ of the set $[n]:=\{1,\ldots n\}$ into $M$ subsets. In the following, we refer to the elements of $P(M,n)$ as parties. We call operators of the form
\begin{equation}
\bigotimes_{\alpha\in P(M,n)} C_{\alpha}\ \text{ with } \ C_{\alpha}\in\mathcal{C}_{|\alpha|},\label{eq_plc}
\end{equation}
party-local Clifford operators and we denote the subgroup generated by these operators by $\mathcal{C}^{P(M,n)}$. Note that $\mathcal{C}^{P(M,n)}\subset \mathcal{C}_n$ and equality only holds for $M=1$. We say that an $n$-qubit stabilizer state is $M$-partite if its qubits are distributed among $M$ parties. In the following we denote by $\text{Stab}(P(M,n))$ the set of $n$-qubit stabilizer states that are distributed according to the partition $P(M,n)$.

As we will discuss in Section~\ref{sec_preliminaries}, the Pauli group can be represented as a vector space over a finite field. We denote by $\mathbb{Z}_d$ the finite field of order $d$ for any prime number $d$. Over this field addition and multiplication are carried out modulo $d$.

\subsection{Motivation\label{sec_motivation}}
In this section we explain our motivation for studying PLC transformation of stabilizer states in quantum networks. We argue that PLC transformations are a physically motivated extension of local operations.

As experimental techniques progress, the distribution of bipartite entanglement between well connected nodes in a quantum network will no longer be an expensive resource. For instance, if two nodes are spatially close they might be considered as well connected. Nevertheless,  two nodes can be not spatially close and well connected if they can establish a large amount of entanglement over long distances, e.g., via a satellite~\cite{Ze18}. One can use this bipartite entanglement to implement non-local gates. This setting effectively breaks the network down into several parties, within which non-local operations are possible.

\begin{figure}
	\centering
	\includegraphics[width=0.9\columnwidth]{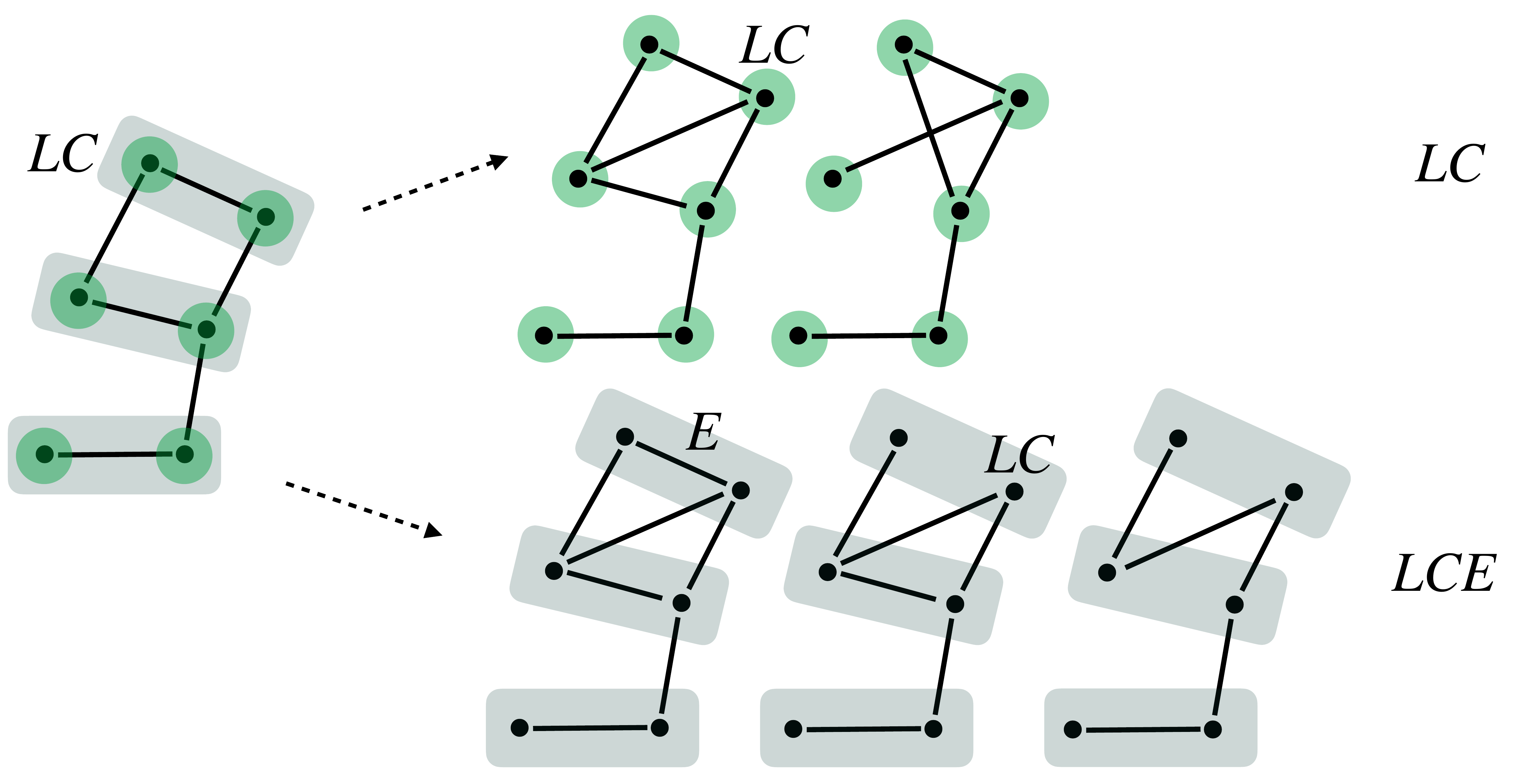}
	\caption{Local complementation (LC) on the parties highlighted in green compared to local complementation and edge removal/addition (LCE) on the parties highlighted in grey.\label{fig_fig_2}}
\end{figure}

Clifford operations play a special role in such a setting as they can be implemented deterministically via gate teleportation (see Refs.~\cite{Go99,Ci01}). Indeed, to apply a non-local gate $U\in\mathcal{C}_2$ to two qubits held by spatially separated parties, we can distribute the Choi state $U\otimes \mathds{1}\ket{\phi^+}^{\otimes 2}$ among the parties. After performing a Bell measurement in each of the parties (each involving one qubit of the Choi state and the one originally held by the party), the state of the two unmeasured qubits is
\begin{equation}
UP\varrho P^\dagger U^\dagger=UPU^\dagger(U\varrho U^\dagger)U P^\dagger U^\dagger,
\end{equation}
where $P\in\mathcal{P}_n$ depends on the measurement result and $\rho$ denotes the initial state of the two qubits originally held by the parties. As Clifford operators map the Pauli group onto itself, local Pauli corrections suffice to invert $UPU^\dagger$. As the Clifford group is generated by single and two-qubit gates (see Section~\ref{sec_preliminaries}), the above well-known construction allows to implement any non-local Clifford gate. For other non-local gates, which are not Clifford gates, the correction $UPU^\dagger$ will not be local in general and thus a deterministic local correction is not possible.

In this work we only consider PLC transformations among special types of states, namely stabilizer states. This is, on the one hand, due to the outstanding importance of these states for quantum information protocols. On the other hand, Clifford operations map stabilizer states to stabilizer states and thus these states are a natural candidate to study in the setting we consider. Additionally, the experimental requirements to implement party-local Clifford gates on stabilizer states are actually significantly less than what we described above. For stabilizer states, the distribution of Bell states, local $Z$ basis measurements and the application of Clifford gates, which act locally on the nodes of the network, suffice to implement any party-local Clifford gate. Indeed, it is easy to see that after distributing the Bell state between two parties both parties can apply a $CZ$ between the qubits of the stabilizer state and the Bell state, perform local complementation on the qubits of the Bell state and then measure them in the $Z$ basis to implement a non-local $CZ$ gate up to local Pauli corrections. 

Motivated by the above, it is the goal of this work to approach a complete characterization of PLC equivalence classes of stabilizer states.

\subsection{Generalized local complementation\label{sec_gen_loc}}
In this section, we study how party-local Clifford transformations from graph states to graph states alter the corresponding graphs. We mentioned in Section \ref{sec_preliminaries} that local Clifford transformations which map a graph state to another graph state correspond to a sequence of local complementations on the level of the graphs \cite{Vn04}. We analyze how these results generalize to party-local Clifford transformations in the following theorem.
\begin{thm}\label{thm_gen_loc}
	Let $\ket{G}$ and $\ket{G'}$ be $n$-qubit graph states distributed among the parties $P(M,n)$. Then $\ket{G}$ is PLC equivalent to $\ket{G'}$ if and only if the graphs $G$ and $G'$ are related via a sequence of local complementations and edge additions/removals within parties (LCE).
\end{thm}
Let us compare Theorem \ref{thm_gen_loc} to the results of \cite{Vn04} for local Clifford operations. We see that the only additional power party-local Clifford transformations have over local Clifford transformations is the possibility of adding and removing edges within parties. However, note that adding and removing edges within parties and local complementations are non-commuting operations and thus LCE generates larger orbits. The proof of Theorem \ref{thm_gen_loc} makes use of Lemma~\ref{lem_gen_loc_formal}, which is stated and proven in Appendix~\ref{app:PLC} . Here, we state only an informal version of it, as this is sufficient to understand the reason for the validity of Theorem~\ref{thm_gen_loc}.

\begin{lem}[informal version]\label{lem_gen_loc}
    The action of any two-qubit Clifford gate on a graph state is equivalent to the action of a suitable sequence of LCE up to local Clifford corrections.
\end{lem}

Using Lemma~\ref{lem_gen_loc} and the fact that $\mathcal{C}_n$ is generated by single and two-qubit Clifford operators (see Section \ref{sec_preliminaries}) we can then show Theorem \ref{thm_gen_loc}. The proof can be found in Appendix \ref{app:PLC}.

\subsection{Classification of PLC classes via invariant polynomials\label{sec_invariant_poly}}
The characterization of equivalence classes of states under certain operations is a common task in quantum information theory. In the context of stochastic local operations assisted by classical communication (SLOCC), invariant polynomials proved to be a fruitful approach \cite{Go13}. Motivated by this, we study the use of invariant polynomials in the classification of PLC classes of stabilizer states. 

Consider a stabilizer state $\ket{\psi}\in \text{Stab}(P(M,n))$. We want to identify all states in the PLC class of $\ket{\psi}$. To this end, let us choose a set of generators $g_1^\psi ,\ldots, g_n^\psi $ for the stabilizer $\mathcal{S}$ of $\ket{\psi}$ and define the polynomial 
\begin{equation}
I_{\ket{\psi}}(\ket{\phi})=\sum_{C\in \mathcal{C}_n^{P(M,n)}}\prod_{j=1}^{n}\braket*{\phi}{C^\dagger g_j^\psi C|\phi}^2\label{eq_invariant_poly}.
\end{equation}
Observe that this polynomial is obviously PLC invariant, i.e., for all $C\in\mathcal{C}^{P(M,n)}$ it holds that $I_{\ket{\psi}}(C\ket{\phi})=I_{\ket{\psi}}(\ket{\phi})$, and that $I_{\ket{\psi}}=I_{C\ket{\psi}}$. Moreover, for any stabilizer state $\ket{\phi}$ this polynomial is nonzero if and only if $\ket{\phi}$ is PLC equivalent to $\ket{\psi}$, as we show in the following. Using the fact that any stabilizer state can be written in terms of its stabilizer as in Equation~\eqref{eq:stabstate}, and that every $P\in \mathcal{P}_n$ not equal to the identity is traceless, a term in the polynomial in Equation~\eqref{eq_invariant_poly} is nonzero if and only if $\{C^\dagger g_i^\psi C\}_{i=1}^n$ generate the stabilizer of $\ket{\phi}$ for some $C\in\mathcal{C}^{P(M,n)}$. Hence, the polynomial is nonzero if and only if $\ket{\psi}$ and $\ket{\phi}$  are PLC equivalent. We will see in Section \ref{sec_previous_results} that it follows from the results of Ref.~\cite{Br06} that the polynomial $I_{\ket{\psi}}$ only depends on the commutation relations of the generators $g_1^\psi ,\ldots, g_n^\psi$ on the individual parties.

While each PLC equivalence class can be characterized by a polynomial as in Equation \eqref{eq_invariant_poly}, this characterization is not very practical. To our knowledge, the only way to compute the value of such a polynomial (with the possible exception of a small number of parties) for a given stabilizer state is to explicitly perform the sum over all terms. If that is the case, computing the polynomial is not more efficient than directly checking PLC equivalence by computing the full PLC orbit.

Let us remark that invariants have been used previously to study PLC equivalence of stabilizer states. Ref.~\cite{Vn05} introduces a finite set of invariants that completely characterize local Clifford equivalence of stabilizer states. That is, two states are local Clifford equivalent if and only if they agree in these invariants. Ref.~\cite{Wi06} takes a similar approach, defining so called homological invariants of stabilizer states, which apply to arbitrary party sizes, and even qudit stabilizer states. These are used to rederive the results on the extraction of GHZ states from Ref.~\cite{Br06} (see Section \ref{sec_previous_results}).

\subsection{Previous results\label{sec_previous_results}}
In this section we summarize the main results obtained in the context of PLC transformations of stabilizer states, focusing on those that we use in the following.

The classification of $3$-partite PLC equivalence classes of stabilizer states was completely solved in Ref.~\cite{Br06}. There, it is shown that every $3$-partite stabilizer state is PLC equivalent to a tensor product of a unique number of copies of the states $\ket{0}$, $\ket{\phi^+}$ and $\ket{GHZ}$, see Figure~\ref{fig:fig1} for an example. That is, two stabilizer states are PLC equivalent if and only if their unique decompositions coincide. In the following we make use of some of the techniques developed in Ref.~\cite{Br06}. Therefore, let us briefly introduce some of the notation that was being used and discuss the content of the paper in more detail. First, recall that the Pauli group $\mathcal{P}_n$ modulo phases is isomorphic to the vector space $\mathbb{Z}_2^{2n}$. To see that, one identifies the single-qubit Pauli operators and the identity operator with elements of $\mathbb{Z}_2^2$ as
\begin{equation}
\sigma_{0,0}=\mathds{1},\ \sigma_{1,0}=X,\ \sigma_{1,1}=Y,\ \sigma_{0,1}=Z. 
\end{equation}
Then, one can define a map $\sigma:\mathbb{Z}_2^{2n}\rightarrow \mathcal{P}_n$ via
\begin{equation}
\sigma((a_1,b_1,\ldots,a_n,b_n))=\sigma_{a_1,b_1}\otimes\ldots\otimes\sigma_{a_n,b_n}.
\end{equation}
From this relation one sees that the multiplication of Pauli operators corresponds to the addition of their corresponding vectors in $\mathbb{Z}_2^{2n}$, i.e.,
\begin{equation}\label{eq:stab_mult}
    \sigma(\vec{f})\sigma(\vec{g})\propto\sigma(\vec{f}+\vec{g})\ \ \ \vec{f},\vec{g}\in\mathbb{Z}_2^{2n}.
\end{equation}
The commutation relations of the Pauli operators define a symplectic bilinear form $\omega:\mathbb{Z}_2^{2n}\rightarrow \mathbb{Z}_2$ by
\begin{equation}
\sigma(\vec{f})\sigma(\vec{g})=e^{i\pi \omega(\vec{f},\vec{g})}\sigma(\vec{g})\sigma(\vec{f}).
\end{equation}
Indeed, writing $\vec{f}=(a_1,b_1\ldots,a_n,b_n)$ and $\vec{g}=(a_1',b_1',\ldots,a_n',b_n')$ the above definition implies that
\begin{align}
\begin{split}&\omega(\vec{f},\vec{g})\\&=(a_1,\ldots,a_n,b_1,\ldots b_n)\begin{pmatrix}
0&\mathds{1}\\\mathds{1}&0
\end{pmatrix}(a_1',\ldots,a_n',b_1',\ldots b_n')^T.
\end{split}
\end{align}
Thus, the bilinear form $\omega$ satisfies $\omega(\vec{f},\vec{f})=0$, and is non-degenerate, i.e., $\omega(\vec{f},\vec{g})=0$ for all $\vec{g}$ implies that $\vec{f}=0$. We conclude that $\mathbb{Z}^{2n}_2$ together with $\omega$ is a symplectic vector space. A brief introduction to symplectic vector spaces can be found in Appendix \ref{sec_sym}. 

Every stabilizer $\mathcal{S}\subset\mathcal{P}_n$ can be written as $\mathcal{S}=\{\epsilon(\vec{f})\sigma(\vec{f})|\vec{f}\in V^{\mathcal{S}}\}$ where $ V^{\mathcal{S}}$ is a maximally isotropic subspace of $\mathbb{Z}^{2n}_2$ and $\epsilon: V^{\mathcal{S}}\rightarrow \{\pm 1\}$ is a suitable sign choice. A maximally isotropic subspace is a subspace which is equal to its orthogonal complement defined via the symplectic form $\omega$ (for details see Appendix \ref{sec_sym}). A set of generators of $\mathcal{S}$ corresponds to a basis of $V^{\mathcal{S}}$ via the map $\sigma$. Choosing a different set of generators of $\mathcal{S}$ is equivalent to a basis change in $V^{\mathcal{S}}$ via Equation~\eqref{eq:stab_mult} and it is described by an invertible matrix. Note that subspaces which are isotropic but not maximally isotropic, i.e., which are contained in their orthogonal complement but not necessarily equal to it, correspond to stabilizer codes.

Any Clifford operator $C\in\mathcal{C}_n$ defines a linear map $u:\mathbb{Z}_2^{2n}\rightarrow\mathbb{Z}_2^{2n}$ via
\begin{equation}
C\sigma(\vec{f})C^\dagger\propto \sigma(u(\vec{f})).\label{eq_eq4}
\end{equation}
As $C$ is unitary, $u$ is invertible and preserves the symplectic form $\omega$, i.e., it is an isometry. In fact, one can show that any isometry on $\mathbb{Z}_2^{2n}$ corresponds to a $C\in\mathcal{C}_n$ via the above relation. This correspondence is unique up to phases. 

Consider a set of parties $P(M,n)$. For every $\vec{f}\in\mathbb{Z}_2^{2n}$ and $\alpha\in P(M,n)$ we denote by $\vec{f}_\alpha\in \mathbb{Z}_2^{2|\alpha|}$ the restriction of $\vec{f}$ to party $\alpha$. For every party $\alpha\in P(M,n)$, we consider two subspaces, the local subspace
\begin{equation}
 V^{\mathcal{S}}_\alpha=\{\vec{f}\in  V^{\mathcal{S}}|\vec{f}_\beta=0\ \forall \beta\in M,\beta\neq \alpha\}
 \label{eq:local}
\end{equation}
and the colocal subspace
\begin{equation}
V_{ \hat{\alpha}}^{\mathcal{S}}=\{\vec{f}\in V^{\mathcal{S}}|\vec{f}_\alpha=0\}.
\label{eq:colocal}
\end{equation}
That is, while $ V^{\mathcal{S}}_\alpha$ contains those vectors of $V^{\mathcal{S}}$ that vanish on every party but $\alpha$, the subspace $ V^{\mathcal{S}}_\alpha$ contains those vectors that vanish on party $\alpha$. Note that the dimension of $V_\alpha^\mathcal{S}$ is related to rank of the reduced state via ${\rm rk}(\rho_\alpha)=2^{|\alpha|-{\rm dim}(V_\alpha^\mathcal{S})}$ (see Equation~\eqref{eq:reducedstate}). In this context, it is helpful to define the support of a vector $\vec{f}=(a_1,b_1,\ldots,a_n,b_n)\in\mathbb{Z}_2^{2n}$ as 
\begin{equation}
\text{supp}(\vec{f})=\{i\in[n]|a_i=1\lor b_i=1\}.
\end{equation}
Similarly, the support of a Pauli operator $P\propto \sigma(\vec{f})$ is given $\text{supp}(P)=\text{supp}(\vec{f})$, i.e., the support of $P$ is where it acts non-trivially.

With all the definitions in place let us recall the findings of Ref.~\cite{Br06} that we will use here. There are two types of results in Ref.~\cite{Br06}. On the one hand, the authors provide necessary and sufficient conditions for PLC equivalence of stabilizer states. On the other hand, they use those conditions to find the exact number of $M$-qubit, $M$-partite GHZ states one can extract from an $M$-partite stabilizer state. Using this result, as well as the conditions for PLC equivalence, they provide a complete solution to the problem of PLC equivalence for $3$-partite stabilizer states. We start by summarizing the necessary and sufficient conditions for PLC equivalence derived in Ref.~\cite{Br06}.

It is shown in Ref.~\cite{Br06} that the PLC equivalence of two stabilizer states solely depends on how elements of their stabilizers commute on the individual parties, as stated in the following theorem.
\begin{thm}[\cite{Br06}]\label{thm_PLC_equiv}
	Let $\ket{\psi},\ket{\phi}\in \text{Stab}(P(M,n))$. Then, $\ket{\psi}$ and $\ket{\phi}$ are PLC equivalent if and only if there exists a basis $\vec{g}_1,\ldots, \vec{g}_n$ of $V^{{\mathcal{S}}_\psi}$ and a basis $\vec{f}_1,\ldots, \vec{f}_n$ of $V^{{\mathcal{S}}_\phi}$ such that
	\begin{equation}
	\omega((\vec{f}_i)_\alpha ,(\vec{f}_j)_\alpha)=\omega((\vec{g}_i)_\alpha ,(\vec{g}_j)_\alpha)
	\end{equation}
	for all $\alpha \in P(M,n)$ and $i,j\in[n]$.
\end{thm}
That is, two stabilizer states are PLC equivalent iff there exist sets of generators for $\mathcal{S}_\psi$ and $\mathcal{S}_\phi$ whose respective commutation relations on all parties agree. In Ref.~\cite{Br06}, this theorem is only stated for stabilizer states of maximal local rank, i.e., states for which $\text{dim}(V_\alpha^\mathcal{S})=0$ for all $\alpha\in P(M,n)$. However, the authors mention that it extends to general stabilizer states. In Appendix \ref{sec_proof_plcequiv_qudits} we provide a proof for the general statement in the case of qudit systems of prime dimension, as we use it in Section~\ref{sec_qudit_systems} where we study qudit systems of prime dimension.

Moreover, Ref.~\cite{Br06} provides a theorem which gives necessary and sufficient conditions for when a stabilizer state $\ket{\phi}$ can be extracted from a stabilizer state $\ket{\psi}$ via PLC. We say a state $\ket{\phi}\in \text{Stab}(P(M,m))$ can be extracted from a state $\ket{\psi}\in \text{Stab}(P(M,n))$ with $n\ge m$ if there exists a PLC transformation $C\in\mathcal{C}^{P(M,n)}$ and a $(n-m)$-qubit stabilizer state $\ket{\phi'}$ such that $C\ket{\psi}=\ket{\phi}\otimes\ket{\phi'}$. In this case, we say $\ket{\psi}$ decomposes into $\ket{\phi}$ and $\ket{\phi'}$.
\begin{thm}[\cite{Br06}]\label{thm_PLC_extract}
	Let $\ket{\psi}\in \text{Stab}(P(M,n))$ and $\ket{\phi}\in \text{Stab}(P(M,m))$ with $n\ge m$. The state $\ket{\phi}$ can be extracted from $\ket{\psi}$ if and only if there exists a linear injective map $T: V^{\mathcal{S}_\phi}\rightarrow V^{\mathcal{S}_\psi}$ such that
	
	\begin{enumerate}
		\item[(i)] $\omega(T(\vec{f})_\alpha,T(\vec{g})_\alpha)=\omega(\vec{f}_{\alpha},\vec{g}_{\alpha})\ \forall \vec{f},\vec{g}\in V^{\mathcal{S}_\phi}\ \alpha\in P(M,n)$
		\item[(ii)] $[T(V^{\mathcal{S}_\phi})]_{\hat{\alpha}}=T(V^{\mathcal{S}_\phi}_{\hat{\alpha}})$ for all $\alpha\in P(M,n)$.
	\end{enumerate}
\end{thm}

Condition $(i)$ in Theorem~\ref{thm_PLC_extract} guarantees that the map $T$ preserves commutation relations on every party. Condition $(ii)$ enforces that the party support of $f$ and $T(f)$ is the same for all $\vec{f}\in V^{\mathcal{S}_\phi}$. This condition is necessary as PLC transformations cannot change the party support of Pauli operators. Note that Theorem~\ref{thm_PLC_equiv} considers an extremal case of Theorem~\ref{thm_PLC_extract} where $n=m$. One can show that condition $(ii)$ in Theorem~\ref{thm_PLC_extract} becomes obsolete in this case.

A simple case of extractability occurs if there exists a party $\alpha\in P(M,n)$ such that $\text{dim}(V^{\mathcal{S}}_\alpha)=d_\alpha>0$. It follows directly from Theorem~\ref{thm_PLC_extract} that then $\ket{\psi}$ is PLC equivalent to a state $\ket{\phi}\otimes \ket{0}^{\otimes d_\alpha}$. Note that an analogous reasoning also applies to the qudit case (see Appendix~\ref{sec_proof_plcequiv_qudits}).

Using, in particular, Theorem~\ref{thm_PLC_extract} one can determine the maximum number of $M$-qubit GHZ states shared among all parties that can be extracted from a state $\ket{\psi}\in \text{Stab}(P(M,n))$ via PLC. This number is given by the minimum number of generators of $\mathcal{S}_\psi$ which are truly $M$-partite. To be more precise, let us consider the subspace 
\begin{equation}
V_{loc}=\text{span}\bigcup_{\alpha\in P(M,n)}V^S_{\hat{\alpha}}.
\end{equation}
It is shown in Ref.~\cite{Br06} that the maximum number of $M$-qubit GHZ states distributed among all $M$ parties that can be extracted from $\ket{\psi}$ via PLC is given by $\Delta=n-\text{dim}(V_{loc})$.

In the $3$-partite case, this result in combination with Theorem~\ref{thm_PLC_extract} leads to a complete characterization of all PLC equivalence classes of stabilizer states of arbitrary qubit number. As we have mentioned before, it is shown in Ref.~\cite{Br06}, that every $3$-partite stabilizer state decomposes into a tensor product of $k_1$ copies of the $\ket{0}$ state, $k_2$ copies of the $\ket{\phi^+}$ state, and $k_3$ copies of the $\ket{GHZ}$ state under PLC. The number of copies of each of those states is unique and can be computed efficiently. Two stabilizer states are PLC equivalent if and only if their decompositions coincide, i.e., if the numbers of copies in the decompositions of each of the states coincides. This statement still holds if one considers party-local unitary operations (PLU) instead of PLC operations. Conversely, any tensor product of the states $\ket{0},\ket{\phi^+}$ and $\ket{GHZ}$ uniquely characterizes a PLC class (of three parties), and also a PLU class (Ref.~\cite{Br06}).

The authors of Ref.~\cite{Br06} therefore suggest to call $\{\ket{0},\ket{\phi^+},\ket{GHZ}\}$ the entanglement generating set (EGS) for $3$-partite stabilizer states. In this work, we use the EGS w.r.t. PLC operations to characterize PLC equivalence classes for an arbitrary number of parties. The entanglement generating set (EGS) for $M$-partite stabilizer states under PLC is a minimal set EGS$_{M}$ of indecomposable $M$-partite stabilizer states such that every $M$-partite stabilizer state decomposes under PLC into a tensor product of the states contained in EGS$_{M}$.
It follows directly from its definition that the EGS is unique up to PLC transformations of the individual states it contains. In particular, it contains one representative of every PLC class, which consists of indecomposable stabilizer states. Note that there are stabilizer states which are LU equivalent but not LC equivalent (see Ref.~\cite{Ji10}). Therefore, the EGS, as we define it, can be overcomplete regarding PLU transformations. Moreover, the EGS$_M$ is defined for a fixed number of parties $M$ but an unbounded number of qubits. Thus, it is possible that it contains infinitely many states. 

The results of Ref.~\cite{Br06} imply that $3$-partite stabilizer states decompose uniquely into states from the EGS$_3$. Note that the results of~\cite{Br06} were also partially extended to $2$- and $3$-partite qudit systems, with $d$ such that its prime decomposition only contains primes of power $0$ or $1$ in Ref.~\cite{Gr11}. More precisely, the authors show that any $2$- and $3$-partite qudit stabilizer state, with $d$ as defined above, decomposes into a tensor product of the states $\ket{0}$, $\ket{\phi_d^+}=\frac{1}{\sqrt{d}}\sum_{j}\ket{jj}$ and $\ket{GHZ^d}=\frac{1}{\sqrt{d}}\sum_j \ket{jjj}$.

In the light of the above summary let us briefly outline the rest of the paper. In the following sections we introduce the commutation matrix formalism as a powerful tool to study PLC equivalence of stabilizer states. This allows us to use results from the classification theory of bilinear forms to study the structure of the EGS$_M$ for qubits and qudits. Particularly, we derive necessary and sufficient conditions for when a state is decomposable. While we do not know whether or not qubit stabilizer state decompose uniquely into states from the EGS, we prove that for qudit stabilizer states of prime dimension the decomposition is indeed unique.

\section{PLC transformations of stabilizer states and the commutation matrix formalism\label{sec_plc_trafo_stab}}

As we saw in the previous section, whether or not two stabilizer states are PLC equivalent solely depends on the local commutation relations of operators in their stabilizers (Theorem~\ref{thm_PLC_equiv} shown in~\cite{Br06}). In this section, we introduce a new formalism to exploit these findings in studying the PLC transformations of stabilizer states.

Let $\ket{\psi}\in \text{Stab}(P(M,n))$ with stabilizer $\mathcal{S}$ and let $\{\vec{f}_1,\ldots, \vec{f}_n\}$ be a basis of its corresponding maximally isotropic subspace $V^{\mathcal{S}}$. As we are only interested in the subspace $V^{\mathcal{S}}$, we define the linear map $R: V^{\mathcal{S}}\rightarrow \mathbb{Z}^n_2$ via $R\vec{f}_i=\vec{e}_i$ for all $i\in[n]$, where $\{\vec{e}_i\}$ denotes the standard basis of $\mathbb{Z}_2^n$~\footnote{This representation of the maximally isotropic subspace $V^\mathcal{S}$ is equivalent to the representation of any element $g\in \mathcal{S}$ as $g=g_1^{x_1}g_2^{x_2}\cdots g_n^{x_n}$, where $\vec{x} \in \mathbb{Z}_2^n$.}. The map $R$ is invertible and therefore describes an isomorphism. For every party $\alpha\in P(M,n)$ we define a bilinear form $C_\alpha:\mathbb{Z}_2^n\times \mathbb{Z}_2^n \rightarrow \mathbb{Z}_2$ via
\begin{equation}
C_\alpha(\vec{k},\vec{l})=\omega((R^{-1}\vec{k})_\alpha, (R^{-1}\vec{l})_\alpha)\label{eq_com_mat}.
\end{equation}
As $\omega$ is a symplectic form, the bilinear forms $(C_\alpha)_{\alpha}$ are \emph{alternating}, that is, $C_\alpha(\vec{k},\vec{l})=-C_\alpha(\vec{l},\vec{k})$, and $C_\alpha(\vec{k},\vec{k})=0$ for all $\vec{k},\vec{l}\in \mathbb{Z}_2^n$. As we operate over the field $\mathbb{Z}_2$ we also have that $C_\alpha(\vec{k},\vec{l})=C_\alpha(\vec{l},\vec{k})$. Note that they are not symplectic in general as they can be degenerate, i.e., if $C_\alpha(\vec{k},\vec{l})=0$ for all $\vec{l}\in\mathbb{Z}_2^n$, then this does not imply that $\vec{k}=0$.

In slight abuse of notation we use $C_\alpha$ also to denote the matrix representation of the respective bilinear form w.r.t. to the basis $\{\vec{e}_i\}_{i\in[n]}$. Let us now make the following definition.
\begin{defi}\label{def_com_mat}
	Let $\ket{\psi}$ be an $M$-partite stabilizer state. We call the matrices $(C_\alpha )_{\alpha\in P(M,n)}$, describing the bilinear forms defined in Equation~\eqref{eq_com_mat} w.r.t. to the basis choice $\{\vec{f}_i\}$ of $ V^{\mathcal{S}}$, commutation matrices.
\end{defi}
A commutation matrix $C_{\alpha}$ contains all information about the commutation relations of the Pauli operators $\sigma(\vec{f}_j)$ on the party $\alpha$. More precisely, $(C_\alpha)_{ij}=0$ iff $\sigma(\vec{f}_i)$ and $\sigma(\vec{f}_j)$ commute on party $\alpha$, i.e., $[\sigma(\vec{f}_i)_\alpha,\sigma(\vec{f}_j)_\alpha]=0$. Moreover, by replacing the field $\mathbb{Z}_2$ with a finite field $\mathbb{Z}_d$ for $d\geq 3$, commutation matrices can also be defined for qudit systems of prime dimension, which we discuss in Section~\ref{sec_qudit_systems}.

Let us discuss some properties of commutation matrices. Commutation matrices are alternating as they describe alternating forms. Moreover they satisfy
\begin{equation}
\sum_{\alpha\in P(M,n)}C_\alpha=0,\label{eq_sum}
\end{equation}
where addition is performed modulo $2$. This condition is a direct consequence of the fact that all generators globally commute, i.e., that $ V^{\mathcal{S}}$ is isotropic. Indeed, as $[\sigma(\vec{f}_i),\sigma(\vec{f}_j)]=0$ for all $i,j$, the operators $\sigma(\vec{f}_i)$ and $\sigma(\vec{f}_j)$ anti-commute on an even number of parties. Therefore $\sum_\alpha(C_\alpha)_{ij}=0$ for all $i,j$. Moreover, for two parties $\alpha_i,\alpha_j\in P(M,n)$ one finds that
\begin{equation}
    C_{\alpha_1}+C_{\alpha_2}=C_{\alpha_1\cup \alpha_2}.
\end{equation}

As noted in the definition, the commutation matrices depend on the chosen basis $\{\vec{f}_1,\ldots, \vec{f}_n\}$. A basis change in $V^{\mathcal{S}}$, which is equivalent to choosing a different set of generators for $\mathcal{S}$, is described by an invertible matrix $Q\in M_{n\times n}(\mathbb{Z}_2)$ on $\mathbb{Z}_2^n$. As commutation matrices represent bilinear forms, the commutation matrices with respect to the new basis are given by
\begin{equation}
    (Q^TC_\alpha Q)_{\alpha\in M}.
\end{equation}
In order to illustrate the formalism, let us present some examples of commutation matrices that are relevant in the following. For the state $\ket{\phi^+}$ distributed among the parties $P(2,2)=\{\{1\},\{2\}\}$ the commutation matrices w.r.t the canonical generators $g_1=X_1Z_2$ and $g_2=Z_1X_2$, defined by Equation~\eqref{eq:generators}, are
\begin{equation}\label{eq_com_mat_bell}
C_{\{1\}}=\begin{pmatrix}
0&1\\1&0
\end{pmatrix}\,\text{and}\,C_{\{2\}}=\begin{pmatrix}
0&1\\1&0
\end{pmatrix}.
\end{equation}
For the state $\ket{GHZ}$ distributed among the parties $P(3,3)=\{\{1\},\{2\},\{3\}\}$ the commutation matrices w.r.t the canonical generators $g_1=X_1Z_2$, $g_2=Z_1X_2Z_3$ and $g_3=Z_2X_3$ are
\begin{align}\label{eq_com_mat_ghz}
C_{\{1\}}&=\begin{pmatrix}
0&1&0\\1&0&0\\0&0&0
\end{pmatrix},C_{\{2\}}=\begin{pmatrix}
0&1&0\\1&0&1\\0&1&0
\end{pmatrix},\notag \\
C_{\{3\}}&=\begin{pmatrix}
0&0&0\\0&0&1\\0&1&0
\end{pmatrix}.
\end{align}
It can be easily verified that they satisfy all the requirements listed above, i.e., they are alternating, and sum up to zero.

If we instead choose the generators $g_1,g_2,g_2g_3$, the associated tuple of commutation matrices is 
\begin{align}
\tilde{C}_{\{1\}}&=\begin{pmatrix}
0&1&1\\1&0&0\\1&0&0
\end{pmatrix},\tilde{C}_{\{2\}}=\begin{pmatrix}
0&1&1\\1&0&1\\1&1&0
\end{pmatrix},\notag \\
\tilde{C}_{\{3\}}&=\begin{pmatrix}
0&0&0\\0&0&1\\0&1&0
\end{pmatrix}.
\end{align}
Again, it is straightforward to verify that these commutation matrices result from the original ones by a congruence transformation with the invertible matrix
\begin{equation}
    Q=\begin{pmatrix}
    1&0&0\\
    0&1&0\\
    0&1&1\\
    \end{pmatrix}.
\end{equation}
In fact, the states that we have just discussed are graph states, as the generators of their stabilizers correspond to the canonical generators in Equation~\ref{eq:generators}. For these states there exists a simple way to compute a tuple of commutation matrices from the adjacency matrix $\Gamma_G$. One finds that if $\Gamma_G(i)$ denotes the matrix that agrees with $\Gamma_G$ in row and column $i$ and is zero everywhere else, then
\begin{equation}
C_\alpha=\sum_{i\in\alpha} \Gamma_G(i).
\end{equation}
However, not every tuple of alternating matrices which satisfies Equation~\eqref{eq_sum} corresponds to a stabilizer state, i.e., a maximally isotropic subspace. As we show in the following Theorem, the ranks of the matrices in tuples corresponding to stabilizer states are constrained.

\begin{thm}\label{thm_rank_condition}
    Let $(C_\alpha)_{\alpha\in P(M,n)}$, $C_\alpha\in M_{n\times n}(\mathbb{Z}_2)$ be a tuple of alternating matrices such that $\sum_{\alpha} C_\alpha=0$ for some partition $P(M,n)$. Then there exists a $\ket{\psi}\in \text{Stab}(P(M,n))$ such that $(C_\alpha)_{\alpha}$ are the corresponding commutation matrices if and only if
	\begin{equation}
	2\,{\rm rk}([C_\alpha]_{\alpha})=\sum_{\alpha} {\rm rk}(C_\alpha)\label{eq_rank_condition}.
	\end{equation}
	Here, $[C_\alpha]_{\alpha}$ denotes the $n\times (nm)$ matrix obtained by concatenating all matrices $C_\alpha$, and the rank is taken over the field $\mathbb{Z}_2$.
\end{thm}
To explain this condition let us note that the rank of a commutation matrix is related to the rank of the reduced stabilizer state as $2^{\text{rk}(C_\alpha)/2}=\text{rk}(\rho_\alpha)$. Therefore, $\text{rk}(C_\alpha)/2$ coincides with the number of qubits in party $\alpha$ that are effectively entangled to another party (see discussion in Appendix~\ref{app:comm}). Observe that ${\rm rk}([C_\alpha]_{\alpha})={\rm rk}([QC_\alpha Q^T]_{\alpha})$ for any invertible matrix $Q$. Therefore, $n-{\rm rk}([C_\alpha]_{\alpha})$ is the maximum number of generators that locally commute with all other generators of $\mathcal{S}$. As $\mathcal{S}$ is a maximally abelian subgroup it follows that those generators can be chosen to have support only on one party, and thus correspond to extractable qubits in the state $\ket{0}$. Thus, the rank condition in Equation~\eqref{eq_rank_condition} states the following. The sum over the number of qubits in each party which are effectively entangled to another one outside the party has to be equal to the difference between $n$ and the total number of qubits in the state $\ket{0}$ that can be extract from $\ket{\psi}$. The latter number is precisely given by $n-{\rm rk}([C_\alpha]_{\alpha})$.

For general alternating matrices $(C_\alpha)_{\alpha\in P(M,n)}$ which sum up to zero, Equation~\eqref{eq_rank_condition} becomes an inequality
\begin{equation}
2\, {\rm rk}([C_\alpha]_\alpha)\le \sum_{i=1}^M {\rm rk}(C_i).\label{eq_rank_inequ}
\end{equation}
Commutation matrices are precisely those sets of matrices that saturate this bound. For the details of the proof of this inequality and for the proof of Theorem~\ref{thm_rank_condition} see Appendix ~\ref{app:comm} and the discussion below. Note that one can show that tuples of alternating matrices which do not satisfy Equation~\eqref{eq_rank_condition} are associated to stabilizer codes. Thus, one can also use this formalism to study PLC equivalence of stabilizer codes. However, in this case the condition for PLC equivalence is more complicated. Addressing this task is beyond the scope of this paper.

Given a tuple of commutation matrices one can construct a corresponding stabilizer state, as we show in the following. The idea is to go through all parties and transform the respective commutation matrix to a form where it is easy to choose a correct form for the generators. Finally, the transformations are reversed. More precisely, one proceeds as follows. Let us first consider first the case ${\rm rk}([C_\alpha]_{\alpha})=n$ and then deal with the case ${\rm rk}([C_\alpha]_{\alpha})<n$.
\begin{itemize}
    \item[--] Initialize the operators $g_1,\ldots, g_n$ to $\mathds{1}$.
    \item[--] For each party $\alpha$ proceed as follows:
    \begin{itemize}
        \item[$\rightarrow$] Assign ${\rm rk}(C_\alpha)/2$ qubits to party $\alpha$. For this iteration of the loop, these qubits are labeled from $1$ to ${\rm rk}(C_\alpha)/2$.
        \item[$\rightarrow$] Find an invertible matrix $P\in M_{n\times n}(\mathbb{Z}_2)$ such that $P C_\alpha P^T$ is of the form described in Theorem~\ref{thm_dickson_gen}, i.e., a direct sum of the blocks
    \begin{equation*}
    \begin{pmatrix}
    0
    \end{pmatrix},\text{ and }\label{eq_blocks}
    \begin{pmatrix}
    0&1\\
    1&0
    \end{pmatrix}.
    \end{equation*}
    \item[$\rightarrow$] Let $\tilde{g}_1,\ldots ,\tilde{g}_n$ be this new set of generators (obtained as products of the old generators). For all $j\in[n-1]$, if $(P C_\alpha P^T)_{j,j+1}=1$, set $(\tilde{g}_j)_\alpha=X_{(j+1)/2}$ and $(\tilde{g}_{j+1})_\alpha=Z_{(j+1)/2}$. \item[$\rightarrow$] Then undo the change of generators according to $P^{-1}$.
    \end{itemize}
    \item[--] After determining all operators $g_j$ check whether $g_j^2=+\mathds{1}$ or $-\mathds{1}$. In the latter case replace $g_j$ by $i g_j$.

\end{itemize}
First, observe that the rank condition in Eq. \eqref{eq_rank_condition} ensures that the algorithm defines generators acting on $n$ qubits. It is clear that implementing this algorithm for every party, the commutation relations of the operators $g_1,\ldots, g_n$ are described by $(C_\alpha)_{\alpha}$. Thus, the $g_1,\ldots, g_n$ commute. Moreover, they are independent as ${\rm rk}([C_\alpha]_{\alpha})=n$ and, due to the last step of the algorithm, $g_i^2=\mathds{1}$ for all $i\in [n]$. 
Therefore, any element of the generated subgroup is a product of powers zero or one of each generator. If any of these powers is one, the generated operator is not proportional to $\mathds{1}$ as the generators are independent. If all powers are zero, then the operator is equal to $\mathds{1}$ due to the last step of the algorithm. We conclude that $g_1,\ldots, g_n$ generate an $n$-qubit stabilizer corresponding to a stabilizer state.

In case ${\rm rk}([C_\alpha]_{\alpha})<n$, the construction has to be slightly modified. First, one finds an invertible matrix $R$ such that the last $w=n-\text{rk}([C_\alpha]_{\alpha})$ rows of $[RC_\alpha R^T]_{\alpha}$ are zero. Then, one can assign $\text{rk}(C_{\alpha})/2$ qubits to party $\alpha$ for every $\alpha \in P(M,n)$ and determine the generators $g_1,\ldots, g_{n-w}$ as above. The remaining $n-{\rm rk}([C_\alpha]_{\alpha})$ qubits, which are in the state $\ket{0}$, can be distributed arbitrarily among the parties. Let us number the qubits in the state such that those appear last. Then we define $g_{j}=Z_j$ for all $j\in\{n-w+1,\ldots, n\}$. Finally, we undo the change of generators according to $R^{-1}$. It is straightforward to verify that the operators obtained this way define a stabilizer corresponding to a stabilizer state. Again Eq.~\eqref{eq_rank_condition} ensures that the algorithm defines generators acting on $n$ qubits.

Let us discuss the relevance of commutation matrices in the context of PLC transformations. Theorem~\ref{thm_PLC_equiv} (\cite{Br06}) states that two stabilizer states are PLC equivalent if their stabilizers admit sets of generators with the same commutation relations on each party. In terms of commutation matrices this Theorem reads as follows.
\begin{thm}\label{thm_PLC_equiv_com}
	Let $\ket{\psi},\ket{\phi}\in \text{Stab}(P(M,n))$. Then $\ket{\psi}$ is PLC equivalent to $\ket{\phi}$ if and only if they admit the same tuple of commutation matrices.
\end{thm}
As all sets of commutation matrices of a state are related via a congruence transformation, this theorem is equivalent to the following: $\ket{\psi}$ and $\ket{\phi}$ are PLC equivalent if any two tuples of commutation matrices $(C_\alpha^{\psi})_\alpha$, and $(C_\alpha^{\phi})_\alpha$ are congruent to each other. Consequently, a set of commutation matrices and every congruent set uniquely characterizes a PLC equivalence class. Thus, in order to characterize all PLC classes we have to find a maximal set of non-congruent tuples of commutation matrices.

\subsection{The EGS and commutation matrices\label{sec_egs}}

We argued in the previous section that the classification of PLC classes is equivalent to the classification of commutation matrices up to congruence. A common approach to classify tuples of matrices under certain operations is to first establish that these tuples decompose uniquely into indecomposable blocks (up to reordering and equivalence of the blocks). Then, the classification task is solved by presenting a complete set of indecomposable tuples. We show in this section that these two tasks are equivalent to the ones we want to solve for stabilizer states. Namely, we show that stabilizer states decompose uniquely into indecomposable states if and only if this is true for commutation matrices. Moreover, finding the EGS is equivalent to finding a maximal set of non-congruent indecomposable commutation matrices.

Let us start by making more precise the notion of decomposability of stabilizer states and their corresponding tuples of commutation matrices. To that end, consider the tensor product of two stabilizer states $\ket{\psi}\otimes\ket{\varphi}$,  and let $\{g_i\}_i$ be the generators of $\mathcal{S}_\psi$ and $\{h_i\}_i$ the generators of $\mathcal{S}_\varphi$ respectively. Such a state is clearly decomposable. Observe that the two sets of generators are mutually commuting, globally but also locally, as they are acting on different qubits. Thus, their commutation matrices are block-diagonal, where one block corresponds to the generators of $\ket{\psi}$ and the other to the generators of $\ket{\varphi}$, i.e.,
\begin{equation}
C_\alpha(\ket{\psi}\otimes\ket{\varphi})=C_\alpha(\ket{\psi})\oplus C_\alpha(\ket{\varphi}).
\end{equation}

To see that the converse is also true the following theorem establishes the connection between decomposability of stabilizer states and the block-diagonalizability of their commutation matrices.
\begin{thm}\label{thm_block_diag}
	Let $\ket{\psi}\in \text{Stab}(P(M,n))$ and let $(C_\alpha)_{\alpha\in P(M,n)}$ be a tuple of commutation matrices. Then $\ket{\psi}$ is decomposable if and only if there exists an invertible matrix $Q\in M_{n\times n}(\mathbb{Z}_2)$, and matrices $B_\alpha^1\in M_{n_1\times n_1}(\mathbb{Z}_2)$, $B_\alpha^2\in M_{n_2\times n_2}(\mathbb{Z}_2)$ such that $(QC_\alpha Q^T)_{\alpha\in M}=(B_\alpha^1\oplus B_\alpha^2)_{\alpha\in M}$. 
\end{thm}

\begin{proof}
	The proof of the only if part follows directly from Theorem \ref{thm_PLC_equiv} and the fact that all tuples of commutation matrices associated to same state are congruent.
	
	Next, we prove the if part. Suppose there exists a an invertible matrix $Q\in M_{n\times n}(\mathbb{Z}_2)$ such that $(QC_\alpha Q^T)_{\alpha\in P(M,n)}=(B_\alpha^1\oplus B_\alpha^2)_{\alpha}$. As $(B_\alpha^1\oplus B_\alpha^2)_{\alpha}$ is a tuple of commutation matrices, it satisfies Equation~\eqref{eq_rank_condition}. That is,
	\begin{equation}
	2\, {\rm rk}([B_\alpha^1+B_\alpha^2])=\sum_\alpha {\rm rk}(B_\alpha^1\oplus B_\alpha^2)
	\end{equation}
	or equivalently,
	\begin{equation}\label{eq:equation}
	2\,{\rm rk}([B_\alpha^1])+2\, {\rm rk}([B_\alpha^2])=\sum_\alpha {\rm rk}(B_\alpha^1)+\sum_\alpha {\rm rk}(B_\alpha^2).
	\end{equation}
	As $2\,{\rm rk}([B_\alpha^i])\le \sum_\alpha {\rm rk}(B_\alpha^i)$ (see Inequality~\eqref{eq_rank_inequ}) Equation~\eqref{eq:equation} can hold only if $2\,{\rm rk}([B_\alpha^i])= \sum_\alpha {\rm rk}(B_\alpha^i)$. Thus, $(B_\alpha^1)$ and $(B_\alpha^2)$ are tuples of commutation matrices. Let $\ket{\phi_1}$ and $\ket{\phi_2}$ be the corresponding $M$-partite stabilizer states. It follows from Theorem \ref{thm_PLC_equiv} that $\ket{\psi}$ is PLC equivalent to $\ket{\phi_1}\otimes \ket{\phi_2}$ and is therefore decomposable.
\end{proof}
This theorem establishes that decomposing stabilizer states into other stabilizer states is equivalent to decomposing tuples of commutation matrices into blocks. Therefore, to identify states in the EGS we have to identify indecomposable tuples of commutation matrices. The question whether or not stabilizer states decompose uniquely into indecomposable ones is equivalent to whether or not tuples of commutation matrices uniquely decompose into indecomposable blocks.

These questions were studied for general tuples of alternating matrices in linear algebra in the last century. We summarize the main results in the following. For a more detailed description we refer the reader to Appendix \ref{sec_clas_forms}. Let us stress that we are only interested in the classification of a special subset of these tuples, namely those which satisfy Equation \eqref{eq_rank_condition}. 

How much is known about the problem very much depends on the field one considers. For the field $\mathbb{Z}_2$, up to our knowledge, both questions have only been answered for tuples of length $1$ and $2$. The solution for a single matrix is given by Theorem~\ref{thm_dickson_gen} (see Appendix~\ref{sec_sym}), which establishes that the only indecomposable blocks are given by 
\begin{equation}
\begin{pmatrix}
0
\end{pmatrix},\text{ and }\begin{pmatrix}
0&1\\1&0
\end{pmatrix}.
\end{equation}
Theorem~\ref{thm_dickson_gen} also implies that the decomposition into these blocks is unique. This follows from the simple fact that a congruence transformation cannot change the rank of a matrix and therefore the number of blocks of each type is fixed. Comparing the above result to the commutation matrices of the state $\ket{\phi^+}$ in Equation~\eqref{eq_com_mat_bell} one sees that this implies that any bipartite stabilizer state decomposes uniquely into a tensor product of copies of the state $\ket{\phi^+}$ and the state $\ket{0}$.

For tuples of length $2$ Ref.~\cite{Sc76} establishes the uniqueness of the decomposition and states a complete set of indecomposable pairs of alternating forms. To obtain these results the author shows that the classification of pairs of alternating forms is equivalent to the classification of the better known matrix pencils. Due to Equation~\eqref{eq_sum}, these results directly apply to the characterization of the EGS for $3$-partite stabilizer states. In fact, they precisely reproduce the results of Ref.~\cite{Br06} regarding PLC equivalence classes. To see this, one starts from the general classification of Ref.~\cite{Sc76} and considers only those indecomposable tuples which satisfy the rank condition in Equation~\eqref{eq_rank_condition}. For a more detailed discussion we refer the reader to Appendix~\ref{sec_clas_forms}.

Up to our knowledge, the solution for tuples of length $\ge 3$ for the field $\mathbb{Z}_2$ is not known. However, for algebraically closed fields of characteristic $\neq 2$ it is shown in Ref.~\cite{Se05} that the classification of triples of alternating forms is a \emph{wild} problem, i.e., it is at least as hard as classifying pairs of matrices up to simultaneous similarity. The latter task is considered to be very challenging. This suggests that also in the case of the field $\mathbb{Z}_2$, classifying general triples of alternating forms is difficult. However, due to Equation~\eqref{eq_rank_condition} we are only interested in very special triples of matrices and thus a complete characterization of indecomposables tuples might be possible beyond pairs of matrices. 

For the field $\mathbb{Z}_d$ with $d\ge 3$ the situation is different. Here, Ref.~\cite{Se88} establishes the uniqueness of the decomposition for arbitrary tuples of alternating matrices (cf. Theorem 2 in Ref.~\cite{Se88}, see also Appendix~\ref{sec_clas_forms}). As mentioned earlier, the commutation matrix formalism generalizes to qudit stabilizer states for prime dimension. In this case, the commutation matrices are defined over the field $\mathbb{Z}_d$. Therefore, by establishing the equivalence between the decomposition of commutation matrices and the decomposition of stabilizer states (see Section~\ref{sec_qudit_systems}) we show via the results of Ref.~\cite{Se88} that qudit stabilizer states decompose uniquely into indecomposable ones under PLC (see Section~\ref{sec_qudit_systems}).

We comment in Appendix~\ref{sec_uniqueness} on our approach to show the uniqueness of the decomposition for qubit systems. We prove that stabilizer states without PLC symmetries decompose uniquely. However, these symmetries seem to be common, at least for four parties (see the states in Figure~\ref{fig:EGS8}) as many states in the EGS$_4$ up to $10$ qubits contain leafs, which yield additional symmetries.

\subsection{Necessary and sufficient conditions for decomposability\label{sec_nec_suf_deco}}
In this section we introduce useful tools for studying decompositions of stabilizer states. In particular, we focus on the application of the commutation matrix formalism. We discuss Fitting's Lemma (see, e.g., Ref.~\cite{Se88}), which provides necessary and sufficient conditions for decomposability of tuples of alternating matrices, and in particular of commutation matrices.

First, let us present a lemma which gives sufficient conditions for the extraction of the state $\ket*{GHZ_{M-1}}$ or $\ket*{GHZ_{M}}$ from an $M$-partite stabilizer state.
\begin{lem}\label{lem_extract_ghz}
	Let $P(M,n)=\{\alpha_1,\ldots ,\alpha_{M}\}$ be a partition and let $\ket{\psi}\in \text{Stab}(P(M,n))$ have full local ranks, i.e., $\text{dim}(V^\mathcal{S}_\alpha)=0$ for all $\alpha\in P(M,n)$. Then one can extract either $\ket*{GHZ_{M-1}}$ or $\ket{GHZ_{M}}$ from $\ket{\psi}$ if and only if there exist elements $f_1,\ldots,f_{M-2}$ in $\mathcal{S}$ such that $\text{supp}(f_j)=\alpha_1\cup \alpha_{j+1}$ and $(f_j)_{\alpha_1}=(f_k)_{\alpha_1}$ for all $j,k\in [M-2]$.
\end{lem}
The proof of this lemma can be found in Appendix~\ref{sec_proof_lem_ghz}. It makes use of the fact that $\mathcal{S}$ is a maximally isotropic subspace. This fact in combination with the special operators $f_j$ it has to contain is enough to show that the conditions of Theorem~\ref{thm_PLC_extract} are satisfied and an extraction is possible~\footnote{Note that Ref.~\cite{Wi06} also contains conditions under which one can extract a GHZ state from a given stabilizer states, even between less than $M-1$ parties.}.

Next, let us discuss how one can analyze decompositions of tuples of commutations (and thereby decompositions of the corresponding stabilizer states). It is a well-known fact that to study decompositions of a tuples of alternating matrices into blocks one can equivalently study the ring of self-adjoint endomorphisms of the corresponding tuple. A self-adjoint endomorphism of a tuple of alternating matrices $(C_\alpha)_{\alpha\in P(M,n)}$, $C_\alpha\in M_{n\times n}(\mathbb{Z}_2)$ is a matrix $E\in  M_{n\times n}(\mathbb{Z}_2)$ such that 
\begin{equation}\label{eq_endo}
C_\alpha E=E^TC_\alpha
\end{equation}
holds for all $\alpha\in P(M,n)$. Self-adjoint endomorphisms enjoy many useful properties that we will frequently make use of in what follows. It directly follows from the defining Equation~\eqref{eq_endo} that the set of self-adjoint endomorphisms is closed under products and additions and therefore forms a ring. In particular, whenever $E$ is a self-adjoint endomorphism, so are its powers $E^m$. It is straightforward to verify that for any $m\in \mathbb{N}$ it holds that $\text{ran}(E^m)\supseteq \text{ran}(E^{m+1})$, and $\text{ker}(E^m)\subseteq \text{ker}(E^{m+1})$. For invertible endomorphisms the range coincides with the full vector space for any power of $E$, and for nilpotent endomorphisms one finds $k\in\mathbb{N}$ such that $E^k=0$. For finite dimensional vector spaces, and in particular over finite fields, one can always find an $l\in \mathbb{N}$ for which $\text{ran}(E^l)= \text{ran}(E^{l+1})$, and $\text{ker}(E^l)= \text{ker}(E^{l+1})$, that is $E^l$ acts as a bijection on its range. Over finite fields, $\text{ran}(E^l)$ is a finite set and, therefore, $E^l$ is a permutation on it. Therefore, there exists a power $p\geq l$ for which $E^p$ acts as the identity on its range and thus fulfills $(E^p)^2=E^p$, i.e., $E^p$ is \emph{idempotent}. The existence of idempotent elements that are neither invertible nor nilpotent plays a crucial role in deciding if a set of commutation matrices is decomposable. The following lemma, known as Fitting's lemma, explains how properties of self-adjoint endomorphisms relate to the decomposition of $(C_\alpha)_{\alpha\in P(M,n)}$ into blocks.

\begin{lem}[Fitting's lemma~\cite{Se88}]\label{lem_fitting}
	Let $F$ be any field, let $I\subset\mathbb{N}$ be a finite index set and let $(A_\alpha)_{\alpha\in I}$, $A_\alpha\in M_{n\times n}(F)$ be a tuple of alternating matrices. Then, $(A_\alpha)_{\alpha\in I}$ is indecomposable under congruence if and only if all self-adjoint endomorphisms of $(A_\alpha)_{\alpha\in I}$ are either nilpotent or invertible.
\end{lem}

Let us include a proof of the lemma for finite fields as it illustrates its application in the context of commutation matrices.

\begin{proof}
	Let $(A_i)_i$ be a tuple of alternating $n\times n$ matrices. As $F$ is finite, for any endomorphism of $(A_i)_i$ there exists a $k\in\mathbb{N}$ such that $E^{2k}=E^k$ (see the discussion above Lemma~\ref{lem_fitting}). Thus, there exists an invertible matrix $R$ such that $E^k=R^{-1}(\mathds{1}\oplus 0)R$. It is straightforward to verify that $\mathds{1}\oplus 0$ is a self-adjoint endomorphism of $(RA_iR^T)$. We conclude that if  $\text{ker}(E^k)\neq F^n$ or $\text{R}(E^k)\neq F^n$, then $(RA_iR^T)_i$ is blockdiagonal.  
\end{proof}

Hence, they are simultaneously decomposable if and only if there exists an idempotent endomorphism which is neither invertible nor nilpotent~\footnote{Note that Fitting's lemma is usually stated for decompositions into blocks under isomorphisms, i.e., under transformations of the form $QC_\alpha R$, where $Q,R$ are invertible (see for instance Ref.~\cite{Ba}). However, for our purposes the version presented in Ref.~\cite{Se88} is the relevant one, as it is formulated for congruence.}

Lemma~\ref{lem_fitting} is particularly useful if one wants to numerically decide whether or not a stabilizer state is decomposable. As Equation~\eqref{eq_endo} is linear in the entries of $E$, it is efficient to compute a basis of the endomorphism ring. The endomorphism ring generated by this basis is finite as we work over the finite field $\mathbb{Z}_2$ and there are only finitely many matrices over $\mathbb{Z}_2$. Thus, if the ring is not too large one can check whether every element is either nilpotent or invertible, or equivalently, if there exists a idempotent element that is neither invertible nor nilpotent. As we will see later, this insight is very useful for numerically computing the EGS.

\subsection{The EGS$_M$ for $M\ge 4$\label{sec_egs4}}
We explained in the previous sections that to characterize the PLC classes of $M$-partite stabilizer states we can characterize the EGS$_{M}$ and establish that $M$-partite stabilizer states decompose uniquely into indecomposable ones. In this section we are concerned with the structure of the EGS$_M$ for $M\ge 4$. We show that the EGS$_4$ and therefore any EGS$_M$ for $M\ge 4$ is an infinite set. This is even the case if we consider PLU instead of PLC transformations.

Let us show that the EGS$_4$ is an infinite set w.r.t. PLC transformations. Consider the spiral graph states $\ket{G_n}$ defined for any $n\in\mathbb{N}$, $n\ge 4$, in Figure \ref{fig_fig1}. These states are distributed among the parties $P(4,n)=\{\alpha_1,\alpha_2,\alpha_3,\alpha_4\}$.

\begin{figure}[t!]
	\centering\includegraphics[width=0.8\columnwidth]{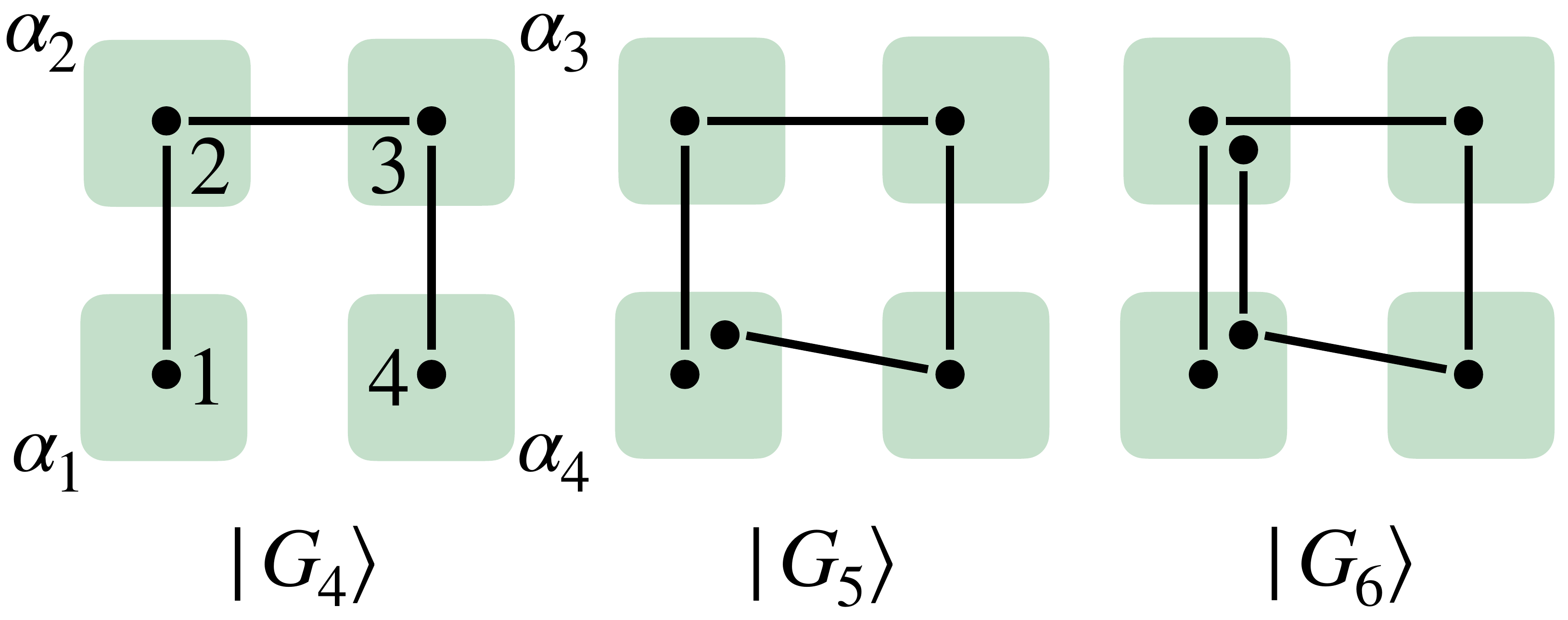}
	\caption{Sequence of the spiral graph states $\ket{G_n}$ and their distribution among four parties $P(4,n)=\{\alpha_1,\alpha_2,\alpha_3,\alpha_4\}$; all of these graph states are indecomposable.\label{fig_fig1}}
\end{figure}

\begin{thm}\label{thm_egs_infinite}
	For any $n\in\mathbb{N}$ the state $\ket{G_n}$ is indecomposable.
\end{thm}
We show this Theorem in Appendix~\ref{app:spiralprop}. The proof makes use of the fact that if a tuple of linear combinations of commutation matrices is indecomposable, then so is the original tuple. Using the complete classification of indecomposable pairs of alternating forms from Ref.~\cite{Sc76} it is straightforward to verify that if $(C_{\alpha_1}^{(n)},C_{\alpha_2}^{(n)},C_{\alpha_3}^{(n)},C_{\alpha_4}^{(n)})$ are the commutation matrices of the spiral graph state $\ket{G_n}$, then $(C_{\alpha_1}^{(n)}+C_{\alpha_2}^{(n)},C_{\alpha_2}^{(n)}+C_{\alpha_3}^{(n)})$ is indecomposable for any $n$. As the EGS$_4$ contains one representative of every PLC class which only consists of indecomposable states, $\ket{G_n}\in\text{EGS}_4$ for any $n$ and thus EGS$_4$ contains infinitely many states. Moreover, as $\text{EGS}_4\subset \text{EGS}_M$ for any $M\ge 4$, we have the following corollary.
\begin{cor}\label{corr_egs_infinite}
	The EGS$_M$ is infinite for any $M\geq 4$.
\end{cor}
Thus, by simply increasing the number of parties from $3$ to $4$ the EGS changes from a finite set of three states (up to permutation) to an infinite set.

Let us address the question whether the EGS is infinite only due to considering PLC operations instead of more general PLU operations. Indeed, it could be that the EGS is finite when considering PLU operations. The following discussion, however, shows that this is not the case. By studying which additional transformations of the spiral graph states $\ket{G_{4n}}$ are enabled by PLU operations compared to PLC operations one realizes that they are equally powerful, as stated in the following theorem. A proof can be found in Appendix~\ref{app:spiralprop}.
\begin{thm}\label{thm_plu_plc}
	Let $U$ be a PLU transformation such that $U\ket{G_{4n}}=\ket{H_{4n}}$ is a stabilizer state. Then, there exists a PLC transformation $C\in\mathcal{C}_{4n}^{P(4,4n)}$ such that $C\ket{G_{4n}}=\ket{H_{4n}}$.
\end{thm}
The theorem relies on properties of the PLU symmetry group of the spiral graph states $\ket{G_{4n}}$. Note that any PLU transformation between stabilizer states defines a bijection between their respective PLU symmetry groups. Theorem \ref{thm_plu_plc} implies that the spiral graph states $\ket{G_{4n}}$ are indecomposable under PLU operations and therefore are elements of the EGS$_M$ w.r.t. PLU operations for any $M\ge 4$. Thus, stabilizer states on $M\geq 4$ parties contain infinitely many different types of stabilizer state entanglement.

\subsection{The EGS$_4$ up to $10$ qubits\label{subsec:EGS}}

\begin{figure*}
	\centering
	\includegraphics[width=0.7\textwidth]{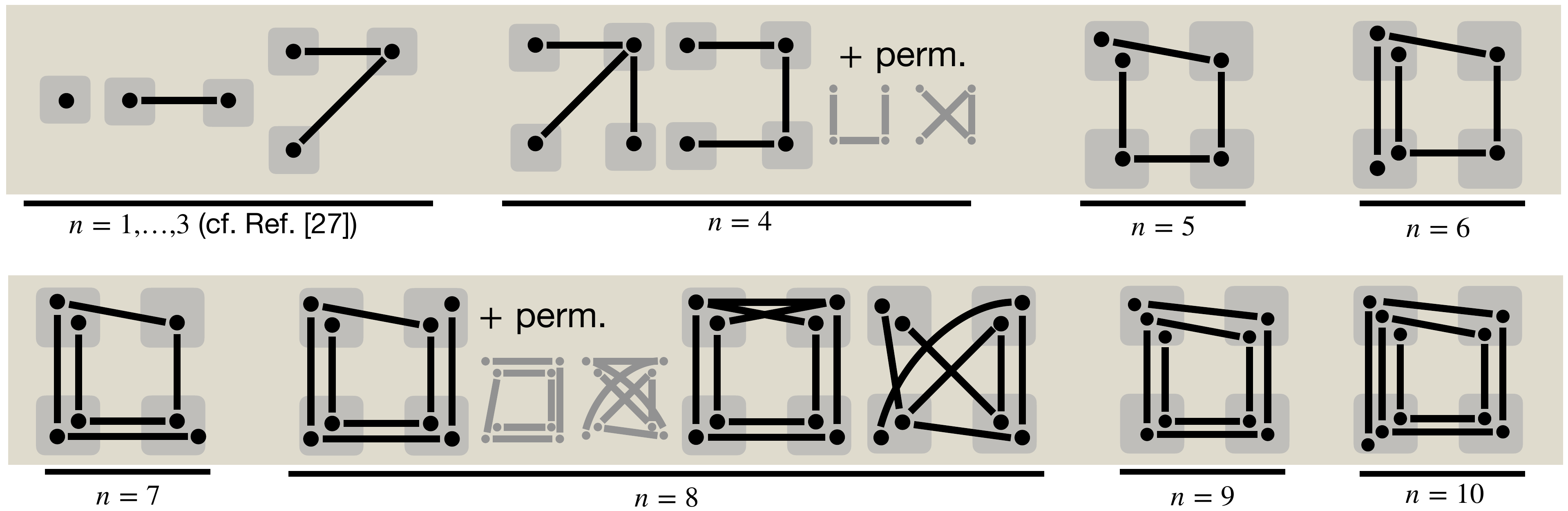}
	\caption{{\bf EGS$_4$ for up to $10$ qubits.} The first three states generate all stabilizer states on two- and three-partite qubit stabilizer states~\cite{Fa04,Br06}. The EGS$_4$ on four qubits contains the GHZ state as well as the linear cluster state (up to permutations) with different positions of its leaf parent pairs. For five, six and seven qubits one finds that only the PLC classes of spiral graphs are indecomposable. For eight qubits, the spiral graph appears three times with different positions of its leaf parent pairs, while the other two graphs have a structure different from the spiral graph. Particularly the last graph for eight qubits has an interesting structure, since it contains three leaf parent pairs. Finally, for nine and ten qubits only the PLC classes of the spiral graphs are indecomposable.}
	\label{fig:EGS8}
\end{figure*}

The results from the previous section suggest that a complete characterization of all indecomposable stabilizer states for an arbitrary number of parties is a difficult task. Nevertheless, let us turn to the simplest non-solved case, namely EGS$_4$, and present some results on the states it contains. Intuitively, we expect that not every distribution of qubits among the $4$ parties allows for the existence of an indecomposable state. Indeed, it can be shown that any state, where one party contains more than twice the number of qubits of the party with the smallest number of qubits, is decomposable, as stated by the following observation.
 
\begin{obs}\label{thm_bounds_egs}
	Let $P(4,n)=\{\alpha_1,\alpha_2,\alpha_3,\alpha_4\}$ be a partition such that $|\alpha_1|\ge|\alpha_2|\ge|\alpha_3|\ge|\alpha_4|$ and $\ket{\psi}\in\text{Stab}(P(M,n))$. Then $\ket{\psi}$ is decomposable if $|\alpha_1|>2|\alpha_4|$.
\end{obs}

The bound is simple consequence of Lemma~\ref{lem_extract_ghz}, which gives sufficient conditions for when one can extract GHZ states from a stabilizer state. Note that this bound is not tight, as we will see later, i.e., there exist qubit configurations for which all states are decomposable which is not implied by the above observation. This is not surprising as Observation~\ref{thm_bounds_egs} does not take into account whether $4$-partite stabilizer states other than the GHZ state can be extracted. Now let us compute the EGS$_4$ up to $10$ qubits. We explain in detail how we utilize the commutation matrix formalism in this computation~\footnote{The code is written for Python 3 and can be made available from the authors upon reasonable request.}.

From Observation~\ref{thm_bounds_egs} it is clear that for up to $9$ qubits only qubits distributed among parties specified by the tuples $(1,1,1,1)$, $(2,1,1,1)$, $(2,2,1,1)$, $(2,2,2,1)$, $(2,2,2,2)$, and $(3,2,2,2)$ can lead to indecomposable stabilizer states. To find the maximal set of indecomposable PLC inequivalent set of stabilizer states we proceed as follows. First, for each configuration containing up to $7$ qubits we generate all possible graph states and remove the decomposable ones. To this end, we compute a tuple of commutation matrices for each state and then check decomposability by computing a basis of the endomorphism ring and checking whether or not there exists an idempotent element that is neither nilpotent nor invertible. According to Fitting's Lemma~\ref{lem_fitting} the existence of such an element is equivalent to the set of commutation matrices, and hence, the state, being decomposable.

It remains to identify one representative of each PLC equivalence class from the obtained list of indecomposable graph states. We select a state from the list and then remove all states which are PLC equivalent to the chosen state. To check PLC equivalence we proceed as follows. Recall that according to Theorem~\ref{thm_PLC_equiv_com} two states are PLC equivalent if they admit congruent tuples of commutation matrices. Thus, given two graph states $\ket{G'}$ and $\ket{G'}$ with commutation matrices $(C_\alpha)_{\alpha}$ and $(D_\alpha)_{\alpha}$ respectively, we have to decide whether $(C_\alpha)_{\alpha}$ and $(D_\alpha)_{\alpha}$ are congruent to each other or not. This can be done by solving the linear system of equations
\begin{equation}
(QC_\alpha)_{\alpha}=(D_\alpha P)_{\alpha}\label{eq_check_congruence}
\end{equation}
for the entries of $Q$ and $P$. Then, we check if any solution satisfies $Q=(P^{-1})^T$. The runtime for both, verifying indecomposability and removing congruent tuples, can in principle be exponentially large in the number of qubits. This is due to the fact that the solution space of the respective linear equation system can have a basis of $O(n^2)$ elements, and hence the solution space contains $O(2^{n^2})$ elements. Nevertheless, the algorithm seems to be fast for most graph states in practice.

For the cases $(2,2,2,2)$ and $(3,2,2,2)$ containing eight and nine qubits we have to follow a different approach as the number of graph states of eight qubits is already too large ($\approx 268$ million). Therefore, we start from a complete list of local Clifford inequivalent states of eight and nine qubits (see~\cite{Ca00}) and consider all possible distributions of these states among $4$ parties resulting in the configuration $(2,2,2,2)$, and $(3,2,2,2)$ respectively. Then, we proceed as before by removing all decomposable states and determining a state for each PLC class.

For $10$ qubits we proceed in a similar way as for eight and nine qubits but now we have to consider qubits distributed as $(3,3,2,2)$ and $(4,2,2,2)$. Note, that the latter distribution is not excluded to contain indecomposable states by Observation~\ref{thm_bounds_egs}. Nevertheless, we find that it contains only decomposable states. For the distribution $(3,3,2,2)$ only the PLC class of the ten qubit spiral graph state is indecomposable.

In Figure~\ref{fig:EGS8} we display all states in the EGS$_4$ up to $10$ qubits up to relabeling of parties. Remarkably, up to $7$ qubits, the only additions to the EGS$_4$ besides the $4$-qubit GHZ state are the spiral graph states. For $8$ qubits more complex indecomposable graph states exist. In particular, the last graph for eight qubits in Figure~\ref{fig:EGS8} is different, since it contains three leaf parent pairs. Finally, for $9$ and $10$ qubits the only additional states are again the spiral graphs.

Going beyond the EGS$_4$ we find that for $6$ qubits on 5 parties there exist $19$ different PLC classes (or $10$ different classes up to permutations of the parties holding a single qubit). This suggests that beyond the EGS$_4$ the number of different PLC classes increases very quickly, already in the case where the EGS$_4$ for $5$ qubits still has a very simple structure, i.e., where it contains only a single PLC class (see Figure~\ref{fig:EGS8}).

\section{Qudit systems \label{sec_qudit_systems}}
In this section, we generalize the commutation matrix formalism to qudit stabilizer states of prime dimension. Similar to the qubit case, non-local qudit Clifford operations can be implemented deterministically using gate teleportation. Consequently, it is again interesting to study the same task of characterizing PLC equivalence classes of qudit stabilizer states.

We first recall the definition of the Pauli group and stabilizer states for qudits. We will closely follow Refs.~\cite{Ho05, Fa14}. In the following, $d\in\mathbb{N}$ is prime and $d\ge 3$. We consider the common generalization of the Pauli operators to qudit systems as
\begin{align}
X&=\sum_{k=0}^{d-1}\ket{k+1}\bra{k},\\
Z&=\sum_{k=0}^{d-1}\eta^k\ket{k}\bra{k},
\end{align}
where $\eta=\exp(2\pi i/d)$ so that $X^d=Z^d=\id$. In analogy to the single-qubit Pauli group the single-qudit Pauli group is defined as
\begin{equation}
\mathcal{P}^d=\{\eta^m \sigma_{a,b}|a,b, m\in \mathbb{Z}_{d}\},
\end{equation}
where $\sigma_{a,b}=X^aZ^b$. As $d$ is prime, all elements in the single-qudit Pauli group have the same order, namely $d$, and the same set of non-degenerate eigenvalues. The $n$-qudit Pauli group $\mathcal{P}_n^d$ is the group generated by all $n$-fold tensor products of elements of $\mathcal{P}^d$.
In analogy to the qubit case, there is a correspondence between the Pauli group and the vector space $\mathbb{Z}_d^{2n}$ with a symplectic form $\omega$. Let us again define the map $\sigma:\mathbb{Z}_d^{2n}\rightarrow \mathcal{P}_n^d$ via $\sigma(a_1,b_1,\ldots,a_n,b_n)=\sigma_{a_1,b_1}\otimes\ldots\otimes\sigma_{a_n,b_n}$. Under this map, the group structure of the Pauli group is preserved. More precisely, from the relation
\begin{equation}
\sigma(\vec{f})\sigma(\vec{g}) \propto \sigma(\vec{f}+\vec{g})\label{eq_eq3}
\end{equation}
we see that elements of the Pauli group modulo phases can be represented as vectors of $\mathbb{Z}_d^{2n}$, and the multiplication of Pauli operators corresponds to addition of their corresponding vectors in $\mathbb{Z}_d^{2n}$. The generalized Pauli operators obey the following commutation relations
\begin{equation}\label{eq:comm_qudits}
\sigma(\vec{f})\sigma(\vec{g})=\eta^{\omega(\vec{f},\vec{g})}\sigma(\vec{g})\sigma(\vec{f})\ \ \vec{f},\vec{g}\in\mathbb{Z}_d^{2n},
\end{equation}
with the symplectic form
\begin{align}
\omega(\vec{f},\vec{g})=(\vec{a},\vec{b})\begin{pmatrix}
0&  \mathds{1}\\
-\mathds{1}&0
\end{pmatrix}(\vec{a'},\vec{b'})^T.
\end{align}
The Clifford group $\mathcal{C}_n(d)$ is defined as the unitary normalizer of the Pauli group up to phases, i.e., $\mathcal{C}_n(d)=\{U\in SU(d^n)|UPU^\dagger \in\mathcal{P}_n^d\ \forall P\in \mathcal{P}_n^d\}/U(1)$. Every $U\in \mathcal{C}_n(d)$ defines an isometry on the symplectic space $\mathbb{Z}_d^{2n}$ via Equation~\eqref{eq_eq4} and vice versa. 

A stabilizer state is the unique $+1$ eigenstate to a maximal abelian subgroup of $\mathcal{P}_n^d$ which does not contain any nontrivial multiples of the identity. Such a subgroup contains $d^n$ elements and is generated by $n$ independent elements. An abelian subgroup of $\mathcal{P}_n^d$ which does not contain any nontrivial multiples of the identity is called a stabilizer. In complete analogy to the qubit case, a stabilizer associated to a stabilizer state corresponds to a maximally isotropic subspace of $\mathbb{Z}_d^{2n}$ via the map $\sigma$. Moreover, every qudit stabilizer state is local Clifford equivalent to a qudit graph state~\cite{Mo06}. A qudit graph state is the unique stabilizer state associated to a multigraph $G$ with vertices $V$ and edges $E$. In contrast to a simple graph, multigraphs allow for an arbitrary number of edges between two vertices. Similar to the qubit case, edges correspond to controlled-$Z$ interactions between qudits, which are represented by the vertices of the graph. The existence of distinct non-trivial powers of the gate $CZ=\sum_{k=0}^{d-1} \ketbra{k}\otimes Z^k$ gives rise to the different numbers of edges between vertices. The number of edges $m$ between two vertices is called the \emph{multiplicity}, and the corresponding interaction is given by $CZ^{m}$, where $0\leq m \leq d-1$. Given a multigraph the canonical generators of the corresponding stabilizer are the operators
\begin{equation}\label{eq:quditgen}
g_i=X_i\bigotimes_{j\in N_i} Z_j^{m_{ij}},
\end{equation}
where $m_{ij}$ is the multiplicity of the edges connecting vertex $i$ and $j$. Finally, the adjacency matrix of a graph state $\ket{G}$ is defined by
\begin{equation}\label{eq:qudit_adjacency}
    (\Gamma_G)_{ij} = m_{ij}.
\end{equation}
In the following we denote by $\text{Stab}_d(P(M,n))$ the set of $n$-qudit stabilizer states for qudits of dimension $d$, which are distributed among $M$ parties according to the partition $P(M,n)$.

The $n$-qudit Pauli group $\mathcal{P}_n^d$ has an analogous relation to the symplectic vector space $(\mathbb{Z}_d^{2n},\omega)$ as $\mathcal{P}_n$ has to the symplectic vector space $(\mathbb{Z}_2^{2n},\omega)$. It is therefore straightforward to define commutation matrices via the same construction as in Section~\ref{sec_plc_trafo_stab} for qudit stabilizer states. Commutation matrices are now tuples of matrices $(C_\alpha)_{\alpha\in P(M,n)}$ such that $C_\alpha\in\ M_{n\times n}(\mathbb{Z}_d)$. The matrices are again alternating due to the fact that $\omega(\vec{f},\vec{g})=-\omega(\vec{g},\vec{f})$ and satisfy Equation~\eqref{eq_sum}, i.e., they sum up to zero modulo $d$.

Before discussing the other properties of commutation matrices, let us consider an example. Consider a four qutrit state distributed over $M=4$ parties and let  the generator of its stabilizer be $g_1=X_1Z_2Z_4$, $g_2=Z_1X_2$, $g_3=X_3Z_4^2$, and $g_4=Z_1Z_3^2X_4$. One can then verify that it admits the following set of commutation matrices.
\begin{equation}\label{eq:comm_example}
\begin{split}
    C_1 &=\begin{pmatrix}
0&2&0&2\\1&0&0&0\\0&0&0&0\\1&0&0&0
\end{pmatrix}
C_2=\begin{pmatrix}
0&1&0&0\\2&0&0&0\\0&0&0&0\\0&0&0&0
\end{pmatrix} \\
C_3 &=\begin{pmatrix}
0&0&0&0\\0&0&0&0\\0&0&0&1\\0&0&2&0
\end{pmatrix}
C_4=\begin{pmatrix}
0&0&0&1\\0&0&0&0\\0&0&0&2\\2&0&1&0
\end{pmatrix}.
\end{split}
\end{equation}
These matrices are alternating as $-1=2$ over $\mathbb{Z}_3$, and they sum up to zero modulo $3$. In fact, this state is again a graph state with edges $\{\{1,2\}_1,\{1,4\}_1,\{3,4\}_2\}$, where the subscript denotes the multiplicity of its edges. One can again verify that the commutation matrices can be obtained from the adjacency matrix $\Gamma_G$ via
\begin{equation}
    C_\alpha=\sum_{i\in\alpha} \Gamma_G(i),
\end{equation}
where $\Gamma_G(i)$ is the matrix that contains the $i$-th column and $(-1)$ times the $i$-th row of $\Gamma_G$, and zero everywhere else. E.g., for the the above example one has
\begin{equation}
    \Gamma_G=\begin{pmatrix}
0&1&0&1\\1&0&0&0\\0&0&0&2\\1&0&2&0
\end{pmatrix},
\end{equation}
from which the commutation matrices in Equation~\eqref{eq:comm_example} directly follow.

Let us continue our discussion on properties of commutation matrices. As stabilizers again correspond to maximally isotropic subspaces of $\mathbb{Z}_d^{2n}$, it is straightforward to verify that the rank condition in  Theorem~\ref{thm_rank_condition} still holds. Moreover, commutation matrices still have the same interpretation concerning PLC equivalence of stabilizer states. In fact, a theorem similar to Theorem~\ref{thm_PLC_equiv_com} in the qubit case still holds in the case of qudits.

\begin{thm}\label{thm_1_qudits}
	Let $\ket{\psi},\ket{\phi}\in \text{Stab}_d(P(M,n))$. Then $\ket{\psi}$ is PLC equivalent to $\ket{\phi}$ if and only if they admit the same tuple of commutation matrices.
\end{thm}
The proof closely follows Ref.~\cite{Br06} and can be found in Appendix~\ref{sec_proof_plcequiv_qudits}.

\begin{figure*}
	\centering
	\includegraphics[width=0.9\textwidth]{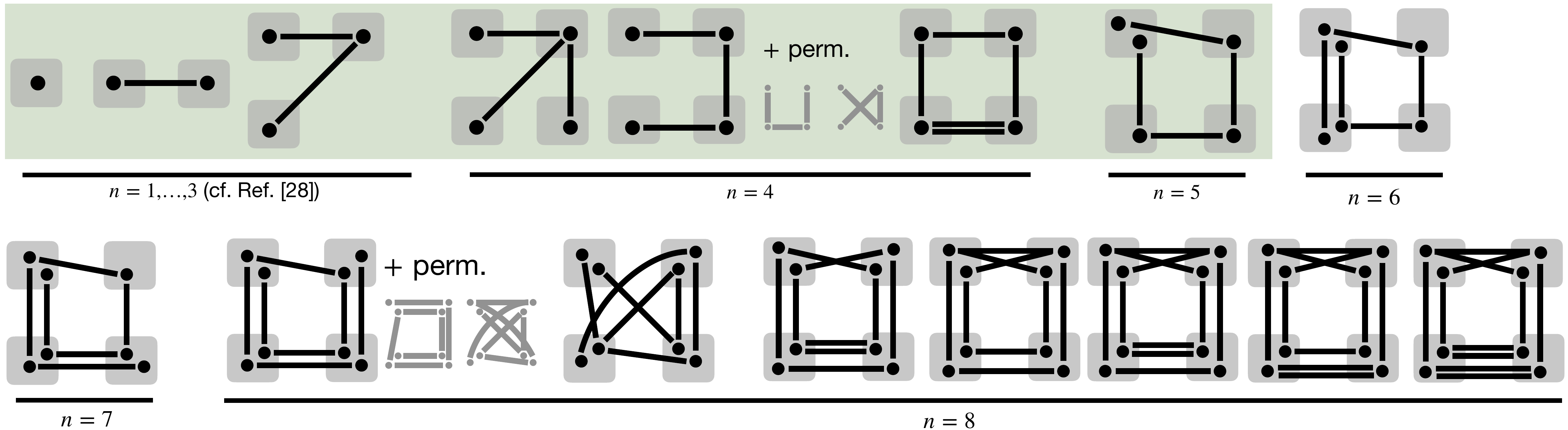}	
\caption{{\bf States in the EGS$_4$ for qutrits.} Highlighted in green is the complete EGS$_4$ for up to five qutrits. The first three states generate all two- and three-partite qutrit stabilizer states~\cite{Gr11}. An exhaustive search through all four-partite graph states up to 5 qutrits reveals six more states from the EGS$_4$. Beyond 
$5$ qutrits we have identified some states that are contained in the EGS$_4$ for qutrits, however, this list might be incomplete. The EGS$_4$ for qutrits already has a much richer structure, due to the appearance of graphs that contain edges of higher multiplicities.}\label{fig_fig4}
\end{figure*}

It directly follows that a qudit stabilizer state is decomposable if and only if its commutation matrices can be block-diagonalized under congruence (cf. Theorem \ref{thm_block_diag}). This is due to the fact that Equation~\eqref{eq_rank_condition} and Inequality~\eqref{ineq_rank_condition_general} still hold. Thus, we again have that finding states in the EGS is equivalent to finding indecomposable tuples of commutation matrices. Stabilizer states decompose uniquely into indecomposable ones if commutation matrices decompose uniquely into indecomposable blocks. In contrast to qubit systems, the commutation matrices are now defined over the finite field $\mathbb{Z}_d$ with $d\ge 3$, i.e., a field with characteristic not equal to two. Theorem $2$ of Ref.~\cite{Se88} shows that tuples of alternating matrices over such fields indeed decompose uniquely into indecomposable blocks. Thus, we have shown the following theorem.
\begin{thm}
All states in $\text{Stab}_d(P(M,n))$ decompose uniquely into indecomposable states in the EGS$_M$, where $d\geq 3$ is a prime.
\end{thm}
Due to this theorem, any PLC class is uniquely characterized by a tensor product of states from the EGS. Therefore, the remaining task is to determine the EGS.

First, however, some remarks are in order. Let us stress that (qubit and qudit) stabilizer states distributed among parties (where each party can hold more than one qubit) can be described by a subset of higher-dimensional qudit stabilizer states with a single qudit per party and a composite dimension (cf. Ref.~\cite{Kr18}). PLC operations then form a subset of qudit Clifford transformations. On the other hand, viewing qudit stabilizer states of non-prime dimension as lower-dimensional weighted graph states~\cite{Kr18} also reveals that for qudit graph states of non-prime dimension the decompositions under PLU are not necessarily unique. Consider a four ququart state on two parties that decomposes into a maximally entangled state $\ket{\phi^+_4}$, containing 2 ebits of entanglement, and two $\ket{0}$ states. Using some of the results of Ref.~\cite{Kr18} one can see that this decomposition is PLU equivalent to a decomposition consisting of two entangled pairs of ququarts that each contain only one ebit of entanglement. Thus, there exist to decompositions that are PLU equivalent, but the states that occur in these decompositions are PLU inequivalent. Note that qudit stabilizer states and qudit Clifford transformations were also investigated in Refs.~\cite{St17,Ri18}.

Moreover, let us remark that, similar to the qubit case, the results of Scharlau~\cite{Sc76} cannot lead to anything else than the EGS$_3$ for prime dimensions. In particular, they cannot be used to construct the EGS$_4$ since choosing two commutation matrices to be indecomposable but not equal to the $\ket{\phi^+_3}$ or the $\ket{GHZ_3^3}$ one cannot find two additional commutation matrices so that the overall rank constraint is fulfilled.

\subsection{States in the EGS$_4$ for qutrits}

Using the same techniques as in the case of qubits in Section~\ref{sec_nec_suf_deco}, i.e., Fitting's Lemma~\ref{lem_fitting}, we can compute the EGS$_4$ for up to $5$ qutrits by performing an exhaustive search through all $4$-partite graph states up to $5$ qutrits. First, one recovers the states $\ket{0}$, $\ket{\phi^+_3}$, and $\ket{GHZ_3}$ that where already proven in Ref.~\cite{Gr11} to generate all three-partite stabilizer states. This can again be derived from the complete set of indecomposable pairs of alternating forms in Ref.~\cite{Sc76}, as these results hold for any field. Additionally, we identify the $6$ additional states from the EGS$_4$ up to $5$ qutrits, see Figure~\ref{fig_fig4}.

Beyond $5$ qutrits the exhaustive search through all graph states becomes computationally intractable due to the increasing number of graphs but also due to the increasing sizes of the endomorphism rings. Nevertheless, one can gain some additional insight into the EGS$_4$ for more than $5$ qutrits. To that end, let us first see if the states that we found to be in the EGS$_4$ for qubits also appear in the case of qutrits. Again employing Lemma~\ref{lem_fitting} (Fitting's Lemma), one can show that the spiral graphs containing only edges of multiplicities $1$ on $5$ to $10$ qubits are also indecomposable and thus in the EGS$_4$ for qutrits. This raises the question if other spiral graph states containing edges of other multiplicities are also contained in the EGS$_4$. In the following we argue that this is not the case, as for many states edges of higher multiplicities can be removed by local Clifford operations.

To this end, let us recall how the action of the local Clifford group can be described on the level of graphs. The action of the local Clifford group on qudit graph states was completely characterized in Ref.~\cite{Mo06}. In particular, it was shown that there exists a local Clifford operation that, when applied to any vertex $v$, multiplies the multiplicities of all edges connected to $v$ by a constant factor $0
\neq b\in\mathbb{Z}_d$. Thus, for qutrit graph states, there exists a local Clifford operation that converts all edges of multiplicity $2$ to edges of multiplicity $1$ and vice versa. Recall, that a \emph{tree graph} is a graph that does not contain any cycle. Then, we have the following observation.
\begin{obs}
Any PLC orbit that contains a tree graph, also contains the same tree graph, where all edges are of multiplicity $1$.
\end{obs}
This observation simply follows from the fact that in a tree graph all edges of higher multiplicities can be moved to a leaf and can then be locally converted to edges of multiplicities $1$. Notably, this means that any qutrit spiral graph state, regardless of the multiplicities of its edges, is always PLC equivalent to the spiral graph states with all multiplicities being equal to $1$.

Moreover, the graph state containing three leafs is indecomposable and edges of higher multiplicities play no role by the arguments above. The only state where edges of higher multiplicities play a role is the one containing cycles. It is again indecomposable for qutrits and appears in four PLC inequivalent forms with different numbers of edges with multiplicity $2$.

Finally, as in the case of four qutrits, the closed loop graph containing a single edge of multiplicity two is indecomposable. From these few states one can already deduce, that the EGS$_4$ for qutrits already has a much richer structure, due to the appearance of graphs that contain edges of higher multiplicities, see Figure~\ref{fig_fig4}.

\subsection{Beyond PLC transformations\label{sec:beyond}}

Let us finally briefly comment on how (stabilizer) states transform under more general types of transformations, also constrained by locality.

First, let us consider general PLU transformations (not just PLC). In Ref.~\cite{Gu21} it was shown that already in the case of two parties decompositions of arbitrary states into indecomposable states is not unique. More precisely, let $\ket{\psi_1}$ ($\ket{\psi_2}$) be a state of two qudits of dimension $d_1$ ($d_2$). Ref.~\cite{Gu21} shows that for any $d_1,d_2\ge 4$ there exist states of the form $\ket{\psi_1}\otimes \ket{\psi_2}$ shared among two parties and a PLU transformation $U$ such that

\begin{equation}
    U\ket{\psi_1}\otimes \ket{\psi_2}=\ket{\psi_1'}\otimes \ket{\psi_2'}
\end{equation}
and such that $\ket{\psi_1}$ (and $\ket{\psi_2}$) is not PLU equivalent to any $\ket{\psi_i'}$. If we consider quasilocal unitary operations instead of PLU, it follows from the results of Ref.~\cite{Kr09} that also decompositions of stabilizer are not unique.

If we allow for measurements in addition to PLC transformations, transformations between PLC inequivalent decompositions become possible . Consider for instance two Bell states $\ket{\phi^+}^{\otimes 2}$ shared between parties $\alpha_1=\{1\}$, $\alpha_2=\{2,3\}$ and $\alpha_3=\{4\}$. Then, applying $CZ_{2,3}$ followed by local complementation w.r.t. qubit $3$ and a $Z$ basis measurement of qubit $3$ results in the state $\ket{GHZ_3}$ shared between all three parties and $\ket{0}$ at party $\alpha_2$ (up to local Pauli corrections). It is clear that $\ket{\phi^+}^{\otimes 2}$ and $\ket{GHZ_3}\otimes \ket{0}$ are PLU inequivalent.

\section{Conclusion and Outlook}

In this work, we studied PLC transformations of stabilizer states. PLC transformations are a physically motivated extension of local Clifford operations that naturally arise in the context of quantum networks that can provide large amounts of bipartite entanglement between well connected nodes. Moreover, the fact that these operations can be implemented deterministically makes them particularly interesting to study.

First, we addressed the question of how PLC transformations between graph states change the corresponding graph. It was shown in Ref.~\cite{Vn04} that local Clifford transformations between graph states correspond to a sequence of local complementations on the respective graphs. We found that any PLC transformation between graph states can be realized by local complementations supplemented by the addition and removal of edges within parties.

Then, we studied the classification of PLC equivalence classes of stabilizer states. We discussed an approach based on invariant polynomials and showed that this approach achieves a complete classification. However, since we are not aware of an efficient way to evaluate those polynomial, the approach seems unpractical.

Motivated by the results of Ref.~\cite{Br06}, we introduced a new mathematical tool, the commutation matrix formalism, to study PLC equivalence classes of stabilizer states. More precisely, we related the problem of classifying PLC equivalence classes to the classification problem of tuples of alternating forms, which is a well studied problem in linear algebra.

We showed that two stabilizer states are PLC equivalent if they admit the same set of commutation matrices, and moreover, states can be decomposed into tensor products of smaller stabilizer states under PLC if their commutation matrices can be block-diagonalized. This approach allowed us to gain several insights into the entanglement structure of multipartite stabilzier states. We showed, that in contrast to the $3$-partite case, $M\geq 4$-partite stabilizer states contain infinitely many different types of entanglement under PLC transformations and even under PLU transformations. We derived necessary and sufficient conditions to decide whether or not a given stabilizer state is decomposable. To demonstrate the power of our approach we numerically computed the EGS$_4$ for up to $10$ qubits. We furthermore showed that the decomposition of states into indecomposable states is unique in case the states do not admit additional PLC symmetries.

Finally, we generalized the commutation matrix formalism to qudit systems of prime dimension $d$. As the commutation matrices are now defined over the field $\mathbb{Z}_d$, the results of Ref.~\cite{Se88} imply that qudit stabilizer states decompose uniquely into tensor products of indecomposable ones. Again we employed the commutation matrix formalism to compute the EGS$_4$ up to $5$ qutrits.

An interesting topic for future research is to resolve the question whether or not qubit stabilizer states decompose uniquely into indecomposable ones. Furthermore, the structure of the EGS$_4$ up to $10$ qubits suggests that maybe a complete description of the stabilizer states in this set is possible. Considering PLU transformations, it would be interesting to understand whether stabilizer states which are indecomposable under PLC can decompose under PLU into states which are not PLU equivalent to stabilizer states. Note that these states would have to be locally maximally entangleable (LME) states~\cite{Kr09}. In case such decompositions exist, one can study the question how the EGS for stabilizer states looks like if it can contain states which are not PLU equivalent to stabilizer states.

\section{Acknowledgements}
We thank Vladimir V. Sergeĭchuk for helpful correspondence. We acknowledge financial support from the Austrian Science Fund (FWF): W1259-N27, P32273-N27, FG5-L, and the SFB BeyondC (Grant No. F7107-N38).

\onecolumngrid

\clearpage

\appendix

\section{Generalized local complementation\label{app:PLC}}
In this section we provide a proof of Theorem~\ref{thm_gen_loc}. To that end, we first prove Lemma~\ref{lem_gen_loc}, which shows that the action of any $2$-qubit Clifford operator on a graph state can be implemented by $CZ$ gates and local complementation (LCE) on the respective qubits up to local Clifford operations. Let us state the lemma in its formal version. 
\begin{manuallemma}{2}[formal version]\label{lem_gen_loc_formal}
	Let $\ket{G}$ be a $n\ge 2$ qubit graph state. Then for any $C\in\mathcal{C}_2$ acting on qubits 1 and 2 there exists a local Clifford operator $L_1\in\mathcal{C}_n^L$ and a LCE transformation with respect to vertex $1$ and $2$ described by the operator $L_2\otimes L_3$, where $L_2\in\mathcal{C}_2$ acts on qubits 1 and 2, and $L_3\in\mathcal{C}_{n-2}^L$ acts on qubits $\{3,\ldots n\}$, such that 
	\begin{equation}
	C\otimes\mathds{1}\ket{G}=L_1(L_2\otimes L_3)\ket{G}.
	\end{equation}
\end{manuallemma}
Note that the 2-qubit Clifford group has $|\mathcal{C}_2|=11520$ elements and its subgroup, the local Clifford group, has $|\mathcal{C}_2^L|=24^2$ elements. Thus, $\mathcal{C}_2$ partitions into $11520/24^2=20$ different right cosets $\mathcal{C}_2^Lg=\{fg|f\in\mathcal{C}_2^L,g\in \mathcal{C}_2\}$ of $\mathcal{C}_2^L$, i.e., $\mathcal{C}_2=\bigcup_{i=1}^{20}\mathcal{C}^L_2g_i$, with $g_i\in \mathcal{C}_2^L$. With these insights let us show Lemma \ref{lem_gen_loc}.

\begin{proof}
Observe, that Lemma \ref{lem_gen_loc} is equivalent to following statement. Clifford operators corresponding to LCE with respect to vertices $1$ and $2$ can generate (up to local Clifford operations) an element of every right coset $\mathcal{C}_2^Lg$, with $g\in\mathcal{C}_2$ acting on qubits $1$ and $2$. We show this in the following by generating an element of each coset via LCE.

To this end we consider the operators that correspond to LCE. Adding and removing the edge between vertex $1$ and $2$ corresponds to a controlled phase gate $CZ_{12}=\ketbra{0}\otimes \mathds{1}+\ketbra{1}\otimes Z$ between qubits $1$ and $2$. Local complementation at qubit $1$ ($2$ qubit) corresponds to the operator $\exp(-i\pi/4X)\otimes \exp(i\pi/4Z)$ ($\exp(i\pi/4Z)\otimes \exp(-i\pi/4X)$) on qubits $1$ and $2$ if qubit $1$ and $2$ are connected. Otherwise, it corresponds to the operator $\exp(-i\pi/4X)\otimes \mathds{1}$ ($\mathds{1}\otimes \exp(-i\pi/4X)$). Moreover, local complementation creates additional local Clifford gates on other neighbours of qubit $1$ and $2$. Let us write all possible operators on qubit $1$ and $2$ as an (ordered) tuple
\begin{equation}
\left(CZ_{12},e^{-i\frac{\pi}{4} X}\otimes \mathds{1},\mathds{1}\otimes e^{-i\frac{\pi}{4} X},e^{-i\frac{\pi}{4} X}\otimes e^{i\frac{\pi}{4} Z},e^{i\frac{\pi}{4}Z}\otimes e^{-i\frac{\pi}{4} X}\right).\label{eq_tuple_lce}
\end{equation}
Figure \ref{tablea} shows a representative of each right coset of $\mathcal{C}_2^L$ in $\mathcal{C}_2$ and a sequence of the above operators generating it. The sequence respects that the operators for local complementation depend on whether qubit $1$ and $2$ are connected. Figure \ref{tablea} considers the case where qubit $1$ and $2$ are initially disconnected. Let $\{C_j\}$ be the operators listed in the table. Clearly, if qubit $1$ and $2$ are initially connected, then $\{C_jCZ_{1,2}\}$ are representatives of all different cosets.
\end{proof}

Next let us show Theorem~\ref{thm_gen_loc}, which we state here again.
\begin{manualtheorem}{1}
	Let $\ket{G}$ and $\ket{G'}$ be $n$-qubit graph states distributed among $P(M,n)$ parties. Then $\ket{G}$ is PLC equivalent to $\ket{G'}$ if and only if the graphs $G$ and $G'$ are related via a sequence of local complementations and edge additions/removals within parties (LCE).
\end{manualtheorem}

\begin{proof}[Proof of Theorem \ref{thm_gen_loc}]
	The if part is trivial.	To prove the only if part let $\ket{G}$ and $\ket{G'}$ be $M$-partite graph states. Suppose there exists a PLC operator $\bigotimes_{\alpha\in P(M,n)} C_\alpha$ such that $\ket{G'}=\bigotimes_\alpha C_\alpha\ket{G}$. Then, for every party $\beta\in P(M,n)$ we do the following. As $\mathcal{C}_{|\beta|}$ is generated by single- and two-qubit Clifford operators (see also Section \ref{sec_preliminaries}) we can write $C_\beta=\prod_{j=1}^kC_\beta^j$, where $C_\beta^j \in \mathcal{C}_2$ for all $j\in[k]$. Lemma \ref{lem_gen_loc} allows to replace the action of $C_\beta^k$ by the action of a local Clifford operator $L_1$ on the state $L_2\otimes L_3\ket{G}$. The state $L_2\otimes L_3\ket{G}$ is a again a graph state as the operator $L_2\otimes L_3$ describes a sequence of LCE. The local operator $L_1$ can then be commuted to the end of the circuit by possibly changing the remaining operators $C_\beta^j$ to some other $2$-qubit Clifford operators and possibly the operators $C_{\alpha}$ for $\alpha\neq \beta$ to some other PLC operators. We now continue like this for the operators $C_\beta^j$ for $j\in[k-1]$. It is straightforward to verify that, after applying this reasoning to every party, we end up with an equation of the form $\ket{G'}=L\ket{G''}$, where $\ket{G''}$ is obtained from $\ket{G}$ via LCE. It remains to show that the operator $L$ corresponds to local complementation. In Ref.~\cite{Vn04} it is shown that two graph states are related by a local Clifford operator if and only if their graphs are related via a sequence of local complementations. We conclude that $G''$ and $G'$ are related via a sequence of local complementations and therefore $G'$ and $G$ by a sequence of LCE.
\end{proof}

\begin{sidewaysfigure}
	\centering
	\small
		\begin{tabular}{c|c|c|c|c|c|c|c}\label{fig1}
			$B$ & generation & $B$& generation& $B$& generation& $B$& generation\\\hline\hline
			
			$\begin{pmatrix}1&1&0&0\\0&1&0&0\\0&0&1&0\\0&0&0&1\end{pmatrix}$ & $[1]$ & $\begin{pmatrix}1&1&0&1\\0&1&0&0\\0&1&1&1\\0&0&0&1\end{pmatrix}$ & $[0, 3, 0, 2, 0, 3, 0]$&$\begin{pmatrix}1&0&1&0\\0&1&0&0\\0&0&1&0\\0&1&0&1\end{pmatrix}$ & $[1, 0, 3, 0]$&$\begin{pmatrix}1&0&0&0\\0&1&0&1\\1&0&1&1\\0&0&0&1\end{pmatrix}$ & $[0, 4, 0]$\\\hline
			
			$\begin{pmatrix}0&0&1&1\\0&0&0&1\\1&0&0&0\\0&1&0&0\end{pmatrix}$ & $[0, 4, 3, 0, 2, 1, 0, 4, 0]$ & $\begin{pmatrix}0&1&1&0\\0&0&0&1\\1&0&0&1\\0&1&0&0\end{pmatrix}$ & $[0, 4, 3, 4, 0]$&$\begin{pmatrix}1&1&1&1\\0&1&0&0\\0&0&1&1\\0&1&0&1\end{pmatrix}$ & $[2, 0, 3, 0]$&$\begin{pmatrix}1&0&0&0\\0&1&1&0\\0&0&1&0\\1&0&0&1\end{pmatrix}$ & $[0]$\\\hline
			
			$\begin{pmatrix}0&1&1&1\\0&1&0&1\\1&1&1&0\\0&1&0&0\end{pmatrix}$ & $[1, 2, 0, 4, 3, 0]$ & $\begin{pmatrix}0&1&1&0\\0&1&0&1\\1&1&1&1\\0&1&0&0\end{pmatrix}$ & $[1, 0, 4, 3, 0]$ &$\begin{pmatrix}1&0&0&0\\0&1&1&1\\0&0&1&1\\1&0&0&1\end{pmatrix}$ & $[2, 0]$&$\begin{pmatrix}1&0&1&0\\0&0&0&1\\1&1&0&1\\0&1&0&1\end{pmatrix}$ & $[1, 0, 3, 4, 0]$\\\hline
			
			$\begin{pmatrix}1&0&1&1\\0&0&1&0\\0&1&1&0\\1&1&1&0\end{pmatrix}$ & $[2, 1, 0, 3, 0, 2, 0]$&$\begin{pmatrix}1&0&1&0\\0&0&1&1\\0&1&1&1\\1&1&1&1\end{pmatrix}$ & $[0, 3, 0, 2, 1, 0]$&$\begin{pmatrix}1&1&1&0\\0&0&0&1\\1&0&0&1\\0&1&0&1\end{pmatrix}$ & $[0, 3, 4, 0]$&$\begin{pmatrix}1&1&0&0\\0&1&0&1\\1&1&1&1\\0&0&0&1\end{pmatrix}$ & $[1, 0, 4, 0]$\\\hline
			
			$\begin{pmatrix}1&1&1&1\\0&0&1&0\\0&1&1&0\\1&0&1&0\end{pmatrix}$ & $[2, 0, 3, 0, 2, 0]$&$\begin{pmatrix}1&1&0&0\\0&1&1&0\\0&0&1&0\\1&1&0&1\end{pmatrix}$ & $[1,0]$&$\begin{pmatrix}1&1&0&0\\0&1&1&1\\0&0&1&1\\1&1&0&1\end{pmatrix}$ & $[2,1,0]$&$\begin{pmatrix}1&1&0&1\\1&1&1&0\\1&0&1&1\\1&1&0&0\end{pmatrix}$ & $[0, 4, 0, 1, 0]$\\\hline

		\end{tabular}
	\caption{This table contains an element $B$ of each right coset of $\mathcal{C}_2^L$ in $\mathcal{C}_2$ and a sequence of operators from the tuple in Eq.~\eqref{eq_tuple_lce} generating it. Each sequence is specified by a list of indices of the respective elements in the tuple (starting from $0$) and assuming that qubit $1$ and $2$ are initially not connected; the left most entry is the index of the operator applied first; the order of the generators is compatible with the fact the operators for local complementation depend on whether qubit $1$ and $2$ are connected.\label{tablea}}

\end{sidewaysfigure}

\section{Symplectic vector spaces and their connection to stabilizers\label{sec_sym}}
There is a close connection between stabilizer states and symplectic vector spaces. In this section we recall some important results on symplectic vector spaces, which have been studied extensively in the literature. For more details we refer the reader to, e.g., Refs.~\cite{As00,Gr06,La02}.

Let $V$ be a finite-dimensional vector space over any field $F$. A bilinear form is a map $\omega:V\times V\rightarrow F$ which is linear in both arguments. After fixing a basis $\{\vec{e}_i\}$ for $V$, we can associate to the bilinear form $\omega$ the matrix $B_\omega\in M_{\text{dim}(V)\times \text{dim}(V)}(F)$ whose entries are given by $(B_\omega)_{ij}=\omega(\vec{e}_i,\vec{e}_j)$. The bilinear form may then be written as
\begin{equation}
\omega(\vec{f},\vec{g})=\vec{f}^TB_\omega \vec{g}.
\end{equation}
A bilinear form is called \emph{alternating} if $\omega(\vec{f},\vec{f})=0$ for all $\vec{f}\in V$. Note that any alternating form is also skew-symmetric, i.e., $\omega(\vec{x},\vec{y})=-\omega(\vec{y},\vec{x})$, which follows from $\omega(\vec{x}+\vec{y},\vec{x}+\vec{y})=0$ together with the bilinearity of $\omega$. Consequently, any matrix representation of an alternating bilinear form satisfies $(B_\omega)_{ij}=-(B_\omega)_{ji}$ and $(B_\omega)_{ii}=0$ for all $i,j\in[\text{dim}(V)]$. Any matrix with these properties is called alternating. 

Given an alternating form and a subspace $W\subseteq V$, we define the dual space as $W^\perp=\{\vec{f}\in V|\omega(\vec{f},\vec{g})=0\ \forall \vec{g}\in W\}$. An important type of subspace which appears in this context is a \emph{hyperbolic plane}. A hyperbolic plane is a $2$-dimensional subspace spanned by a hyperbolic pair, which is a pair of vectors $\vec{f},\vec{g}\in V$ such that $\omega(\vec{f},\vec{g})=1$. Given a vector space and an alternating form, the vector space decomposes into a orthogonal direct sum of hyperbolic planes and a subspace on which the form is $0$. This is stated in the following theorem, a proof of which can for instance be found in Ref.~\cite{As00}.
\begin{thm}\label{thm_dickson_gen}
	Let $V$ be a vector space and let $\omega:V\times V\rightarrow F$ be an alternating form. Then $V=V_1\oplus\ldots\oplus V_m\oplus V^\perp$ where $V_1,\ldots ,V_m$ are hyperbolic planes for some $m\in \mathbb{N}$ and $\omega(\vec{f},\vec{g})=0$ if $\vec{f}\in V^\perp$ or $\vec{g}\in V^\perp$.
\end{thm}
This Theorem is equivalent to the following statement. For every alternating form $\omega$ on a vector space $V$ there exists a basis of $V$ w.r.t. which the matrix representation of $\omega$ is a direct sum of the blocks
\begin{equation}
\begin{pmatrix}
0
\end{pmatrix},\text{ and }
\begin{pmatrix}
0&1\\
-1&0
\end{pmatrix}.\label{eq_dickson_mat_form}
\end{equation}
Stated differently, there always exists a basis change in $V$ described by the invertible matrix $R$ such that $R^TB_\omega R$ is block-diagonal, with the blocks as specified above. Note that for $V=\mathbb{Z}_2^n$ Theorem~\ref{thm_dickson_gen} is also known as \emph{Dickson's Theorem} (cf. Chapter 15 in Ref.~\cite{Sl77}). We state it here for general fields as, when we treat qudit stabilizer states with qudits of prime dimension $d$, we need the Theorem for the field $\mathbb{Z}_d$.

Let us continue with some definitions. A bilinear form is called \emph{non-degenerate} if the conditions $\omega(\vec{f},\vec{g})=0$ for all $\vec{g}\in V$ imply that $\vec{f}=0$. An alternating bilinear form which is non-degenerate is called \emph{symplectic}. A symplectic vector space is a vector space $V$ together with symplectic bilinear form $\omega$. Symplectic vector spaces have even dimension as follows for instance from Theorem~\ref{thm_dickson_gen}. A subspace $W\subseteq V$ of a symplectic vector space $V$ is called \emph{isotropic} if $W\subseteq W^\perp $ and \emph{maximally isotropic} if $W^\perp=W$. Moreover, for any subspace $W\subseteq V$ it holds that $\text{dim}(W)+\text{dim}(W^\perp)=\text{dim}(V)$.

An invertible linear map $U:V\rightarrow V$ is referred to as \emph{isometry} for a symplectic form $\omega$ if it preserves $\omega$, i.e., if $\omega(U\vec{f},U\vec{g})=\omega(\vec{f},\vec{g})$ for all $\vec{f},\vec{g}\in V$, or equivalently, if $U^TB_\omega U=B_\omega$.
Such transformations are also called \emph{symplectic} transformations. Symplectic vector spaces allow for the extension of any isometry between subspaces to isometries on the whole vector space, as stated in the following lemma which can be found in Ref.~\cite{As00}.

\begin{lem}[Witt's Lemma]\label{lem_witt_lemma}
	Let $V$ be a symplectic space with symplectic form $\omega:V\times V\rightarrow \mathcal{F}$. Let $A,B$ be subspaces of $V$ and let $U:A\rightarrow B$ be an invertible linear map which satisfies $\omega(U\vec{f},U\vec{g})=\omega (\vec{f},\vec{g})$. Then $U$ can be extended to an isometry of $V$.
\end{lem}

Let us see how the above results connect to the Pauli group and stabilizer states. As described in Section \ref{sec_previous_results}, the Pauli group $\mathcal{P}_n$ modulo phases is isomorphic to the vector space $\mathbb{Z}_2^{2n}$. The commutation relations of the Pauli operators define a symplectic form on $\mathbb{Z}_2^{2n}$, i.e., the Pauli group modulo phases is isomorphic to a symplectic vector space. This correspondence also exists for the qudit Pauli group and $\mathbb{Z}_d^{2n}$ (see Section \ref{sec_qudit_systems}). A stabilizer corresponding to a stabilizer state is a maximal abelian subgroup of $\mathcal{P}_n$ which does not contain $-\mathds{1}$. Via the above isomorphism it corresponds to a maximally isotropic subspace of $\mathbb{Z}_2^{2n}$.

For instance, consider the Bell state $\ket{\phi^+}$, whose stabilizer is $\mathcal{S}=\{\mathds{1},X_1Z_2,Z_1X_2,-Y_1 Y_2\}$. It is straightforward to verify that any operator in $\mathcal{P}_2$ which commutes with every element of $\mathcal{S}$ is already an element of $\mathcal{S}$ up to a phase. The stabilizer $\mathcal{S}$ corresponds to the subspace $V^\mathcal{S}=\{(0,0,0,0),(1,0,0,1),(0,1,1,0),(1,1,1,1)\}\subset\mathbb{Z}_2^4$ via the homomorphism $\sigma:\mathbb{Z}_2^4\rightarrow \mathcal{P}_n$ defined by
\begin{align}
    \sigma((1,0,0,0))&=X_1\\
    \sigma((0,1,0,0))&=Z_1\\
    \sigma((0,0,1,0))&=X_2\\
    \sigma((0,0,0,1))&=Z_2.
\end{align}
This subspace is maximally isotropic, i.e., if $\vec{f}=(a_1,b_1,a_2,b_2)\in \mathbb{Z}_2^4$ such that $\omega(\vec{f},\vec{g})=0$ for all $\vec{g}=(a_1',b_1',a_2',b_2')\in V^\mathcal{S}$ with respect to the symplectic form
\begin{equation}
    \omega(\vec{f},\vec{g})=(a_1,a_2,b_1,b_2)^T\begin{pmatrix}
    0&\mathds{1}_2\\\mathds{1}_2&0
    \end{pmatrix}(a_1',a_2',b_1',b_2'),
\end{equation}
then $\vec{f}\in V^\mathcal{S}$. It is straightforward to see that $\omega$ describes the commutation relations of the respective Pauli operators, for example $\omega((1,0,0,0),(0,1,0,0))=1$ as $X_1Z_1=-Z_1X_1$ but $\omega((1,0,0,0),(0,0,1,0))=0$ as $X_1X_2=X_2X_1$.

Clifford operators map Pauli operators to Pauli operators. As they preserve the commutation relations of the Pauli operators Clifford operators correspond to isometries on the symplectic vector space $(\mathbb{Z}_2^{2n},\omega)$. On the contrary, any homomorphism between subgroups of the Pauli group which preserves commutation relations and is invertible can be extended to an isometry due to Lemma~\ref{lem_witt_lemma} and can be implemented by a Clifford operator (cf. Equation~\eqref{eq_eq4}).

\section{Properties of commutation matrices\label{app:comm}}

In this section we are concerned with the rank condition (cf. Equation~\eqref{eq_rank_condition}), which is necessary for commutation matrices. First, we prove Theorem~\ref{thm_rank_condition} and then show that the equality in Equation~\eqref{eq_rank_condition} is an inequality for general tuples of alternating matrices. We note that the theorem also holds for any finite field, and the proof follows simply by replacing $\mathbb{Z}_2$ with $\mathbb{Z}_d$.

For convenience, let us restate Theorem~\ref{thm_rank_condition}.

\begin{manualtheorem}{6}
    Let $(C_\alpha)_{\alpha\in P(M,n)}$, $C_\alpha\in M_{n\times n}(\mathbb{Z}_2)$ be a tuple of alternating matrices such that $\sum_{\alpha} C_\alpha=0$ for some partition $P(M,n)$. Then there exists a $\ket{\psi}\in \text{Stab}(P(M,n))$ such that $(C_\alpha)_{\alpha}$ are the corresponding commutation matrices if and only if
	\begin{equation}
	2\,{\rm rk}([C_\alpha]_{\alpha})=\sum_{\alpha} {\rm rk}(C_\alpha)\label{eq:rank_appendix}.
	\end{equation}
	Here, $[C_\alpha]_{\alpha}$ is the $n\times (nm)$ matrix obtained by concatenating all matrices $C_\alpha$, and the rank is taken over the field $\mathbb{Z}_2$.
\end{manualtheorem}

\begin{proof}
First, we prove the only if part of the statement. Suppose that $(C_\alpha)_{\alpha}$, with $C_\alpha\in M_{n\times n}(\mathbb{Z}_2)$, correspond to an $n$ qubit stabilizer state. We want to prove that Equation~\eqref{eq:rank_appendix} holds. To begin with, let us consider the l.h.s. of Equation~\eqref{eq:rank_appendix}. Observe that ${\rm rk}([C_\alpha]_{\alpha})={\rm rk}([QC_\alpha Q^T]_{\alpha})$ for any invertible $Q\in M_{n\times n}(\mathbb{Z}_2)$. Consider a $Q$ such that $[QC_\alpha Q^T]_\alpha$ has the maximum number of zero rows. The number of zero rows is equal to the number of generators of $\mathcal{S}$ that locally commute with all other elements in the stabilizer. These generators correspond to qubits that can be locally extracted, as they are not entangled with any other qubit outside of their corresponding parties (see discussion above Theorem \ref{thm_block_diag}). For convenience, let us now assume that those states have been extracted, so that ${\rm rk}([C_\alpha]_{\alpha})=N\leq n$ is of full rank, i.e., equal to the number $N$ of qubits that are entangled with at least one other qubit outside of their own party. Next, let us consider the r.h.s of Equation~\eqref{eq:rank_appendix}, and let us establish the connection between the ranks of the commutation matrices and the ranks of the corresponding reduced states. Recall, that for stabilizer states the reduced states are given by Equation~\eqref{eq:reducedstate}, i.e., by the sum over all elements $s\in\mathcal{S}$ with their support fully contained in $\alpha$. Consider the matrix $C_\alpha$. Due to Theorem~\ref{thm_dickson_gen} there exists an invertible matrix $Q$ such that $QC_\alpha Q^T$ is a direct sum of blocks of the form as in Equation~\eqref{eq_dickson_mat_form}. We observe that whenever a generator anti-commutes with another generator on party $\alpha$ they also have to anti-commute on at least one other party (as the generators have to commute as a whole), and hence, every nonzero block increases the rank of the reduced state by a factor of two, and thus, $\text{rk}(\varrho_\alpha)=2^{\text{rk}(C_\alpha)/2}$. From this equation it follows that $\text{rk}(C_\alpha)/2$ is the number of qubits in $\alpha$ that are entangled with at least one other qubit in another party. Thus, $\sum_\alpha\text{rk}(C_\alpha)/2$ is equal to the total number $N$ of entangled qubits, and thus, Equation~\eqref{eq:rank_appendix} holds.

The if part of the statement follows directly form the construction below Theorem~\ref{thm_rank_condition} in the main text.
\end{proof}

Next, let us show that Equation \eqref{eq_rank_condition} is an inequality for general sets of alternating matrices. We show the following theorem over any field $F$ as we use it later not only for $\mathbb{Z}_2$ but any $\mathbb{Z}_d$ with $d$ a prime number.

\begin{thm}\label{thm_rank_ineq}
	Let $\{A_\alpha\}_{\alpha\in [M]}\subset M_{n\times n}(F)$ with $M\in\mathbb{N}$  be a finite set of alternating matrices such that $\sum_\alpha A_\alpha=0$. Then
	\begin{equation}
	2\,{\rm rk}([A_\alpha]_{\alpha \in [M]})\le\sum_{\alpha\in[M]} {\rm rk}(A_\alpha)\label{ineq_rank_condition_general}
	\end{equation}
	where $[A_\alpha]_{\alpha \in [M]}\in M_{n\times (k n)}(F)$ is the matrix obtained by concatenating the matrices $\{A_\alpha\}_{\alpha\in [M]}$.
\end{thm}

To prove Theorem \ref{thm_rank_ineq} we use the following Lemma shown in Ref.~\cite{Ma64}.

\begin{lem}\label{thm_rank_general}
	Let $A,B\in M_{n\times m}(F)$. Then
	\begin{equation}
	{\rm rk}([A,B])+{\rm rk}\begin{pmatrix}
	A \\ B
	\end{pmatrix}\le {\rm rk}(A)+ {\rm rk}(B) + {\rm rk}(A+B).\label{eq_them_rank_general}
	\end{equation}
\end{lem}

\begin{proof}[Proof of Theorem \ref{thm_rank_ineq}]
	Following Ref.~\cite{Ma64}, we denote in the following for any matrix $A\in M_{k\times l}(F)$, $k,l\in\mathbb{N}$ by $\overline{A}\in M_{k\times k}(\mathbb{Z}_2)$ the orthogonal projector on the columnspace of $A$. 
	
	We prove the statement via induction. For $M=3$ the statement is equivalent to Lemma~\ref{thm_rank_general} and holds. Suppose the statement holds for all $M$ from $3$ to $N$. Let us show that then it also holds for $M=N+1$.
	
	Let $\{A_\alpha\}_{\alpha\in [N+1]}\subset M_{n\times n}(F)$ be a set of alternating matrices such that $\sum_\alpha A_\alpha=0$. We want to show that Inequality~\eqref{ineq_rank_condition_general} holds. To this end, let us first show another inequality that will be helpful in the proof. In the following we write $C_X$ for the column space of the matrix $X$. Then, we find that
    \begin{align}
        &{\rm rk}((\mathds{1}-\overline{[A_\alpha]_{\alpha\in[N-1]}})A_N)+\text{dim}(C_{\sum_{\alpha=1}^{N-1}A_\alpha}\cap C_{A_N})\\\le &{\rm rk}((\mathds{1}-\overline{[A_\alpha]_{\alpha\in [N-1]}})A_N)+\text{dim}(C_{[A_\alpha]_{\alpha\in[N-1]}}\cap C_{A_N})={\rm rk}(A_N).\label{ineq_eq2}
    \end{align}
    The inequality follows from the fact that $C_{\sum_{\alpha=1}^{N-1} A_\alpha}\subseteq C_{[A_\alpha]_{\alpha\in[N-1]}}$. The last equality can be seen as follows. First observe that the equality holds if and only if it holds for $A_NR$, where $R$ is an invertible matrix. Let $f_j=f_j^{\perp}\oplus f_j^{||}$ for any $j$ be the columns of $A_NR$ where $(\mathds{1}-\overline{[A_\alpha]_{\alpha\in[N-1]}})f_j^{\perp}=f_j^\perp$ and $(\mathds{1}-\overline{[A_\alpha]_{\alpha\in[N-1]}})f_j^{||}=0$. We choose $R$ such that there exists a $k\le n$ such that $f_j^\perp=0$ for all $j>k$ and $f_1^\perp,\ldots,f_k^\perp$ are independent. Then, Equation~\ref{ineq_eq2} follows from the fact that the number of independent columns of $A_NR$ coincides with $k$ plus the number of independent vectors in $\{f_{k+1},\ldots ,f_{n}\}$. The latter is equal to the dimension of the intersection of the column spaces of $[A_\alpha]_{\alpha\in[N-1]}$ and $A_N$. Let us now use Inequality~\eqref{ineq_eq2} to show the induction step. First, observe that
    \begin{align}
    2\,{\rm rk}([A_\alpha]_{\alpha\in [N+1]})&=2\,{\rm rk}([A_\alpha]_{\alpha\in [N]})\\
    &=2\,{\rm rk}([A_\alpha]_{\alpha\in[N-1]})+2\,{\rm rk}((\mathds{1}-\overline{[A_\alpha]_{\alpha\in[N-1]}})A_N),\label{eq_eq_6}
    \end{align}
    where the first equation follows from $\sum_{\alpha=1}^{N+1}A_\alpha=0$. The second one follows from the fact that ${\rm rk}([A_\alpha]_{\alpha\in [N]})={\rm rk}([A_\alpha]_{\alpha\in [N]}R)$ for any invertible matrix $R$. Choosing the $R$ such that the last $n$ columns of $[A_\alpha]_{\alpha\in [N]}R$ are orthogonal to the others (and leaving the others invariant) the equality is straightforward to verify. Equation~\ref{eq_eq_6} allows us to use the assumption of the induction step, i.e. that the Inequality~\ref{ineq_rank_condition_general} holds for $M=N$. Thus, we can bound Equation~\ref{eq_eq_6} as
    \begin{align}
    2\,{\rm rk}([A_\alpha]_{\alpha\in [N+1]})\le \sum_{\alpha=1}^{N-1}{\rm rk}(A_\alpha)+{\rm rk}(\sum_{\alpha=1}^{N-1}A_\alpha)	+2\,{\rm rk}((\mathds{1}-\overline{[A_\alpha]_{\alpha\in[N-1]}})A_N).\label{eq_eq1}
    \end{align}
    Next, we want to bound the term which contains the projection operator from above. Inserting Inequality~\eqref{ineq_eq2} we find that
    \begin{align}
    2\,{\rm rk}([A_\alpha]_{\alpha\in [N+1]})\le \sum_{\alpha=1}^{N-1}{\rm rk}(A_\alpha)+{\rm rk}(\sum_{\alpha=1}^{N-1}A_\alpha)	+2\,{\rm rk}(A_N)-2\,\text{dim}(C_{\sum_{\alpha=1}^{N-1}A_\alpha}\cap C_{A_N}).\label{ineq_eq_4}
    \end{align}
    To eliminate the last term in the latter inequality we use that due to the assumption of the induction step the Inequality~\ref{ineq_rank_condition_general} holds in particular for $M=3$. Thus,
    \begin{align}
        2{\rm rk}([\sum_{\alpha=1}^{N-1}A_\alpha,A_N,A_N+\sum_{\alpha=1}^{N-1}A_\alpha])\le{\rm rk}(\sum_{\alpha=1}^{N-1}A_\alpha)+{\rm rk}(A_N)+{\rm rk}(A_N+\sum_{\alpha=1}^{N-1}A_\alpha).
    \end{align}
    Using the fact that ${\rm rk}([\sum_{\alpha=1}^{N-1}A_\alpha,A_N,A_N+\sum_{\alpha=1}^{N-1}A_\alpha])={\rm rk}(\sum_{\alpha=1}^{N-1}A_\alpha)+{\rm rk}(A_N)-\text{dim}(C_{\sum_{\alpha=1}^{N-1}A_\alpha}\cap C_{A_N})$ (see also Ref.~\cite{Ma64}) the latter inequality can be written as
    \begin{align}
    {\rm rk}(\sum_{\alpha=1}^{N-1}A_\alpha)+{\rm rk}(A_N)-2\,\text{dim}(C_{\sum_{\alpha=1}^{N-1}A_\alpha}\cap C_{A_N})\le {\rm rk}(\sum_{\alpha=1}^{N}A_\alpha).\label{ineq_eq_3}
    \end{align}
    Finally, using Inequality~\eqref{ineq_eq_3} in Inequality~\eqref{ineq_eq_4} we find
    \begin{align}
        2\,{\rm rk}([A_\alpha]_{\alpha\in [N+1]})&\le\sum_{\alpha=1}^{N-1}{\rm rk}(A_\alpha)+{\rm rk}(A_N)+{\rm rk}(\sum_{\alpha=1}^{N}A_\alpha)=\sum_{\alpha=1}^{N+1}{\rm rk}(A_\alpha),
    \end{align}
    which completes the proof.
\end{proof}

\section{Classification of forms\label{sec_clas_forms}}

In this appendix we summarize, in mathematical terms, results from linear algebra obtained in the context of the classification of systems of linear maps and forms. The classification of tuples of commutation matrices is a special case of this general framework. Given an equivalence relation, the term \emph{classification} refers to the task of identifying a representative of every equivalence class. For systems of linear maps and forms the equivalence relation is induced by basis changes on the spaces on which they act.

Let us start by summarizing the results of Ref.~\cite{Sc76}, which is concerned with the classification of so called \emph{bialternating modules} over any field $F$. A bialternating module is a pair of bilinear, alternating forms $A_1,A_2:V\times V\rightarrow F$ acting on the same vector space $V$ over the field $F$. The equivalence relation induced by a basis change on $V$ is a congruence transformation of $A_1$ and $A_2$ with an invertible linear map $Q:V\rightarrow V$.

Before we outline the results of Ref.~\cite{Sc76} let us introduce some terminology. A bialternating module $(A_1,A_2)$ is decomposable if there exists an invertible linear map $Q:V\rightarrow V$ such that $QA_iQ^T=B_i^1\oplus B_i^2$ for $i\in[2]$. A decomposition of a tuple into indecomposable tuples is said to be unique if any other decomposition of the same tuple into indecomposable ones is the same up to reordering of the indecomposable blocks and equivalence of the blocks.

To solve this classification problem, Ref.~\cite{Sc76} uses the following approach, which is commonly used in these types of classification problems. First, one establishes that bialternating modules decompose uniquely into blocks. Then the problem is solved by stating a maximal list of non-equivalent, indecomposable bialternating modules. Before we comment on the idea how to solve both these tasks for bialternating modules, let us present the complete list of indecomposable bialternating modules. A detailed explanation of the results of Ref.~\cite{Sc76} can be found in Ref.~\cite{Ba99} from where the following list is taken.

Let $n\in\mathbb{N}$ and let $f(x)=x^n+a_n x^{n-1}+\ldots+a^1$ be a power of an irreducible polynomial over $F$. A polynomial is called irreducible if it cannot be written as a product of two non-constant polynomials. A complete list of non-congruent, indecomposable, bialternating modules is given by modules $(C_1,C_2)$ where
\begin{equation}
C_1=\begin{pmatrix}
0&M_1\\
-M_1^T&0
\end{pmatrix},\ \ C_2=\begin{pmatrix}
0&M_2\\
-M_2^T&0
\end{pmatrix}\label{eq_standard_form4}
\end{equation}
and 
\begin{equation}
M_1=\begin{pmatrix}
0&\mathds{1}_n
\end{pmatrix},\ \ M_2=\begin{pmatrix}
\mathds{1}_n&0
\end{pmatrix}\label{eq_standard_form1}
\end{equation}
or
\begin{equation}
M_1=\mathds{1}_n,\ \ M_2=\begin{pmatrix}
0&0\\
\mathds{1}_{n-1}&0
\end{pmatrix}\label{eq_standard_form2}
\end{equation}
or 
\begin{equation}
M_1=\mathds{1}_n,\ \ M_2=\begin{pmatrix}
-a_1&-a_2&\ldots &-a_n\\
1&0&\ldots &0\\
0&1&\ldots &0\\
\vdots&\vdots&&\vdots\\
0&0&\ldots&1  
\end{pmatrix}.\label{eq_standard_form3}
\end{equation}
It is straightforward to check, that the only modules of this list, which satisfy the rank condition in Equation~\eqref{eq_rank_condition} are (congruent to) the ones corresponding to the GHZ state and the Bell state (see Equations~\eqref{eq_com_mat_bell} and~\eqref{eq_com_mat_ghz}).

Both results of Ref.~\cite{Sc76}, i.e., the uniqueness of the decomposition and the complete set of indecomposables, follow from a correspondence between bialternating modules and so-called \emph{Kronecker modules}. A Kronecker module, also known as \emph{matrix pencil}, is a pair of linear maps $K_1,K_2:V\rightarrow W$ between $F$ vector spaces $V$ and $W$. These objects were studied under simultaneous isomorphy corresponding to a basis change on $V$ and $W$. Two Kronecker modules $(K_1,K_2)$ and $(\tilde{K}_1,\tilde{K}_2)$ are isomorphic if there exist invertible matrices $P\in M_{n\times n}(F)$ and $Q\in M_{m\times m}(F)$ such that $(\tilde{K}_1,\tilde{K}_2)=(QK_1P,QK_2P)$. It was shown in earlier works that Kronecker modules decompose uniquely into indecomposable blocks. For an overview of the results regarding Kronecker modules we refer the reader to Ref.~\cite{Fal04} and the references therein. Moreover, a complete set of indecomposable modules was known~\cite{Di46}. To apply those results to the classification of bialternating modules, Ref.~\cite{Sc76} establishes that every bialternating module is, up to congruence, of the form
\begin{equation}
\left(\begin{pmatrix}
0&K_1\\K_1^T&0
\end{pmatrix},\begin{pmatrix}
0&K_2\\K_2^T&0,
\end{pmatrix}\right)\label{eq_neutral}
\end{equation}
for some pair of matrices $(K_1,K_2)$. Bialternating modules of this form are called \emph{neutral} in Ref.~\cite{Sc76}. This property together with the knowledge of the classification of Kronecker modules then allows one to show that two bialternating modules are congruent if and only if they are isomorphic (understood as Kronecker modules). The main results of Ref.~\cite{Sc76} then follow directly from this observation.

Kronecker modules are a special case of the classification of representations of so-called \emph{quivers}. A basic introduction into the representation theory of quivers can be found in Ref.~\cite{Ba}. A quiver is a mathematical graph with vertices $V$ and directed edges $E\subset V^2$. Multiple edges and self-loops are allowed. A representation of a quiver is a map $R$ that assigns a vector space $R(v)$ to every $v\in V$ and a linear map $R((v_1,v_2)):R(v_1)\rightarrow R(v_2)$ to every edge $(v_1,v_2)\in E$. Two representations $R$ and $R'$ of a quiver are called isomorphic if for every $v\in V$ there exist invertible matrices $Q_v:R(v)\rightarrow R'(v)$ such that for any edge $(v_1,v_2)\in E$ it holds that $R'((v_1,v_2))=Q_{v_2} R(v)Q_{v_1}^{-1}$. Kronecker modules are representations of a quiver, which consists of two vertices that are connected by two edges of the same direction, i.e.,
\begin{equation}
    v_1 \,\substack{\xrightarrow{\hspace*{1cm}}\\[-1em] \xrightarrow{\hspace*{1cm}} } \,v_2.
\end{equation}
Therefore, this quiver is also called the \emph{Kronecker quiver}. Note that classifying the representations of the Kronecker quiver with 3 and more edges is a wild problem. A important result in the representation theory of quivers is the so-called \emph{Krull-Remak-Schmidt theorem} (see for instance Ref.~\cite{Ba}).

\begin{thm}[Krull-Remak-Schmidt]\label{thm_krm}
	Let $R_1,\ldots, R_n$ and $S_1,\ldots, S_m$ be indecomposable representations of a quiver such that $R_1\oplus \ldots \oplus R_n$ is isomorphic to $S_1\oplus \ldots \oplus S_m$, i.e., there exist invertible matrices $Q_v:R_1(v)\oplus\ldots \oplus R_n(v)\rightarrow S_1(v)\oplus\ldots\oplus S_m (v)$ such that
	\begin{equation}
	    S_1((v_1,v_2))\oplus\ldots\oplus S_m ((v_1,v_2))=Q_{v_2}R_1((v_1,v_2))\oplus\ldots\oplus R_m ((v_1,v_2))Q_{v_1}^{-1}
	\end{equation}
    holds for all $(v_1,v_2)\in E$.	Then $n=m$ and there exists a permutation $\pi$ and matrices $Q_v^i:R_i(v)\rightarrow S_i(v)$ such that
    \begin{equation}
	    S_{\pi(i)}((v_1,v_2))=Q_{v_2}^{i}R_i((v_1,v_2)) (Q_{v_1}^{i})^{-1}
	\end{equation}
	holds for all $(v_1,v_2)\in E$ and for all $i\in[n]$.
\end{thm}

In other words, representations of quivers decompose uniquely into indecomposable representations under isomorphism (for any field). We make use of this theorem in Appendix~\ref{sec_uniqueness} where we discuss whether commutation matrices over the field $\mathbb{Z}_2$ decompose uniquely into indecomposable blocks.

Representations of quivers can yet again be viewed as a special case of self-adjoint representations of linear categories with involution, a framework studied in Ref.~\cite{Se88}. The main result of Ref.~\cite{Se88} is a Krull-Remak-Schmidt-type Theorem for self-adjoint representations of linear categories with involution. This theorem implies in particular that the decomposition of tuples of alternating forms over the field $\mathbb{Z}_d$ is unique for $d>2$. Via the commutation matrix formalism this result implies that qudit stabilizer states for qudits of prime dimension decompose uniquely under PLC, as stated in the main text (see Section \ref{sec_qudit_systems}).  An explanation of the results of Ref.~\cite{Se88} can be found in Ref.~\cite{Pl14}.

\section{Extraction of $\ket{\Phi^+}$ and $\ket{GHZ}$ states\label{sec_proof_lem_ghz}}

In this section we prove Observation~\ref{thm_bounds_egs}. To this end we first show Lemma~\ref{lem_extract_ghz}. For convenience let us restate the lemma.

\begin{manuallemma}{9}
	Let $P(M,n)=\{\alpha_1,\ldots ,\alpha_{M}\}$ be a partition and let $\ket{\psi}\in \text{Stab}(P(M,n))$ have full local ranks, i.e., $\text{dim}(V^\mathcal{S}_\alpha)=0$ for all $\alpha\in P(M,n)$. Then one can extract either $\ket*{GHZ_{M-1}}$ or $\ket{GHZ_{M}}$ from $\ket{\psi}$ if and only if there exist elements $f_1,\ldots,f_{M-2}$ in $\mathcal{S}$ such that $\text{supp}(f_j)=\alpha_1\cup \alpha_{j+1}$ and $(f_j)_{\alpha_1}=(f_k)_{\alpha_1}$ for all $j,k\in [M-2]$.
\end{manuallemma}

\begin{proof}
	To prove the statement, we show that there exists a graph state, which is PLC equivalent to $\ket{\psi}$ and for which the claim directly follows from Theorem \ref{thm_PLC_extract}.
	
	It follows from Witt's Lemma (Lemma~\ref{lem_witt_lemma}) that there exists a PLC transformation $C\in\mathcal{C}^{P(M,n)}$ such that $Cf_jC^\dagger|_{\alpha_1}$ and $Cf_jC^\dagger|_{\alpha_{j+1}}$ act nontrivially on a single qubit each for all $j\in [M-2]$. W.l.o.g. we order the qubits such that for any $j$ these are qubits 1 and $j+1$. Next find an $D\in\mathcal{C}_n^{L}$ such that $DC\ket{\psi}=\ket{G}$ is a graph state (see Ref.~\cite{Vn04}). W.l.o.g. we choose $D$ such that $DCf_{1}C^\dagger D^\dagger|_{\alpha_1}=Z\otimes \mathds{1}$. Indeed, qubit $1$ is connected to another qubit in $G$ as $\text{dim}(V_{\alpha_1}^\mathcal{S})=0$. Thus, if $DCf_{j}C^\dagger D^\dagger|_{\alpha_1}=Y\otimes \mathds{1}$ or $DCf_{j}C^\dagger D^\dagger|_{\alpha_1}=X\otimes \mathds{1}$ there exists a $D'\in\mathcal{C}_n^L$ which corresponds to a local complementation on $G$ such that $D'DCf_{j}C^\dagger D^\dagger(D')^\dagger|_{\alpha_1}=Z\otimes \mathds{1}$.
	
	We discuss now what the above properties imply for the structure of $\ket{G}$. Let us use that $\mathcal{S}$ is generated by the canonical generators (see Equation~\eqref{eq:generators}). Combining this with $DCf_{j}C^\dagger D^\dagger|_{\alpha_1}=Z\otimes \mathds{1}$ and the fact that $DCf_{j}C^\dagger D^\dagger|_{\alpha_{j+1}}$ acts only on a single qubit for any $j$, we conclude that $DCf_{j}C^\dagger D^\dagger|_{\alpha_{j+1}}=X\otimes \mathds{1}$ for any $j\in [M-2]$. That is, the qubits $\{2,\ldots M-1\}$ are only connected to qubit $1$. Let us use this structure to further simplify $\ket{G}$. Suppose that vertex $1$ is connected to another vertex $k$ in party $\alpha_j$ with $k\neq j$. If qubit $k$ and $j$ are not connected, then $CZ_{kj}$, followed by $L\in\mathcal{C}_n^L$ which corresponds to local complementation at vertex $j$, followed by another operator $CZ_{kj}$ removes the connection between vertex $1$ and $k$ and leaves $G$ unchanged otherwise. If qubit $k$ and $j$ are connected, one obtains the same result by leaving out the first $CZ_{kj}$ gate. This reasoning applies to any connection of vertex $1$ to any vertex not in $\{2,\ldots,M-1\}$, which is associated to a party in $\{\alpha_2,\ldots ,\alpha_{M-1}\}$. 
	
	Let $N_j(\alpha_k)$ be the neighbourhood of vertex $j$ in party $\alpha_k$. Suppose there exists a (possibly empty) set of vertices $W\subset\bigcup_{j\in \{2,\ldots,M-1\}}\alpha_j$ such that
	\begin{equation}
	\bigotimes_{i\in N_1(\alpha_{M})}Z_i=\prod_{j\in W}\bigotimes_{i\in N_j(\alpha_{M})}Z_i.\label{eq_eq2}
	\end{equation}
	Then the linear map $T$ from the stabilizer of the $M-1$-partite GHZ state shared between parties $1,\ldots, m-1$ to the stabilizer of $\ket{G}$ defined by
	\begin{equation}
	Z_1Z_2,\ldots,Z_1Z_{M-1},X_1\ldots X_{M-1}\rightarrow g_2,\ldots,g_{M-1},g_1\prod_{j\in W}g_j
	\end{equation}
	satisfies all constraints of Theorem~\ref{thm_PLC_extract} (\cite{Br06}). Here, $g_j=X_j\bigotimes_{k\in N_j} Z_k$ are the canonical generators of the stabilizer of the graph state $\ket{G}$ (see Equation~\eqref{eq:generators}). We conclude that we can extract $\ket*{GHZ_{M-1}}$ from $\ket{G}$.
	
	If such a set $W$ does not exist, let us show that we can extract $\ket*{GHZ_{M}}$ from $\ket{G}$. Indeed, then any operator in the stabilizer $S$ of $\ket{G}$, which contains $g_1$ in the decomposition w.r.t. the canonical generators acts nontrivial on all parties. Indeed, as all vertices in $\{2,\ldots, M-1\}$ are only connected to vertex $1$, such an operator acts nontrivially on parties $\alpha_1,\ldots ,\alpha_{M-1}$. It also acts nontrivial on party $\alpha_{M}$ as no set of vertices $W$ exists such that Equation~\eqref{eq_eq2} holds. As any set of generators of $S$ has to contain an operator whose decomposition w.r.t. the canonical generators contains $g_x$, we conclude that each set of generators for the stabilizer $S$ of $\ket{G}$ contains an element with support on all parties. Thus, we have that $S\neq S_{loc}$. Then Theorem 3 from Ref.~\cite{Br06} implies that we can extract $\ket*{GHZ_{M}}$ from $\ket{G}$.	
\end{proof}

Next, we use Lemma~\ref{lem_extract_ghz} to show Observation~\ref{thm_bounds_egs} which establishes bounds on the ratios of qubit numbers in the case of $4$ parties for which there possibly exist states in EGS. Besides the extraction of the states $\ket{GHZ_4}$ and $\ket{GHZ_3}$, the observation also takes into account the extraction of $\ket{\phi^+}$ and $\ket{0}$ states. For readability let us state the observation again.
\begin{manualobservation}{14}
	Let $P(4,n)=\{\alpha_1,\alpha_2,\alpha_3,\alpha_4\}$ be a partition such that $|\alpha_1|\ge|\alpha_2|\ge|\alpha_3|\ge|\alpha_4|$ and $\ket{\psi}\in\text{Stab}(P(M,n))$. Then $\ket{\psi}$ is decomposable if $|\alpha_1|>2|\alpha_4|$.
\end{manualobservation}

\begin{proof}
	Let us denote $V=V^{\mathcal{S}_\psi}$. We start by considering the reduced density matrix $\rho_{\alpha_1\cup\alpha_2}$ of $\ket{\psi}$. From Equation~\eqref{eq:reducedstate}
	it is evident that the rank of $\rho_{\alpha_1\cup\alpha_2}$ is equal to $2^{|\alpha_1|+|\alpha_2|-\text{dim}(V_{\alpha_1\cup\alpha_2})}$, where $\text{dim}(V_{\alpha_1\cup\alpha_2})$ is the dimension of the local subspace defined in Equation~\eqref{eq:local}. As ${\rm rk}(\rho_{\alpha_1\cup\alpha_2})={\rm rk}(\rho_{\alpha_3\cup\alpha_4})\le 2^{|\alpha_3|+|\alpha_4|}$ we find that $\text{dim}(V_{\alpha_1\cup\alpha_2})\ge |\alpha_1|+|\alpha_2|-|\alpha_3|-|\alpha_4|$. An analogous argument shows that $\text{dim}(V_{\alpha_1\cup\alpha_3})\ge |\alpha_1|+|\alpha_3|-|\alpha_2|-|\alpha_4|$. From the assumptions of the observation it follows that the sum of both these inequalities can be bounded from below as
	\begin{equation}
	\text{dim}(V_{\alpha_1\cup\alpha_2})+\text{dim}(V_{\alpha_1\cup\alpha_3})\ge2|\alpha_1|-2|\alpha_4|> |\alpha_1|.\label{ineq_1}
	\end{equation}
	Let us distinguish two cases. If there exists a $j\in[4]$ such that $\text{dim}(V_{\alpha_j})>0$, then it is easy to see from Theorem~\ref{thm_PLC_extract} that the state is decomposable, namely, one can extract $\ket{0}$ from the state. If $\text{dim}(V_{\alpha_j})=0$ for all $j\in [4]$, any vector in $V$ has support on at least two parties, i.e., the reduced state of every party has full rank. Assuming that, let us again distinguish two cases. First, if $\omega(f_{\alpha_1},g_{\alpha_1})=0$ for all $f,g\in V_{\alpha_1\cup\alpha_2}$ and $\omega(h_{\alpha_1},k_{\alpha_1})=0$ for all $h,k\in V_{\alpha_1\cup\alpha_3}$, then there exists $f\in V_{\alpha_1\cup\alpha_2}$ and $h\in V_{\alpha_1\cup\alpha_3}$ with $f,h\neq 0$ such that $f_{\alpha_1}=h_{\alpha_1}$. Indeed, due to the above assumptions the subspace $\left<f_{\alpha_1}|f\in V_{\alpha_1\cup\alpha_2}\cup V_{\alpha_1\cup\alpha_3}\right>$ is isotropic, and thus, its dimension is at most $|\alpha_1|$. Inequality~\eqref{ineq_1} implies that there are more independent vectors than that in $V_{\alpha_1\cup\alpha_2}\cup V_{\alpha_1\cup\alpha_3}$. Thus, there exist $f,h\in V_{\alpha_1\cup\alpha_2}\cup V_{\alpha_1\cup\alpha_3}$ such that $f\neq h$ and $f_{\alpha_1}=h_{\alpha_1}$. Note that it is not possible that either $f,h\in V_{\alpha_1\cup\alpha_2}$ or $f,h\in V_{\alpha_1\cup\alpha_3}$ as $\text{dim}(V_{\alpha_j})=0$ for all $j\in [4]$. We conclude that in case $\text{dim}(V_{\alpha_j})=0$ for all $j\in [4]$ and $\omega(f_{\alpha_1},g_{\alpha_1})=0$ for all $f,g\in V_{\alpha_1\cup\alpha_2}\cup V_{\alpha_1\cup\alpha_3}$ the assumptions of Lemma~\ref{lem_extract_ghz} are satisfied and $\ket{\psi}$ is decomposable.
	
	Second, if there exist $f,g\in V_{\alpha_1\cup\alpha_2}$ (or $f,g\in V_{\alpha_1\cup\alpha_3}$) such that $\omega(f_{\alpha_1},g_{\alpha_1})=1$, then also $\omega(f_{\alpha_2},g_{\alpha_2})=1$ ($\omega(f_{\alpha_3},g_{\alpha_3})=1$) as $V^{\mathcal{S}}$ is maximally isotropic. It is easy to see from Theorem \ref{thm_PLC_extract} that one can extract $\ket{\phi^+}$ from $\ket{\psi}$ and thus $\ket{\psi}$ is decomposable. 
\end{proof}

\section{Properties of the spiral graph state\label{app:spiralprop}}

In this section we are concerned with properties of the spiral graph states $\ket{G_n}$ of $n$ qubits distributed among $4$ parties as shown in Figure~\ref{fig_fig1}. First, we prove Theorem~\ref{thm_egs_infinite}, i.e., we show that $\ket{G_n}$ is indecomposable under PLC transformations for any $n$. Then, we prove that $\ket{G_{4n}}$ is indecomposable even under PLU. In fact, we show the more general statement that any PLU transformation from a state $\ket{G_{4n}}$ to another graph state for any $n$ can be implemented by a PLC transformation. This result implies that the EGS$_4$ is infinite for both, PLC and PLU transformations.

To show that $\ket{G_n}$ is indecomposable under PLC transformations, let us first prove the following simple lemma.

\begin{lem}\label{lem_indecomp_com}
	Let $M\in \mathbb{N}$ and let $(A_i\in M_{n\times n}(F))_{i\in [M]}$. If any tuple of linear combinations of the matrices ${(A_i)}_{i\in [M]}$ is indecomposable, then also $(A_i)_{i\in [M]}$ is indecomposable.
\end{lem}

\begin{proof}
	If $(A_i)_{i\in [M]}$ is decomposable, then also any tuple of linear combinations of the matrices ${(A_i)}_{i\in [M]}$ is decomposable. The lemma is the negated version of this statement.
\end{proof}

With the help of this lemma we can now show that $\ket{G_n}$ is indecomposable for any $n$.

\begin{proof}[Proof of Theorem \ref{thm_egs_infinite}]
	Let $\ket{G_n}$ be the spiral graph state distributed among $4$ parties with different labeling depending on whether $n$ is even or odd, see Figure~\ref{fig_numbering_qubtis}.
	
	\begin{figure}
		\centering\includegraphics[width=0.5\textwidth]{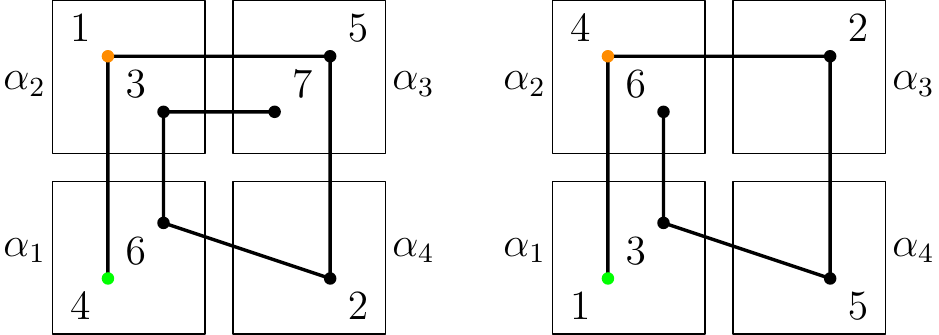}
		\caption{\textbf{Different ways to number qubits for $n$ even or odd:} if $n$ is odd, start from the orange qubit label every second qubit, then continue from green one; if $n$ is even, start the green qubit and label every second qubit, then continue from the orange one \label{fig_numbering_qubtis}}
	\end{figure}
	
	Let us start with the case $n$ even. We consider the commutation matrix $C_{\alpha_1\cup\alpha_2}$ w.r.t. the canonical generators. It is straightforward to verify that the only generators that anticommute on $\alpha_1\cup \alpha_2$ are the pairs $g_{2k},g_{n/2+2k+1}$ and $g_{2l+1},g_{n/2+2l+2}$ for $k,l\in\mathbb{N}$ (and such that $n/2+2k+1\le n$ and $n/2+2l+2\le n$). That is
	
	\begin{equation}
	C_{\alpha_1\cup\alpha_2}=\begin{pmatrix}
	0&M_2\\
	M_2^T&0
	\end{pmatrix},
	\end{equation}
	
	where $M_2$ is defined in Eq.~\eqref{eq_standard_form2}. Next, let us consider the commutation matrix $C_{\alpha_2\cup\alpha_3}$ again w.r.t. to the canonical generators. It is straightforward to verify that the only generators that anticommute on $\alpha_2\cup\alpha_3$ are the pairs $g_{2k+1},g_{n/2+2k+1}$ and $g_{2l+2},g_{n/2+2l+2}$ for $k,l\in\mathbb{N}_0$ (and such that $n/2+2k+1$ and $n/2+2l+2\le n$). That is
	
	\begin{equation}
	C_{\alpha_2\cup\alpha_3}=\begin{pmatrix}
	0&M_1\\
	M_1^T&0
	\end{pmatrix},
	\end{equation}
	
	where $M_1$ is defined in Eq.~\eqref{eq_standard_form2}. We conclude that the tuple $(C_{\alpha_1\cup\alpha_2},C_{\alpha_2\cup\alpha_3})$ is indecomposable and by Lemma~\ref{lem_indecomp_com} that also $(C_{\alpha_1},C_{\alpha_2},C_{\alpha_3})$ is indecomposable. This is equivalent to the statement that $\ket{G_n}$ is indecomposable for any even $n$. 
	
	Let us now consider the case $n$ odd. It is straightforward to verify that the commutation matrices $(C_{\alpha_1\cup\alpha_2},C_{\alpha_2\cup\alpha_3})$ are of the form presented in Eq.~\eqref{eq_standard_form4} and \eqref{eq_standard_form1}. Thus, by analogous arguments as above the state $\ket{G_n}$ is also indecomposable for any odd $n$.
\end{proof}

Note that the condition provided by Lemma \ref{lem_indecomp_com} is only a necessary, but not a sufficient condition for decomposability. This can be easily seen by considering for instance the 4 qutrits spiral graph state. There, the tuple $(C_{\alpha_1\cup\alpha_2},C_{\alpha_2\cup\alpha_3})$ is decomposable.

Next, we show that PLU transformations from $\ket{G_{4n}}$ to any other stabilizer state are possible if and only if they are possible via a PLC transformation. This implies that the EGS for more than $3$-partite stabilizer states is also infinite with respect to PLU transformations. Thus, the fact that the EGS with respect to PLC transformations is infinite is not just an artifact of considering a restricted set of operations. As we are only interested in the latter statement, it is sufficient to consider the states $\ket{G_{4n}}$.

\begin{manualtheorem}{13}
    Let $U$ be a PLU transformation such that $U\ket{G_{4n}}=\ket{H_{4n}}$ is a stabilizer state. Then, there exists a PLC transformation $C\in\mathcal{C}_{4n}^{P(4,4n)}$ such that $C\ket{G_{4n}}=\ket{H_{4n}}$.
\end{manualtheorem}

To prove this theorem we study the structure of the PLU symmetry group $\Sigma_{G_{4n}}$ of $\ket{G_{4n}}$, i.e.,
\begin{equation}
\Sigma_{G_{4n}}=\{U|U\text{ is PLU}, U\ket{G_{4n}}=\ket{G_{4n}}\}.
\end{equation}
Note that in particular $\mathcal{S}_{G_{4n}}\subset \Sigma_{G_{4n}}$. For the spiral graph states $\ket{G_{4n}}$ distributed among $4$ parties as specified above the elements of the stabilizer have very particular support. From now on we return to the numbering of qubits shown in Figure~\ref{fig_fig1}.
\begin{lem}\label{structure_stab}
	For any $n$ there is no element in $\mathcal{S}_{4n}$, which has support on only one party. Moreover, there are exactly two elements in $\mathcal{S}_{4n}$ which have support on only two parties, namely $g_1$ and $g_{4n}$.
\end{lem}

\begin{proof}

    The following arguments are based on the fact that if the decomposition of any operator $f$ with respect to the canonical generators (see Equation~\eqref{eq:generators}) contains $g_j$ for some $j\in[n]$, then $j\in\text{supp}(f)$.
	
	First let us argue that there is no element with support on only one party. Any canonical generator has support on at least two parties, and the canonical generators corresponding to the qubits of the same party all act on disjoint sets of qubits. We conclude that any element of $\mathcal{S}_{4n}$ has at least support on two parties.
	
	Next, we show that there are exactly two elements with support on two parties. Obviously, $g_1$ and $g_{4n}$ ($g_1$, $g_{4n}$ and $g_1g_{4n}$) have support on only two parties. We show that there are no further elements with this property. 
	
	Let us first show that such an element does not exist on any pair of adjacent parties, starting with $\alpha_1,\alpha_2$. Suppose there exists an element $f\in \mathcal{S}_{4n}$ such that $\text{supp}(f)=\{\alpha_1,\alpha_2\}$. For any $j\in \alpha_1$ with $j\neq 1,l$ the operator $g_j$ cannot be in the decomposition of $f$ with respect to the canonical generators. Indeed, $g_j$ acts non-trivially on qubit $j-1$ (party $D$) and no other generator corresponding to a qubit of $\alpha_1$ or $\alpha_2$ acts non-trivially on this qubit. Thus, if $g_j$ was contained in the decomposition, then $\text{supp}(f)=\{\alpha_1,\alpha_2,\alpha_4\}$. Similarly, no generator $g_j$ with $j\in \alpha_2$ and $j\neq 1,l$ can be contained in the decomposition of $f$. A similar argument for any other pair of adjacent parties shows that the only options for $f$ to have support on only two adjacent parties are $f=g_1$ and $f=g_{4n}$.
	
	Next, consider pairs of non-adjacent parties, i.e., $\alpha_1,\alpha_3$ and $\alpha_2,\alpha_4$. Suppose there exists an element $f\in \mathcal{S}_{4n}$ such that $\text{supp}(f)=\{\alpha_1,\alpha_3\}$ or $\text{supp}(f)=\{\alpha_2,\alpha_4\}$. Let $g_j$ for some $j$ be part of the decomposition of $f$ with respect to the canonical generators.
	
	If $j\neq 1, l$, then $g_j$ acts non-trivially on qubit $j-1$ and $j+1$ (both not part of the support of $f$) and in order to cancel those contributions the decomposition of $f$ also has to contain $g_{j+2}$ and $g_{j-2}$. Applying the same reasoning to $g_{j+2}$ and $g_{j-2}$ we find that the decomposition has to contain all canonical generators from elements in $\alpha_1,\alpha_3$ or $\alpha_2,\alpha_4$. Therefore $f$ acts non-trivially on both, qubits $1$ and $l$. However, as qubits $1$ and $l$ are not both contained in either $\alpha_1,\alpha_3$ or $\alpha_2,\alpha_4$ for any $n$, $f$ always has support on three parties which is a contradiction.	
	
	If $j= 1 (j=l)$, then $g_j$ has support on one adjacent party and, in order to cancel the contribution on this party, the decomposition also has to contain $g_{2}$ $(g_{l-2})$. Similarly, it has to contain $g_{4}$ $(g_{l-4})$ in order to cancel the contributions from $g_{2}$ $(g_{l-2})$ (has support on three parties). We find that the decomposition has to contain all canonical generators from elements in $\alpha_1,\alpha_3$ or $\alpha_2,\alpha_4$. Again $f$ would have support on three parties which is a contradiction.
	
	We conclude that there does not exist an $f$ with support on only $\alpha_1,\alpha_3$ or $\alpha_2,\alpha_4$. This completes the proof.
\end{proof}

For any $n$ the group $\Sigma_{G_{4n}}$ also contains symmetries associated to the beginning and end of the loop. These so called leaf symmetries are of the form
\begin{align}
B_{1,2}(\beta)&=e^{i\beta X}\otimes e^{-i\beta Z}\otimes \mathds{1},\\
\text{and  }B_{4n-1,4n}(\beta)&=\mathds{1}\otimes e^{-i\beta Z}\otimes e^{i\beta X},
\end{align}
where $\beta\in \mathbb{R}$ is arbitrary~\cite{Zh09}. While $\Sigma_{G_{4n}}$ might contain many more symmetries than those generated by the stabilizer and the leaf symmetries, we do not compute those symmetries. Instead show two lemmata that are concerned with their structure.

\begin{lem}\label{lem_cliff_sym}
	Let $C\in\Sigma_{G_{4n}}$ be a PLC. If there exists a party $\alpha\in P(M,n)$ such that $C_\alpha=\mathds{1}$, then either $C=S$ or $C=SB_{1,2}(\pm \pi/4)$  or $C=SB_{4n-1,4n}(\pm\pi/4)$ for some  $S\in\mathcal{S}_{G_{4n}}$.
\end{lem}

\begin{proof}
Let us fist outline the proof, before we present the technical details. 
We prove the lemma by distinguishing two cases, $C_{\alpha_1}=\mathds{1}$ and $C_{\alpha_2}=\mathds{1}$. Considering the first case, we show that $C$ has to leave $X_i$ invariant for all $i\in\{1, \ldots 4n-2, 4n\}$ and $Z_i$ invariant for all $i\in\{1,\ldots, 4n-1\}$. This implies that $C_i$ is a Pauli operator for all $i$ except $4n-1, 4n$. As there exists no Pauli operator other than the elements of $\mathcal{S}_{G_{4n}}$ as a symmetry of the stabilizer state and as $C g_{4n-1}C^\dagger$, which acts only non-trivially on $4n-2,4n-1,4n$, must be in $\mathcal{S}_{G_{4n}}$, we have due to the properties of $\mathcal{S}_{G_{4n}}$ that $C=SB_{4n-1,4n}(\pm\pi/4)$ or $C=S$ for some  $S\in\mathcal{S}_{G_{4n}}$. The second case can be treated analogously. 

Let us now present the details of this proof. As mentioned above, we distinguish two cases, $C_{\alpha_1}=\mathds{1}$ and $C_{\alpha_2}=\mathds{1}$. The other cases follow from the permutation symmetry of the spiral graph. In the following we denote by $\{g_i\}\subset \mathcal{S}_{G_{4n}}$ the canonical generators.

Suppose $C_{\alpha_1}=\mathds{1}$. First note that
$g_{k}Cg_{k}C^\dagger$ for $k=4j+1$ ($k=4j+2$) 
is (i) a Pauli operator, (ii) a symmetry of the state, and (iii) acts trivially on ${\alpha_1}, {\alpha_3}$ (${\alpha_2}, {\alpha_4}$), respectively. Using that (i) and (ii) implies that $g_{k}Cg_{k}C^\dagger \in \mathcal{S}_{G_{4n}}$ and then employing Lemma \ref{structure_stab} implies that $g_{k}Cg_{k}C^\dagger=g_{k}$, for $k=4j+1$ and $k=4j+2$.  

Using this result, we find that for any $j<n$ the operator $g_{4j+3} C g_{4j+3}C^\dagger$ is a Pauli operator and has support only on parties ${\alpha_3}$ and ${\alpha_4}$. Therefore, due to Lemma \ref{structure_stab}, $C g_{4j+3}C^\dagger=g_{4n}^{\delta_{4j+3}} g_{4j+3}$ where $\delta_{4j+3}\in\{0,1\}$. Similarly, for any $j\le n$ we find $C g_{4j}C^\dagger=g_{4n}^{\delta_{4j}} g_{4j}$. For $j=n$, we have $\delta_{4n}=0$ as $C\ldots C^\dagger$ has to preserve support. Now let us show that $\delta_{4j+3}=0$ for all $j\le n-2$ and $\delta_{4j}=0$ for all $j\le n$. Observe that $[g_{4(n-1)+3},g_{4j+3}]_{\alpha_3}=0$ and $[g_{4(n-1)+3},g_{4j+3}]_{\alpha_4}=0$ for all $j\le n-2$ and $[g_{4(n-1)+3},g_{4j}]_{\alpha_3}=0$ and $[g_{4(n-1)+3},g_{4j}]_{\alpha_4}=0$ for all $j\le n$. As conjugation by $C$ preserves these commutation relations, we find that $\delta_{4j+3}=0$ for all $j\le n-2$ and $\delta_{4j}=0$ for all $j\le n$. Thus, $C$ leaves all canonical generators invariant except for possibly $g_{4(n-1)+3}$. If $C g_{4(n-1)+3}C^\dagger=g_{4(n-1)+3}$, then $C\in\mathcal{S}_{G_{4n}}$. If $C g_{4(n-1)+3}C^\dagger=g_{4j}g_{4(n-1)+3}$, then $C_{4n-1}=X^\delta\exp(\pm i\pi/4 Z)$ and $C_{4n}=Z^\delta \exp(\mp i\pi/4 X)$ where $\delta \in\{0,1\}$ and $C_j\in \mathcal{P}$ for any $j<4n-1$. Thus, we can write $C=SB_{4n-1,4n}(\pm\pi/4)$. As $\exp(\pm i\pi/4 Z_{4n-1})\otimes \exp(\mp i\pi/4 X_{4n})\ket{G_{4n}}=\ket{G_{4n}}$, also $S\ket{G_{4n}}=\ket{G_{4n}}$ and therefore $S\in\mathcal{S}_{G_{4n}}$.
	
Suppose $C_{\alpha_2}=\mathds{1}$. Considering the symmetries $g_{4j+1}Cg_{4j+1}C^\dagger$ and $g_{4j+2}Cg_{4j+2}C^\dagger$ and using the same arguments as above one finds that $C_{\alpha_1}\in\mathcal{P}_{|\alpha_1|}$. The rest of the poof is completely identical. Thus, we have shown again that either $C=S$ or $C=SB_{4n-1,4n}(\pm\pi/4)$, where $S\in\mathcal{S}_{G_{4n}}$.
\end{proof}

\begin{lem}\label{lem_non_cliff_sym}
	For $n\ge 2$ let $U\in\Sigma_{G_{4n}}$ be a PLU symmetry which is not a Clifford operator. If there exists a party $\alpha\in P(M,n)$ such that $U_\alpha=\mathds{1}$, then either $U=\mathds{1}$ or $U=SB_{1,2}(\beta )$  or $U=SB_{4n-1,4n}(\beta)$ where $S\in\mathcal{S}_{G_{4n}}$ and $\beta\in\mathbb{R}$, $\beta\neq  k\pi/4$ with $k\in\mathbb{Z}$.
\end{lem}

\begin{proof}
	Let $U\in\Sigma_{G_{4n}}$ be a symmetry which is not a Clifford operator, i.e., in particular $U\neq \mathds{1}$. We consider the cases $U_{\alpha_1}=\mathds{1}$ and $U_{\alpha_2}=\mathds{1}$. The cases $U_{\alpha_3}=\mathds{1}$ and $U_{\alpha_4}=\mathds{1}$ follow from the permutation symmetry of the spiral graph.
	
Let us again outline the proof. We consider various operators of the form $W U W U^\dagger$, with $W \in \mathcal{S}_{G_{4n}}$ to first show (using Lemma \ref{lem_cliff_sym}) that $U_{\alpha_2}$ is a Pauli operator and that $U_{\alpha_3}$, $U_{\alpha_4}$ also have to be Pauli operators on all qubits but $4n-1,4n$. Similarly to the proof of Lemma \ref{lem_cliff_sym}, this implies the form given in the lemma. 	
	
Let us now present the details of the proof. Suppose that $U_{\alpha_1}=\mathds{1}$. First, we show that this assumption implies that $U_{\alpha_2}$ is a Clifford operator, $(UX_{4j+3}U^\dagger)_{\alpha_3} \in\mathcal{P}_{|{\alpha_3}|}$ ,$(UZ_{4j+4}U^\dagger)_{\alpha_4} \in\mathcal{P}_{|{\alpha_4}|}$ for all $j<n-1$ and $(UZ_{4j+3}U^\dagger)_{\alpha_3} \in\mathcal{P}_{|{\alpha_3}|}$, $(UX_{4j+4}U^\dagger)_{\alpha_4} \in\mathcal{P}_{|{\alpha_4}|}$ for all $j\le n-1$. Suppose any of the latter conditions does not hold. Then one of the operators 
$W U W U^\dagger$, with 

\begin{equation}
    W\in \{g_{4l+2},g_{4l+1}, g_{4l+2}g_{4l+4}, g_{4m+3}g_{4m+5}\}
\end{equation}

for $0\le l\le n-1$ and $0\le m\le n-2$, is not a Pauli operator and in particular not the identity. These operators are constructed such that they act nontrivial only on party ${\alpha_2}$ and ${\alpha_3}$ or ${\alpha_2}$ and ${\alpha_4}$. Let $V$ be one of the operators from above such that $V$ is not a Pauli operator and not the identity. Let us distinguish two cases. If $V$ is a Clifford operator, Lemma \ref{lem_cliff_sym} and Lemma \ref{structure_stab} imply that $V=\mathds{1}$ as $V$ has support only on parties ${\alpha_2}$ and ${\alpha_3}$ or ${\alpha_2}$ and ${\alpha_4}$. However, this is a contradiction to the assumption that $V$ is not the identity. If $V$ is not a Clifford operator, then let us again distinguish two cases. If $V$ has support on party ${\alpha_2}$ and ${\alpha_3}$, then one of the operators $W U W U^\dagger$, with 

\begin{equation}
W \in \{g_{4l+1},g_{4l+4}, g_{4l+2}g_{4l+4},g_{4l+3}g_{4l+1}\}
\end{equation}

for $0\le l\le n-1$, acts nontrivial on only one party and is a symmetry. This is a contradiction to the fact that for $\ket{G_{4n}}$ every party is fully entangled with the other ones. If $V$ has support on party ${\alpha_2}$ and ${\alpha_4}$, then one of the operators $W U W U^\dagger$ with 

\begin{equation}
W \in \{g_{4l+2},g_{4l+4}, \prod_{k\le l}g_{4k+1}\prod_{0\le k\le l-1}g_{4k+3}, \prod_{k\le l}g_{4k+1}\prod_{0\le k\le l}g_{4k+3}\}
\end{equation}

for $0\le l\le n-1$, acts nontrivial on only one party and is a symmetry, which leads again to the same contradiction. We conclude that $U_{\alpha_2}$ is a Clifford operator, $(UX_{4j+3}U^\dagger)_{\alpha_3} \in\mathcal{P}_{|{\alpha_3}|}$, $(UZ_{4j+4}U^\dagger)_{\alpha_4} \in\mathcal{P}_{|{\alpha_4}|}$ for all $j<n-1$ and $(UZ_{4j+3}U^\dagger)_{\alpha_3} \in\mathcal{P}_{|{\alpha_3}|}$, $(UX_{4j+4}U^\dagger)_{\alpha_4} \in\mathcal{P}_{|{\alpha_4}|}$ for all $j\le n-1$. Using these insights and following the exact steps of the proof of Lemma \ref{lem_cliff_sym}, we find that $U$ leaves all canonical generators invariant except $g_{4n-1}$. Expressing $U$ in the Pauli basis and imposing those conditions, we find that $U_{j}\in\mathcal{P}$ for any $j\le n-2$ and $U_{4n-1,4n}=P_1 \exp(i\beta_1 X)\otimes P_2 \exp(-i\beta_2 X)$ for some Pauli operators $P_1, P_2$ and with $\beta_{1},\beta_2\in\mathbb{R}$. We claim that $\beta_1=\beta_2+k\pi/2$ for some $k\in\mathbb{Z}$. Indeed, suppose $\beta_1\neq \beta_2+k\pi/2$. Let us choose a $S\in \mathcal{S}_{G_{4n}}$ such that $S_{4n}=P_2$. Then the symmetry $(SUB_{4n-1,4n}(-\beta_1))^2\in\Sigma_{G_{4n}}$ acts nontrivial only on qubit $4n$. This is a contradiction to the fact that this qubit is fully entangled with the rest of the state. Thus, we have shown that $U_{4n-1,4n}=\tilde{P}_1 \exp(i\beta X)\otimes \tilde{P}_2 \exp(-i\beta X)$ with $\beta\in\mathbb{R}$, $\tilde{P}_1,\tilde{P}_2\in\mathcal{P}$. As we assumed that $U$ is not a Clifford operator, $\beta\neq k\pi/4$ for any $k\in\mathbb{Z}$. Finally, we can write $U=SB_{4n-1,4n}(\beta)$ for some $S\in\mathcal{P}_{4n}$. As $B_{4n-1,4n}(\beta)\in\Sigma_{G_{4n}}$, we conclude $S\in\mathcal{S}_{G_{4n}}$.
	
The case $U_{\alpha_2}=\mathds{1}$ can be shown similarly considering different operators for $W$.
\end{proof}

Thus, with the above lemmata we have shown, that any PLU symmetry of $\ket{G_{4n}}$, which acts trivially on one party, is a product of elements of the stabilizer $\mathcal{S}_{G_{4n}}$ and leaf symmetries. Let us use these insights to show Theorem \ref{thm_plu_plc}.

\begin{proof}[Proof of Theorem \ref{thm_plu_plc}]
	Suppose there exists a PLU operation $U$ such that $\ket{H_{4n}}=U\ket{G_{4n}}$ is a graph state. We want to show that $\ket{G_{4n}}$ and $\ket{H_{4n}}$ are PLC equivalent.

	First, observe that $U^\dagger \Sigma_{H_{4n}}U=\Sigma_{G_{4n}}$ and therefore $U^\dagger \mathcal{S}_{H_{4n}}U\subset \Sigma_{G_{4n}}$. We show that not all elements of $\Sigma_{G_{4n}}$ can be a preimage of elements of $\mathcal{S}_{H_{4n}}$. The canonical generators of $\mathcal{S}_{G_{4n}}$ all have support on at most three parties. Therefore, it follows from the results of Ref.~\cite{Br06} that one cannot extract a $4$-partite GHZ state from $\ket{G_{4n}}$ (see Section~\ref{sec_previous_results}) using PLU operations. As $\ket{H_{4n}}$ is PLU equivalent to $\ket{G_{4n}}$, one also cannot extract a GHZ state from $\ket{H_{4n}}$ using PLU. Thus, $\mathcal{S}_{H_{4n}}$ admits a set of generators $\{h_i\}_{i\in[4n]}$ such that no generator $h_i$ has support on more than $3$ parties. Consider the preimages $\{U^\dagger h_i U\}_{i\in [4n]}$, which are PLU symmetries of the state $\ket{G_{4n}}$. It follows form Lemma \ref{lem_cliff_sym} and Lemma \ref{lem_non_cliff_sym} that $U^\dagger h_i U=S_i$ or $U^\dagger h_i U=S_iB_{1,2}(\beta_i^1 )$  or $U^\dagger h_i U=S_iB_{4n-1,4n}(\beta_i^2)$ where $S_i\in\mathcal{S}_{G_{4n}}$ and $\beta_i^1,\beta_i^2\in\mathbb{R}$ with $\beta_i^1\neq k\pi/4$, $\beta_i^2\neq k\pi/4$ for any $k\in\mathbb{Z}$ for all $i\in[4n]$.	
	
	If $\beta_i^1,\beta_i^2=0$ for all $i\in[4n]$, then $\{U^\dagger h_i U\}$ are Pauli operators. They are independent as the operators $\{h_i\}$ are independent. Moreover, the operators $\{U^\dagger h_i U\}$ have the same support and commutation relations on parties as the operators $\{h_i\}$ (as these properties are preserved by $U$). Therefore, it follows from Theorem \ref{thm_PLC_equiv} that the transformation from $\ket{G_{4n}}$ to $\ket{G_{4n}'}$ can be implemented by a PLC operator.
	
	Next, we consider the case where there exists an $i\in[4n]$ such that $\beta_i^1\neq 0$ and  $\beta_j^2= 0$ for all $j\in[4n]$. Let us show that this implies that for all $j$ such that $\beta_j^1\neq 0$ it holds that $\beta_j^1=\beta_i^1+l_j^1\pi/2$ for $l_j^1\in \mathbb{Z}$. To this end, observe that $h_j^2=\mathds{1}$ for all $j\in[4n]$ implies that $(U^\dagger h_i U)^2=\mathds{1}$ for all $j\in[4n]$. Thus, for any $j$ such that $\beta_j^1\neq 0$ we have
	\begin{equation}
	(U^\dagger h_i U)|_{1,2}=P_1^ie^{i\beta_j^1 X}\otimes P_2^ie^{-i\beta_j^1 Z},\label{eq_form_anticom}
	\end{equation}
	where $P_1^i,P_2^i\in\mathcal{P}$ are such that $P_1^iX=-XP_1^i$ and $P_2^iZ=-ZP_2^i$. Let us combine this structure with the fact that $(h_jh_k)_\alpha=\pm (h_kh_j)_\alpha$ for all $\alpha\in \{{\alpha_1},{\alpha_2},{\alpha_3},{\alpha_4}\}$ and $j,k\in [4n]$ implies that also $(U^\dagger h_jUU^\dagger h_k U)_\alpha=\pm (U^\dagger h_kUU^\dagger h_jU)_\alpha$. We obtain that if $j,k\in[4n]$ with $j\neq k$ are such that $\beta_j^1\neq 0$ and $\beta_k^1\neq 0$, then the latter condition only holds if
	\begin{align}
	e^{i(\beta_k^1-\beta_j^1) X}&=\pm e^{i(-\beta_k^1+\beta_j^1) X},\label{eq_loc10}
	\end{align}
	where we also used that $(U^\dagger h_jU)_q,(U^\dagger h_kU)_q\in\mathcal{P}$ for any $q\in {\alpha_1}\cup {\alpha_2}$ with $q\neq 1,2$ (due to Lemma~\ref{lem_non_cliff_sym}). Equation \eqref{eq_loc10} holds iff $\beta_k^1-\beta_j^1=l_j^1\pi/2$ for some $l_j^1\in\mathbb{Z}$. This is precisely the claim from above. It follows form this result that for any $j$ for which $l_j^1\neq 0$ we can write $U^\dagger h_j U=S_jB_{1,2}(\beta_j^1)=\tilde{S}_jB_{1,2}(\beta_j^1)$, where $\tilde{S}_j\in\mathcal{S}_{G_{4n}}$ is a different element of the stabilizer. To keep notations simple let us also define $\tilde{S}_j=S_j$ for all $j$ for which $l_j^1=0$. We show now that the operators $\{\tilde{S}_j\}$ are independent. Indeed, suppose they are not, i.e., suppose that there exists an index set $I\subset[4n]$ such that $\prod_{j\in I}S_j=\mathds{1}$. As operators $\{U^\dagger h_i U\}$ are independent (as the operators $\{h_i\}$ are independent), we have that  $\prod_{j\in I}U^\dagger h_j U\neq \mathds{1}$. One verifies that the operator $\prod_{j\in I}U^\dagger h_j U$ can only act nontrivially on qubits $1$ and $2$. To understand better how $\prod_{j\in I}U^\dagger h_j U$ acts on those two qubits consider the following. For any $j$ such that $\beta_j^1=0$ the fact that $(U^\dagger h_jUU^\dagger h_k U)_\alpha=\pm (U^\dagger h_kUU^\dagger h_jU)_\alpha$ holds for all $k$ and $\alpha\in \{\alpha_1,\alpha_2,\alpha_3,\alpha_4\}$ implies that $(\tilde{S}_j)_1\in\{\mathds{1},X\}$ and $(\tilde{S}_j)_2\in\{\mathds{1},Z\}$. Combining this with Equation \eqref{eq_form_anticom} we conclude $(\prod_{j\in I}U^\dagger h_i U)_{1,2} = \exp(i\beta_i^1 X)\otimes  \exp(-i\beta_i^1Z)$ which is a contradiction to the fact that $(\prod_{j\in I}U^\dagger h_i U)^2=\mathds{1}$. Thus, we have shown that the operators $\{\tilde{S}_j\}$ are independent, i.e., $\{\tilde{S}_j\}$ generate $\mathcal{S}_{G_{4n}}$. Moreover, observe that the operators $\tilde{S}_j$ have the same support as the corresponding ones in $h_j$. As $(U^\dagger h_jUU^\dagger h_k U)_\alpha=\pm (U^\dagger h_kUU^\dagger h_jU)_\alpha$ holds for all $j,k$ and $\alpha\in\{\alpha_1,\alpha_2,\alpha_3,\alpha_4\}$, it follows from the above considerations that the leaf part of the symmetries $\{U^\dagger h_i U\}$ does not affect commutation relations. We conclude with Theorem \ref{thm_PLC_equiv} that there exists a PLC transformation $C$ such that $C\tilde S_j C^\dagger =h_j$. Thus, $\ket{G_{4n}}$ and $\ket{H_{4n}}$ are PLC equivalent.
	
	The case where there exists $i$ such that $\beta_i^2\neq 0$ and $\beta_j^1=0$ for all $j$ follows from analogous arguments. Finally, consider the case where there exist $i, j\in [4n]$ such that $\beta_i^1\neq 0$ and $\beta_j^2\neq 0$. Due to Lemma \ref{lem_non_cliff_sym} none of the symmetries $\{U^\dagger h_i U\}$ is such that $\beta_i^1\neq 0$ and $\beta_i^2\neq 0$. Thus, we can again apply the same arguments as above to show the claim.
\end{proof}

\section{PLC transformations of qudit stabilizer states and commutation matrices\label{sec_proof_plcequiv_qudits}}

In this appendix we prove that also for qudit stabilizer states of prime dimensions the commutation matrix formalism characterizes PLC equivalence between stabilizer states. To this end we prove Theorem~\ref{thm_1_qudits}, which is a generalization of Theorem~\ref{thm_PLC_equiv_com} to qudit stabilizer states of prime dimensions. Let us first restate the Theorem.
\begin{manualtheorem}{15}
	Let $\ket{\psi}$ and $\ket{\phi}$ be $n$-qudit, $M$-partite stabilizer states. Then $\ket{\psi}$ is PLC equivalent to $\ket{\phi}$ if and only if they admit the same tuple of commutation matrices.
\end{manualtheorem}
The proof of Theorem~\ref{thm_1_qudits} follows similar lines as the proof of Theorem~$2$ in Ref.~\cite{Br06}. There, the proof is based on the following Lemma, which can be straightforwardly generalized to prime dimensions. For completeness we also include the proof here.
\begin{lem}\label{lem_condition_clifford}
	Let $V,W\subset\mathbb{Z}_d^{2n}$ be subspaces of the symplectic vector space $(\mathbb{Z}_d^{2n},\omega )$, and let $\text{dim}(V)=\text{dim}(W)=p$. Then the following statements are equivalent:
	\begin{itemize}
		\item For every party $\alpha\in P(M,n)$ there exists an isometry $u_\alpha:\mathbb{Z}_d^{2|\alpha|}\rightarrow \mathbb{Z}_d^{2|\alpha| }$ such that
		\begin{equation}
		W=\left(\bigoplus_{\alpha\in P(M,n)}u_\alpha \right)V
		\end{equation}
		\item There exists an invertible linear map $T:V\rightarrow W$ such that
		\begin{itemize}
			\item[(i)] $\omega(T(\vec{f})_\alpha,T(\vec{g})_\alpha)=\omega(\vec{f}_\alpha,\vec{g}_\alpha)$ for all $\vec{f},\vec{g}\in V$ and $\alpha\in P(M,n)$,
			\item[(ii)] $T(V_{\hat{\alpha}})=W_{\hat{\alpha}}$ for all $\alpha\in P(M,n)$.
		\end{itemize}
	\end{itemize}
\end{lem}
\begin{proof}
	The proof is analogous to the one presented in Ref.~\cite{Br06}. The only if part is trivial as $T=\oplus_{\alpha\in P(M,n)}u_\alpha$. Consider the if part of the statement. Let $\{\vec{f}_j\}_{j\in [p]}$ be a basis of $V$ and let $\vec{f}_j':= T(\vec{f}_j)$. As $T$ is an invertible map, $\{\vec{f}_j'\}_{j\in [p]}$ is a basis of $W$. For every party $\alpha\in P(M,n)$ let us define the linear map $v_\alpha: V_\alpha\rightarrow W_\alpha$ by $v_\alpha ((\vec{f}_j)_\alpha)=(\vec{f}_j')_\alpha$. It is linear as $T$ is linear. We now show that $v_\alpha$ is an isometry.
	
	First note that condition $(i)$ implies that $v_\alpha$ preserves the symplectic form $\omega$. Let us show that $v_\alpha$ is invertible due to condition $(ii)$. To that end we consider the equation $v_\alpha[\mathbf{x}]=0$ and show that this implies $\mathbf{x}=0$. Writing $\mathbf{x}=\sum_{j\in [p]}x_j(\vec{f}_j)_\alpha$ one obtains $v_\alpha[\sum_{j\in [p]}x_j(\vec{f}_j)_\alpha]=\sum_{j\in [p]}x_jv_\alpha[(\vec{f}_j)_\alpha]=\sum_{j\in [p]}x_j(\vec{f}'_j)_\alpha=0$. The definition of the colocal subspace implies that $\sum_{j\in [p]}x_j(\vec{f}'_j) \in W_{\hat{\alpha}}$. Due to condition (ii) this holds only if there exists $\sum_{j\in [p]}x_j(\vec{f}_j)\in V_{\hat{\alpha}}$ for which $T(\sum_{j\in [p]}x_j(\vec{f}_j))=\sum_{j\in [p]}x_j(\vec{f}'_j)$. Employing again the definition of the colocal subspace one realizes that $\sum_{j\in [p]}x_j(\vec{f}_j)\in V_{\hat{\alpha}}$ only if $\sum_{j\in [p]}x_j(\vec{f}_j)_\alpha =\mathbf{x}=0$, and thus, $v_\alpha $ is invertible. Hence, $v_\alpha: V_\alpha\rightarrow W_\alpha$ is an isometry and by Witt's Lemma (Lemma~\ref{lem_witt_lemma}) $v_\alpha$ extends to an isometry $u_\alpha:\mathbb{Z}_d^{2|\alpha|}\rightarrow \mathbb{Z}_d^{2|\alpha|}$. This concludes the proof.
	
\end{proof}

Following the same arguments as Ref.~\cite{Br06}, we show, in order to prove Theorem~\ref{thm_1_qudits}, that for maximally isotropic subspaces condition $(ii)$ is redundant. Let us first show the statement for stabilizer states with full local ranks using the same argument as Ref.~\cite{Br06}. Recall that a stabilizer state $\ket{\psi}$ is said to have full local ranks if $\rho_\alpha$ has full rank for every party $\alpha$ or equivalently if $\text{dim}(V_\alpha^\mathcal{S})=0$.

If $\text{dim}(V_\alpha^\mathcal{S})=d_\alpha> 0$, then there exists a Clifford operator $C\in\mathcal{C}_{|\alpha|}$ acting on party $\alpha $ such that $C\otimes\mathds{1}\ket{\psi}=\ket{\phi}\otimes\ket{0}^{\otimes d_\alpha}$. This follows directly from Witt's Lemma (Lemma~\ref{lem_witt_lemma}). Indeed, let $\{e_i\}$ be the standard basis of $\mathbb{Z}_d^{2|\alpha|}$ and let ${\vec{f}_i}$ be a basis for $V_\alpha^\mathcal{S}$. A bijection between $\{e_i\}_{i=2,\ldots, 2|\alpha|}$ and ${\vec{f}_i}$ preserves the symplectic form $\omega:\mathbb{Z}_d^{2|\alpha|}\times \mathbb{Z}_d^{2|\alpha|}\rightarrow \mathbb{Z}_d$. Thus, it extends to a isometry on $\mathbb{Z}_d^{2|\alpha|}$ by Witt's Lemma (Lemma~\ref{lem_witt_lemma}). The Clifford operator $C$ which corresponds to this isometry then gives $C\otimes\mathds{1}\ket{\psi}=\ket{\phi}\otimes\ket{0}^{\otimes d_\alpha}$.

\begin{lem}\label{lem_plc_equiv_qudit_full}
	Let $\ket{\phi},\ket{\psi}\in \text{Stab}(P(M,n))$ have full local ranks. Then, $\ket{\psi}$ and $\ket{\phi}$ is PLC equivalent to $\ket{\phi}$ iff they admit the same tuple of commutation matrices.
\end{lem}

\begin{proof}
	The only if part of the statement follows trivially from the fact that PLC does not alter the commutation matrices. Let us show the if part. Suppose that $\ket{\psi}$ and $\ket{\phi}$ admit the same set of commutation matrices. Using $V=V^{\mathcal{S}_\psi}$ and $W=V^{\mathcal{S}_\phi}$, it is straightforward to verify that then there exists a map $T:V\rightarrow W$ such that condition $(i)$ of Lemma~\ref{lem_condition_clifford} holds. Analogously to Ref.~\cite{Br06} let us show that this implies that condition $(ii)$ also holds as $V,W$ are maximally isotropic with full local ranks. Choose some party $\alpha\in P(M,n)$. Let $f\in V_{\hat{\alpha}}$. By the definition of the colocal subspace $f_\alpha=0$. Then, due to property $(i)$, it follows that $0=\omega(f_\alpha,[T^{-1}(g)]_\alpha)=\omega(T(f)_\alpha,g_\alpha)$ for all $g\in W$. Hence, $T(f)_\alpha\oplus 0\in W^\perp$ and since $W$ is maximally isotropic $T(f)_\alpha\oplus 0\in W$. As $\ket{\phi}$ has full local ranks, i.e., the local subspace $W_\alpha$ only contains the zero vector, we conclude that $T(f)_\alpha=0$ and therefore $T(V_{\hat{\alpha}})\subseteq W_{\hat{\alpha}}$. A similar argument with $T^{-1}$ shows that $T(W_{\hat{\alpha}})\subseteq V_{\hat{\alpha}}$ and therefore $T(V_{\hat{\alpha}})= W_{\hat{\alpha}}$.
\end{proof}

\begin{proof}[Proof of Theorem~\ref{thm_1_qudits}]
	The only if part of the statement follows trivially from the fact that PLC does not alter the commutation matrices. Let us show the if part. Suppose that $\ket{\psi}$ and $\ket{\phi}$ admit the same set of commutation matrices. Let us find PLC transformations $P_1$ and $P_2$ such that $P_1\ket{\psi}=\ket{\psi'}\otimes \ket{0}^{\otimes k}$ and $P_2\ket{\phi}=\ket{\phi'}\otimes\ket{0}^{\otimes k}$ and $\ket{\psi'}$ and $\ket{\phi'}$ are stabilizer states with full local ranks. Note that we can extract the same number of zeros in each party from $\ket{\psi}$ and $\ket{\phi}$ as they admit the same tuple of commutation matrices. We know that all nontrivial commutation relations of elements of the stabilizer of $\ket{\psi'}\otimes \ket{0}^{\otimes k}$ come from elements of the stabilizer of $\ket{\psi'}$ and analogously for $\ket{\phi'}\otimes\ket{0}^{\otimes k}$. As $\ket{\psi'}\otimes\ket{0}^{\otimes k}$ and $\ket{\phi'}\otimes\ket{0}^{\otimes k}$ admit the same tuple of commutation matrices (as PLC transformations do not alter commutation relations), we conclude that also $\ket{\psi'}$ and $\ket{\phi'}$ admit the same tuple of commutation matrices. Thus, according to Lemma \ref{lem_plc_equiv_qudit_full} they are PLC equivalent. This concludes the proof of the if part.
\end{proof}
From the proof it also follows that two states $\ket{\psi}$ and $\ket{\varphi}$ are PLC equivalent if and only if the states $\ket{\psi}\otimes \ket{0}^{\otimes k}$ and $\ket{\varphi}\otimes \ket{0}^{\otimes k}$ are PLC equivalent (where the $\ket{0}$ states are distributed in the same way).

\section{$m$-partite qubit stabilizer states and the uniqueness of the decomposition}\label{sec_uniqueness}

For qudit systems of prime dimension with $d>2$ we showed in Section~\ref{sec_qudit_systems} that the decomposition of stabilizer states into indecomposable states is unique. In this section we present a possible way to establish the uniqueness of decompositions into indecomposables in the case of qubits under the condition that for commutation matrices isomorphy and congruence transformations are equally powerful. While this condition is not true for arbitrary tuples of alternating matrices of length larger or equal to four, as we demonstrate below, for certain special tuples of commutation matrices, namely those corresponding to states without additional PLC symmetries, isomorphy and congruence are indeed equally powerful. Hence, decompositions into these stabilizer states are unique. For states admitting additional PLC symmetries the problem remains open.

Before discussing our approach to establish the uniqueness of the decompositions for qubit stabilizer states, a remark is in order.

First, let us discuss as to why the results of Scharlau regarding the decompositions of tuples of matrices~\cite{Sc76} cannot be extended to an arbitrary length of tuples. A crucial ingredient of the arguments in Ref.~\cite{Sc76} is to show that every pair of alternating matrices is congruent to a neutral one (see Equation~\eqref{eq_neutral}). To our knowledge, it is not known whether tuples of length $3$ still have this property. In particular, it is easy to see that for tuples of length $\ge 4$ the statement is no longer valid. Indeed, one can verify that the commutation matrices corresponding to the $5$ qubit ring graph distributed among $5$ parties with one qubit per party are not congruent to a neutral tuple. Let us note here, that Ref.~\cite{Sc76} uses the neutrality to show that congruence transformations ($Q\ldots Q^T$, $Q$ invertible) between pairs of alternating matrices are possible if and only if there exists an isomorphism ($Q\ldots P$, $Q$ and $P$ invertible) between the respective tuples. However, Ref.~\cite{Sc76} does not establish an equivalence between neutrality and the claim that isomorphy and congruence of tuples of alternating matrices. Thus, this does not prevent us from trying to prove the latter claim without relying on neutrality. 

Moreover, for general tuples of alternating matrices of length $\ge 4$ isomorphy and congruence transformations are not necessarily equally powerful, as the following example demonstrates. Consider the tuple $(A_1,A_2,A_3,A_4)$, where
\begin{align}
A_1&=\begin{pmatrix}0&1&0&0&0&0&0&0\\1&0&0&0&0&0&0&0\\0&0&0&0&0&0&0&0\\0&0&0&0&1&0&0&0\\0&0&0&1&0&1&0&0\\0&0&0&0&1&0&0&0\\0&0&0&0&0&0&0&0\\0&0&0&0&0&0&0&0\end{pmatrix}\ 
A_2=\begin{pmatrix}0&1&0&0&0&0&0&0\\1&0&1&0&0&0&0&0\\0&1&0&0&0&0&0&0\\0&0&0&0&0&0&0&0\\0&0&0&0&0&1&0&0\\0&0&0&0&1&0&1&0\\0&0&0&0&0&1&0&0\\0&0&0&0&0&0&0&0\end{pmatrix}\ \notag
\end{align}
\begin{align}
A_3&=\begin{pmatrix}0&0&0&0&0&0&0&0\\0&0&1&0&0&0&0&0\\0&1&0&1&0&0&0&0\\0&0&1&0&0&0&0&0\\0&0&0&0&0&0&0&0\\0&0&0&0&0&0&1&0\\0&0&0&0&0&1&0&1\\0&0&0&0&0&0&1&0\end{pmatrix}\ 
A_4=\begin{pmatrix}0&0&0&0&0&0&1&0\\0&0&0&0&0&0&0&0\\0&0&0&0&1&0&0&0\\0&0&0&0&0&0&0&0\\0&0&1&0&0&0&0&0\\0&0&0&0&0&0&0&0\\1&0&0&0&0&0&0&0\\0&0&0&0&0&0&0&0\end{pmatrix}.
\end{align}
Let us define $B_i = PA_i$ for all $i\in[4]$, where
\begin{equation}
P=\begin{pmatrix}1&0&0&0&0&0&0&0\\0&1&0&0&0&1&0&0\\0&0&1&0&0&0&0&0\\0&0&0&1&0&0&0&1\\1&0&0&0&1&0&0&0\\0&0&0&0&0&1&0&0\\0&0&1&0&0&0&1&0\\0&0&0&0&0&0&0&1\end{pmatrix},
\end{equation}
so that the $B_i$'s are alternating. One can then verify, e.g., by using the techniques described in Section~\ref{sec_egs4} (see Equation~\eqref{eq_check_congruence}), that there does not exist an invertible matrix $Q$ such that $QA_iQ^T=B_i$ for all $i\in[4]$.

We emphasize, that this example does not imply that isomorphy and congruence transformations are not equally powerful for commutation matrices as the tuple above does not satisfy the rank condition in Equation~\eqref{eq_rank_condition}, which is a necessary condition for commutation matrices. Note also that it is not clear whether the above example can be used to find an example of a tuple of alternating matrices which does not decompose uniquely into blocks.

Let us now outline a possible path to prove that decompositions into indecomposable qubit stabilizer states are unique. Let us start by considering only those tuples of commutation matrices for which isomorphy and congruence transformations are equally powerful, i.e., which satisfy the following condition.
\begin{prem}\label{clm_cong_iso}
Let $(A_i)_i$ and $(B_i)_i$ be tuples of commutation matrices such that there exists an invertible matrix $Q$ with $QA_iQ^T=B_i$ for all $i\in[n]$ if and only if there exist invertible matrices $P,R$ such that $PA_iR=B_i$ for all $i\in[n]$.
\end{prem}
We will later prove that this is indeed the case for commutation matrices corresponding to stabilizer states without PLC symmetries (of order $2$ or $4$). It is unclear if the condition holds for the general case.

Now, let us show the relevance of Condition~\ref{clm_cong_iso} in the context of proving the uniqueness of the decomposition of tuples of commutation matrices. We argue in the following that Condition~\ref{clm_cong_iso} together with the Krull-Remak-Schmidt theorem (Theorem~\ref{thm_krm}) and an observation from Ref.~\cite{Se88} are sufficient to show that commutation matrices decompose uniquely into indecomposable blocks. First, let us present the observation from Ref.~\cite{Se88}, which is concerned with decompositions of tuples of alternating matrices over fields with characteristic two.
\begin{lem}[\cite{Se88}]\label{lem_decom_char2}
    Let $(A_i\in M_{k\times k}(\mathbb{Z}_2))_{i\in[n]}$, with $k\in\mathbb{N}$, be a tuple of alternating matrices. If $(A_i)_i$ is indecomposable under congruence, then either $(A_i)_i$ is indecomposable under isomorphy or $(A_i)_i$ is isomorphic to a tuple $(B_i\oplus B_i^T)_i$, where $(B_i)_i$ is indecomposable.
\end{lem}
That is, a tuple of alternating matrices which is indecomposable under congruence, can at most be decomposed into two indecomposable blocks under isomorphy. Those blocks are equal up to transposition. If Condition~\ref{clm_cong_iso} is true, then the block $(B_i)_i$ is not isomorphic to a tuple of alternating matrices. Indeed, suppose $(B_i)_i$ is isomorphic to $(\tilde{B}_i)_i$, where $\tilde{B}_i$ is alternating for all $i\in [n]$. Then, $(A_i)_i$ is isomorphic to the tuple of alternating matrices $(\tilde{B_i}\oplus \tilde{B}_i^T)_i$, and therefore by Condition~\ref{clm_cong_iso} congruent to it. This is a contradiction to the assumption that $(A_i)_i$ is indecomposable under congruence.

It is now straightforward to see that in combination with the above result and the Condition~\ref{clm_cong_iso} the Krull-Remak-Schmidt theorem implies that commutation matrices over the field $\mathbb{Z}_2$ decompose uniquely into indecomposable blocks under congruence. Indeed, suppose $D_1=(A_i^1\oplus\ldots \oplus A_i^n)_i$ and $D_2=(B_i^1\oplus\ldots \oplus B_i^m)_i$ are decompositions into indecomposable blocks under congruence of the same tuple of commutation matrices. At most any of these blocks decomposes into two blocks under isomorphy. After this decomposition, the Krull-Remak-Schmidt theorem states that the two decompositions have the same number of blocks and, up to reordering, the blocks are isomorphic to each other. The above considerations imply that it cannot be the case that one of the new blocks in $D_1$ is isomorphic to one in $(B_i^1\oplus\ldots \oplus B_i^m)_i$. We conclude that $n=m$ and pairs of blocks in $D_1$ and $D_2$ are isomorphic. As all of these blocks are tuples of commutation matrices, Condition~\ref{clm_cong_iso} implies that these blocks are congruent. We see that if Condition~\ref{clm_cong_iso} is true, then tuples of commutation matrices decompose uniquely into indecomposable blocks under congruence. 

Let us now comment on our attempts to prove that Condition~\ref{clm_cong_iso} holds for all commutation matrices. Observe that one direction of the claim is trivial, namely that congruence implies isomorphy. From now on we only consider the other direction, i.e., we want to show that if there exist invertible matrices $P,R$ such that $(PA_iR)_i$ is a tuple of commutation matrices, then there exists an invertible matrix $Q$ such that $(PA_iR)_i=(QA_iQ^T)_i$. By applying the congruence transformation $R^{-1}\ldots (R)^{-T}$ this claim reduces to the following one. If there exists an invertible matrix $W=R^{-1}P$ such that $(WA_i)_i$ is a tuple of commutation matrices, then there exists an invertible matrix $H$ such that $(WA_i)_i=(HA_iH^T)_i$.

Let us investigate which operators $W$ are such that $(WA_i)_i$ is again a tuple of commutation matrices, given that $(A_i)_i$ is a tuple of commutation matrices. First recall that commutation matrices are symmetric, implying that these operators have to satisfy $(WA_i)_i=(A_iW^T)_i$. Therefore, $W$ is a self-adjoint endomorphism of $(A_i)_i$. As $W$ is invertible and defined over the finite field $\mathbb{Z}_2$, there exists a power $k\in\mathbb{N}$ such that $W^k=\mathds{1}$. Let us combine this property with the fact that $W$ is a self-adjoint endomorphism. If $k=2l+1$ is odd, then one finds that $(W^{2l+2}A_i)_i=(W^{l+1}A_i(W^{l+1})^T)_i$, i.e., the transformation is possible via congruence. Thus the only case in which $(WA_i)_i$ might not be congruent to the tuple $(A_i)$  is if $k=2l$ is even. In that case we find $(P^{2l}A_i)_i=(P^{l}A_i(P^{l})^T)_i=(A_i)_i$, i.e., the stabilizer state corresponding to $(A_i)_i$ has a PLC symmetry. This can be easily understood from the fact that any congruence transformation corresponds to a PLC transformation and vice versa (see Theorem~\ref{thm_PLC_equiv_com}). We conclude that only commutation matrices associated with stabilizer states with non-trivial PLC symmetries can possibly transform differently under isomorphy. Hence, for stabilizer states without additional PLC symmetries the above arguments imply that they decompose uniquely into indecomposable ones. However, it is questionable how useful this insight is. It is straightforward to verify that every indecomposable stabilizer state up to $10$ qubits has additional symmetries (see Figure~\ref{fig:EGS8}). Note that LC symmetries of stabilizer states have been studied in \cite{Bo91,Bo93,Zh09,En20}.

Summarizing the above discussion, if Condition~\ref{clm_cong_iso} is fulfilled the decompositions into indecomposables is unique. Moreover, we have shown that Condition~\ref{clm_cong_iso} is indeed fulfilled unless there exist states with PLC symmetries for which not all isomorphy transformations can be represented by congruence transformations.

Finally, let us remark that in Ref.~\cite{En20} we analyzed the local symmetries of stabilizer states (where each party only holds one qubit). While working on the results presented here, we became aware of Ref.~\cite{Bo91} and \cite{Bo93}, which study equivalence of graphs under local complementation. In this context, Ref.~\cite{Bo91} also derives the local Clifford symmetries of graph states.

\end{document}